%% file: main-arxiv.tex
\documentclass[a4paper,11pt]{article}
\usepackage[margin=2.5cm]{geometry}
\usepackage[utf8]{inputenc}
\usepackage[T1]{fontenc}
\usepackage{pdflscape}
\usepackage{enumitem}
\usepackage{hyperref}
\usepackage{tikz}
\usepackage{ifthen}
\usepackage{xspace}
\usepackage[sort&compress,numbers,square,comma,longnamesfirst]{natbib}
\usepackage{subcaption}
\usepackage{authblk}

\usetikzlibrary{patterns,arrows,decorations.pathreplacing,decorations.pathmorphing,calc,arrows.meta,shapes,backgrounds}

\definecolor[named]{urlblue}{cmyk}{1,0.58,0,0.21}

\hypersetup{
    breaklinks=true,
    colorlinks=true,
    citecolor=purple!70!blue!60!black,
    linkcolor=purple!70!blue!90!black,
    urlcolor=urlblue,
    pdflang={en},
    pdfborder={0 0 0},
    pdftitle={Tight Complexity Bounds for Counting Generalized Dominating Sets in Bounded-Treewidth Graphs Part I: Algorithmic Results},
    pdfauthor={Jacob Focke, Dániel Marx, Fionn {Mc Inerney}, Daniel Neuen, Govind S. {Sankar}, Philipp Schepper, and Philip Wellnitz},
    pdfcreator={},
}

\usepackage{amssymb}

\input{macros}
\usepackage[noabbrev,capitalize]{cleveref}

\numberwithin{equation}{section}
\numberwithin{figure}{section}

\newtheorem{theorem}{Theorem}[section]
\newtheorem{lemma}[theorem]{Lemma}
\newtheorem{fact}[theorem]{Fact}
\newtheorem{corollary}[theorem]{Corollary}

\newtheorem{example}[theorem]{Example}

\newtheorem{definition}[theorem]{Definition}
\newtheorem{remark}[theorem]{Remark}
\newtheorem{claim}[theorem]{Claim}

\crefname{mtheorem}{Main Theorem}{Main Theorems}
\Crefname{mtheorem}{Main Theorem}{Main Theorems}
\crefname{obs}{Observation}{Observations}
\Crefname{obs}{Observation}{Observations}
\crefname{fact}{Fact}{Facts}
\Crefname{fact}{Fact}{Facts}
\crefname{problem}{Problem}{Problems}
\Crefname{problem}{Problem}{Problems}
\crefname{conjecture}{Conjecture}{Conjectures}
\Crefname{conjecture}{Conjecture}{Conjectures}
\crefname{claim}{Claim}{Claims}
\Crefname{claim}{Claim}{Claims}

\newenvironment{claimproof}[1][\unskip]{ \begin{proof}[Proof of Claim #1.]
   }{ \end{proof} }

\title{	Tight Complexity Bounds for\\
		Counting Generalized Dominating Sets in\\
		Bounded-Treewidth Graphs\\[1ex]
		\large Part I: Algorithmic Results}
\author[1]{Jacob Focke}
\author[1]{D\'{a}niel Marx}
\author[2]{Fionn {Mc Inerney}}
\author[3]{Daniel Neuen}
\author[4]{\\Govind S. {Sankar}}
\author[1]{Philipp Schepper}
\author[5]{Philip Wellnitz}
\affil[1]{CISPA Helmholtz Center for Information Security}
\affil[2]{Algorithms and Complexity Group, TU Wien}
\affil[3]{Max Planck Institute for Informatics, SIC}
\affil[4]{Duke University}
\affil[5]{National Institute of Informatics and
    The Graduate University for Advanced Studies, SOKENDAI}
\date{}

\begin{document}

\maketitle

\vspace*{-1.0cm}

\begin{abstract}
  We investigate how efficiently a well-studied family of domination-type problems can be solved on bounded-treewidth graphs.
  For sets $\sigma,\rho$ of non-negative integers, a $(\sigma,\rho)$-set of a graph $G$ is a set $S$ of vertices such that $|N(u)\cap S|\in \sigma$ for every $u\in S$, and $|N(v)\cap S|\in \rho$ for every $v\not\in S$.
  The problem of finding a $(\sigma,\rho)$-set (of a certain size) unifies standard problems such as \textsc{Independent Set}, \textsc{Dominating Set}, \textsc{Independent Dominating Set}, and many others.

  For all pairs of finite or cofinite sets $(\sigma,\rho)$, we determine
  (under standard complexity assumptions) the best possible value $c_{\sigma,\rho}$ such that there is an algorithm that counts $(\sigma,\rho)$-sets in time $c_{\sigma,\rho}^\tw\cdot n^{\O(1)}$ (if a tree decomposition of width $\tw$ is given in the input).
  Let $\sigMax$ denote the largest element of $\sigma$ if $\sigma$ is finite, or the largest missing integer $+1$ if $\sigma$ is cofinite; $\rhoMax$ is defined analogously for $\rho$.
  Surprisingly, $c_{\sigma,\rho}$ is often significantly smaller than the natural bound $\sigMax+\rhoMax+2$ achieved by existing algorithms [van Rooij, 2020].
  Toward defining $c_{\sigma,\rho}$, we say that $(\sigma, \rho)$ is $\mname$-structured if there is a pair $(\alpha,\beta)$ such that every integer in $\sigma$ equals $\alpha$ mod $\mname$, and every integer in $\rho$ equals $\beta$ mod $\mname$.
  Then, setting
  \begin{itemize}
  	\item $c_{\sigma,\rho}=\sigMax+\rhoMax+2$ if $(\sigma,\rho)$ is not
  	$\mname$-structured for any $\mname\ge 2$,
  	\item $c_{\sigma,\rho}=\max\{\sigMax,\rhoMax\}+2$ if $(\sigma,\rho)$ is 2-structured,
  	but not $\mname$-structured for any $\mname \ge 3$, and $\sigMax=\rhoMax$ is even, and
  	\item $c_{\sigma,\rho}=\max\{\sigMax,\rhoMax\}+1$, otherwise,
  \end{itemize}
  we provide algorithms counting $(\sigma,\rho)$-sets in time $c_{\sigma,\rho}^\tw\cdot n^{\O(1)}$.
  For example, for the \textsc{Exact Independent Dominating Set} problem (also known as
  \textsc{Perfect Code}) corresponding to $\sigma=\{0\}$ and $\rho=\{1\}$, this improves the
  $3^\tw\cdot n^{\O(1)}$ algorithm
  of van Rooij to $2^\tw\cdot n^{\O(1)}$.

  Despite the unusually delicate definition of $c_{\sigma,\rho}$,
  an accompanying paper
  shows that our algorithms are most likely optimal, that is,
  for any pair $(\sigma, \rho)$ of finite or cofinite sets where the problem is non-trivial, and any $\varepsilon>0$, a $(c_{\sigma,\rho}-\varepsilon)^\tw\cdot
  n^{\O(1)}$-algorithm counting the number of $(\sigma,\rho)$-sets would violate the Counting Strong Exponential-Time Hypothesis (\#SETH).
  For finite sets $\sigma$ and $\rho$, these lower bounds also extend to the decision version, and hence, our algorithms are optimal in this setting as well.
  In contrast, for many cofinite sets, we show that further significant improvements for the decision and optimization versions are possible using the technique of representative sets.
\end{abstract}

\makeatletter{\renewcommand*{\@makefnmark}{}
	\footnotetext{%
	An extended abstract containing the results of this work
  and the accompanying paper \cite{FockeMMNSSW23ii}
  appeared in the proceedings of the 2023 {ACM-SIAM} Symposium on Discrete Algorithms (SODA 2023) \cite{FockeMINSSW23soda}.
  Research supported by the European Research Council (ERC) consolidator grant No.~725978 SYSTEMATICGRAPH and the Austrian Science Foundation (FWF, project Y1329).}\makeatother}

\input{s1-intro}
\input{s2-technical-overview}
\input{s3-prelims}

\input{s4-algorithms}
\input{s5-conclusion}

\small
\bibliographystyle{plainnat}
\bibliography{gen-dom-set}

\end{document}

%% file: macros.tex
\usepackage{amstext, amsthm, mathtools}
\usepackage{thm-restate}

\makeatletter
\DeclareRobustCommand{\crefnosort}[1]{%
	\begingroup\@cref@sortfalse\cref{#1}\endgroup
}
\makeatother

\newcommand{\CA}{{\mathcal A}}

\newcommand{\CC}{{\mathcal C}}

\newcommand{\CR}{{\mathcal R}}
\newcommand{\CS}{{\mathcal S}}

\newcommand{\CW}{{\mathcal W}}

\newcommand{\CZ}{{\mathcal Z}}

\newcommand{\ZZ}{{\mathbb Z}}
\newcommand{\FF}{{\mathbb F}}

\newcommand{\NN}{{\ZZ_{\geq 0}}}

\newcommand{\tw}{\textup{\textsf{tw}}}
\newcommand{\pw}{\textup{\textsf{pw}}}

\DeclareMathOperator{\cost}{cost}

\DeclareMathOperator{\bit}{bit}

\mathchardef\mhyph="2D
\newcommand{\DomSetGeneral}[4]{\ensuremath{(#1,#2)\mhyph}\textsc{{#3}Dom\-Set\ensuremath{^{#4}}}\xspace}

\newcommand{\srDomSet}{\DomSetGeneral{\sigma}{\rho}{}{}}

\newcommand{\srCountDomSet}{\DomSetGeneral{\sigma}{\rho}{\#}{}}

\DeclarePairedDelimiter{\floor}{\lfloor}{\rfloor}
\DeclarePairedDelimiter{\abs}{\lvert}{\rvert}

\renewcommand{\O}{O}
\newcommand{\Oh}{O} 

\newcommand{\ttop}{\textup{top}}
\newcommand{\rhoMax}{r_{\ttop}}
\newcommand{\sigMax}{s_{\ttop}}
\newcommand{\allMax}{t_{\ttop}}
\newcommand{\allCost}{t_{\cost}}

\newcommand{\rhoStates}{{\mathbb R}}
\newcommand{\sigStates}{{\mathbb S}}
\newcommand{\allStates}{{\mathbb A}}

\newcommand{\rhoStatesExt}{\mathbb R_{{\sf full}}}
\newcommand{\sigStatesExt}{\mathbb S_{{\sf full}}}
\newcommand{\allStatesExt}{\mathbb A_{{\sf full}}}

\usepackage{bm}

\def\fragmentco#1#2{\bm{[}\,#1\,\bm{.\,.}\,#2\,\bm{)}}
\def\fragmentoc#1#2{\bm{(}\,#1\,\bm{.\,.}\,#2\,\bm{]}}
\def\fragmentoo#1#2{\bm{(}\,#1\,\bm{.\,.}\,#2\,\bm{)}}
\def\fragment#1#2{\bm{[}\,#1\,\bm{.\,.}\,#2\,\bm{]}}
\def\nset#1{\numb{#1}}
\def\position#1{\bm{[}\,#1\,\bm{]}}
\let\pos\position
\def\vposition#1{\bm{[}\,#1\,\bm{]}}
\def\occ#1#2{\ensuremath \#_{#2}(#1)}

\def\mname{{\mathrm{m}}}
\def\sigvec#1{\ensuremath \vec{\sigma}(#1)}
\def\comp#1#2{\ensuremath {{#2\!\!\downarrow}_{#1}}}
\def\degvec#1{\ensuremath \vec{w}(#1)}
\def\capvec#1{\ensuremath \mathrm{cap}_{#1}}
\def\remvec#1#2#3{\ensuremath \mathrm{rem}_{#1}(#2, #3)}
\def\mdegvec#1{\ensuremath \vec{w}_{\mname}(#1)}

\newcommand{\numb}[1]{\fragment{1}{#1}}

\makeatletter
\renewenvironment{cases}{%
  \matrix@check\cases\env@cases
}{%
  \endarray\right.%
}
\def\env@cases{%
  \let\@ifnextchar\new@ifnextchar
  \left\lbrace
  \def\arraystretch{1.1}%
  \array{@{\;}c@{\quad}l@{}}%
}
\makeatother

\def\emptyset{\varnothing}

\def\epsilon{\varepsilon}

%% file: s1-intro.tex
\section{Introduction}
\label{sec:intro}

Since treewidth was defined independently in multiple equivalent ways in the 70s \cite{halin72,DBLP:journals/jct/BerteleB73,DBLP:journals/jct/RobertsonS84}, algorithms on bounded-treewidth graphs have been investigated for decades from different viewpoints.
It was observed already in the 80s that many of the basic NP-hard problems can be solved efficiently on bounded-treewidth graphs using a dynamic programming approach \cite{DBLP:conf/icalp/Bodlaender88,DBLP:journals/dam/ArnborgP89,BERN1987216}.
Courcelle's Theorem \cite{COURCELLE199012} formalized this observation for a large class of algorithmic problems definable in monadic second-order logic.
Algorithms on bounded-treewidth graphs were studied not only for their own sake, but also because they served as useful tools in other algorithms, most notably for planar problems and problems in the parameterized setting \cite{DBLP:journals/talg/BorradaileKM09,DBLP:conf/focs/Klein05,DBLP:conf/soda/AroraGKKW98,DemaineH08,DemaineFHT05,CyganFKLMPPS15}.

Over the years, the focus shifted to trying to make the algorithms as efficient as possible.
For example, given a tree decomposition of width $\tw$, the \textsc{Dominating Set} problem can be solved in time $3^\tw\cdot n^{\Oh(1)}$ using dynamic programming and subset convolution, but this running time was achieved only after multiple rounds of improvements \cite{TelleP93,AlberN02,RooijBR09,BjorklundHKK07}.
The search for more efficient algorithms is complemented by conditional lower bounds showing that certain forms of running times are the best possible.
For problems that can be solved in time $c^\tw\cdot n^{\Oh(1)}$, a line of research started by \citet{LokshtanovMS18} gives tight lower bounds on the best possible $c$ appearing in the running time \cite{DBLP:journals/siamcomp/OkrasaR21,DBLP:conf/esa/OkrasaPR20,DBLP:conf/stacs/EgriMR18,DBLP:conf/soda/CurticapeanLN18,DBLP:conf/iwpec/BorradaileL16,DBLP:journals/dam/KatsikarelisLP19,MarxSS21,CurticapeanM16,DBLP:journals/talg/FockeMR24}.
For example, \citet{LokshtanovMS18} showed that $3^\tw\cdot n^{\Oh(1)}$ is probably optimal for \textsc{Dominating Set}: assuming the Strong Exponential-Time Hypothesis (SETH), there is no algorithm for \textsc{Dominating Set} that solves the problem in time $(3-\epsilon)^\tw \cdot n^{\Oh(1)}$ for some $\epsilon>0$ if given a graph with a tree decomposition of width $\tw$.
The goal of this paper together with the accompanying paper \cite{FockeMMNSSW23ii}
is to prove similar tight bounds for a class of generalized domination problems.
In this paper, we focus on the algorithmic side, and provide improved algorithms for an infinite subclass of generalized domination problems.
Our findings show that, despite decades of intensive research, even very simple problems are poorly understood, and surprises can be found even when looking at the simplest of problems.

\citet{Telle94} introduced the notion of $(\sigma,\rho)$-sets as a common generalization of independent sets and dominating sets.
For sets $\sigma,\rho$ of non-negative integers, a \emph{$(\sigma,\rho)$-set} of a graph $G$ is a set $S$ of vertices such that $|N(u)\cap S|\in \sigma$ for every $u\in S$, and $|N(v)\cap S|\in \rho$ for every $v\not\in S$.
 With different choices of $\sigma$ and $\rho$, the problem of finding a $(\sigma, \rho)$-set (of a certain size) can express various well-studied algorithmic problems:
\begin{itemize}
\item $\sigma=\{0\}$, $\rho=\{0,1,\ldots\}$\\
  \textsc{Independent Set}: find a set $S$ of $k$ vertices that are pairwise non-adjacent.
\item $\sigma=\{0\}$, $\rho=\{0,1\}$\\
  \textsc{Strong Independent Set}: find a set $S$ of $k$ vertices that are pairwise at distance at least 3.
\item $\sigma=\{0,1,\ldots\}$, $\rho=\{1,2,\ldots\}$\\
  \textsc{Dominating Set}: find a set $S$ of $k$ vertices such that every remaining vertex has a neighbor in $S$.
\item $\sigma=\{0\}$, $\rho=\{1,2,\ldots\}$\\
  \textsc{Independent Dominating Set}: find an independent set $S$ of $k$ vertices such that every remaining vertex has a neighbor in $S$.
\item $\sigma=\{0\}$, $\rho=\{1\}$\\
  \textsc{Exact Independent Dominating Set/Perfect Code}: find an independent set $S$ of $k$ vertices such that every remaining vertex has \emph{exactly one} neighbor in $S$.
\item $\sigma=\{1,2,\ldots\}$, $\rho=\{1,2,\ldots\}$\\
 \textsc{Total Dominating Set}: find a set $S$ of $k$ vertices such that every vertex in the graph has \emph{at least one} neighbor in $S$.
\item $\sigma=\{0,1,\ldots\}$, $\rho=\{1\}$\\
  \textsc{Perfect Dominating Set}: find a set $S$ of $k$ vertices such that every remaining vertex has \emph{exactly one} neighbor in $S$.
\item $\sigma=\{0,1\ldots, d\}$, $\rho=\{0,1,\ldots\}$\\
\textsc{Induced Bounded-Degree Subgraph}: find a set $S$ of $k$ vertices that have \emph{at most $d$} neighbors in $S$.
\item $\sigma=\{d\}$, $\rho=\{0,1,\ldots\}$\\
\textsc{Induced $d$-Regular Subgraph}: find a set $S$ of $k$ vertices that have \emph{exactly $d$} neighbors in $S$.
\end{itemize}

Problems related to finding $(\sigma,\rho)$-sets received significant attention both from the complexity viewpoint and for demonstrating the robustness of algorithmic techniques \cite{DBLP:conf/iwpec/JaffkeKST18,DBLP:journals/dam/HalldorssonKT00,DBLP:journals/siamdm/TelleP97,DBLP:journals/algorithmica/FominGKKL11,DBLP:journals/corr/abs-1004-2642,DBLP:journals/ipl/FominGKKL09,DBLP:journals/dam/GolovachKS12,DBLP:journals/tcs/Bui-XuanTV13,DBLP:journals/tcs/Bui-XuanTV11,Rooij20,DBLP:conf/csr/Rooij21}.
Some authors call these types of problems \emph{locally checkable vertex subset problems} (LC-VSP) \cite{DBLP:journals/tcs/Bui-XuanTV13}.

For the case when each of $\sigma$ and $\rho$ is finite or cofinite, \citet{Rooij20} presented a general technique for finding $(\sigma,\rho)$-sets on graphs of bounded treewidth.
For a set $\sigma$ of finite or cofinite integers, we write $\sigMax$ to denote the maximum element of $\sigma$ if $\sigma$ is finite, and the maximum missing integer plus one if $\sigma$ is cofinite (for $\sigma = \NN$ we set $\sigMax = 0$); $\rhoMax$ is defined analogously based on $\rho$.
When one tries to solve a problem in time $f(\tw)\cdot n^{\Oh(1)}$ on a graph with a given tree decomposition of width $\tw$, and the goal is to make the function $f(\tw)$ as slowly growing as possible, then typically there are two main bottlenecks: the number of subproblems in the dynamic programming, and the efficient handling of join nodes.
It was observed by \citet{Rooij20} that, for the problem of finding a $(\sigma,\rho)$-set, each vertex in a partial solution has essentially $\sigMax+\rhoMax+2$ states.
For example, if a vertex is unselected in a partial solution and $\rho$ is finite, then we need to distinguish between having exactly $0$,~$1$,~$\dots$,~$\rhoMax$ neighbors in the partial solution, yielding $\rhoMax+1$ possibilities.
If $\rho$ is cofinite, then we need to distinguish between having exactly $0$, $1$, $\dots$, $\rhoMax-1$, or at least $\rhoMax$ neighbors (again $\rhoMax+1$ possibilities). In a similar way, a selected vertex has $\sigMax+1$ different states, giving a total number of $\sigMax+\rhoMax+2$ states for each vertex.
This suggests that we need to consider about $(\sigMax+\rhoMax+2)^\tw$
different subproblems at each node of the tree decomposition. Furthermore, \citet{Rooij20} showed that all these subproblems can be solved in time $(\sigMax+\rhoMax+2)^\tw\cdot n^{\Oh(1)}$ by using a fast generalized convolution algorithm in each step.
The algorithm can be extended to require a specific size for the set $S$, thus, allowing us to solve minimization/maximization problems or to count the number of solutions.

\begin{theorem}[\citet{Rooij20}]
 \label{thm:vanrooij-intro}
 Let $\sigma$ and $\rho$ be two finite or cofinite sets. Given a graph $G$ with a tree decomposition of width $\tw$ and an integer $k$, we can count the number of $(\sigma,\rho)$-sets of size exactly $k$ in time $(\sigMax+\rhoMax+2)^\tw\cdot n^{\Oh(1)}$.
\end{theorem}

No better algorithms were known for any pair of finite or cofinite sets $(\sigma,\rho)$ for any of the following variants: decision, minimization/maximization, and counting (ignoring problems for which polynomial-time algorithms are known\footnote{
For example, for the \textsc{Dominating Set} problem, the entire vertex set is always a solution, and hence, both the decision and maximization versions are trivial.
Further polynomial-time solvable cases are listed in \cite[Table~1]{Telle94}.}).

Is the upper bound in \cref{thm:vanrooij-intro} optimal for every pair $(\sigma,\rho)$?
As the main result of this work, we show that there are infinitely many pairs $(\sigma,\rho)$ for which $(\sigMax+\rhoMax+2)^\tw$ overstates the number of possible subproblems that we need to consider at each step of the dynamic programming algorithm.
Together with efficient convolution techniques that we develop for this problem, it follows that there are pairs $(\sigma,\rho)$ for which the  $(\sigMax+\rhoMax+2)^\tw\cdot n^{\Oh(1)}$ algorithm is not optimal and can be improved.

To be more specific, we say that $(\sigma, \rho)$ is \emph{$\mname$-structured} if there is a pair $(\alpha,\beta)$ such that $s \equiv \alpha \pmod \mname$ for every $s \in \sigma$, and $r \equiv \beta \pmod \mname$ for every $r \in \rho$.
For example, the pairs $(\{0,3\},\{3\})$ and $(\{0,3\},\{1,4\})$ are both $3$-structured, but the pair
$(\{0,3\},\{3,4\})$ is not $\mname$-structured for any $\mname\ge 2$. Notice that if a set is cofinite, then it cannot be $\mname$-structured for any $\mname\ge 2$. Furthermore, if $|\sigma|=|\rho|=1$, then $(\sigma,\rho)$ is $\mname$-structured for every $\mname$.
\begin{definition}
  \label{def:intro:baseOfRunningTime}
Let $\sigma$ and $\rho$ be two finite or cofinite sets of non-negative integers. We define
$c_{\sigma,\rho}$ by setting
  \begin{itemize}
  	\item $c_{\sigma,\rho} \coloneqq \sigMax+\rhoMax+2$ if $(\sigma,\rho)$ is not $\mname$-structured for any $\mname\ge 2$,
  	\item $c_{\sigma,\rho} \coloneqq \max\{\sigMax,\rhoMax\}+2$ if $(\sigma,\rho)$ is 2-structured, but not $\mname$-structured for any $\mname\ge 3$, and $\sigMax=\rhoMax$ is even, and
  	\item $c_{\sigma,\rho} \coloneqq \max\{\sigMax,\rhoMax\}+1$, otherwise.
  \end{itemize}
\end{definition}
Note that $c_{\sigma,\rho} = \max\{\sigMax,\rhoMax\}+1$ applies if $(\sigma,\rho)$ is $\mname$-structured for some $\mname\ge 3$, or 2-structured with $\sigMax\neq \rhoMax$,  or 2-structured with $\sigMax=\rhoMax$ being odd.
For example, $c_{\{0,3\},\{3\}}=4$, $c_{\{0,3\},\{1,4\}}=5$, $c_{\{1,3\},\{4\}}=5$, and $c_{\{2,4\},\{4\}}=6$.  Our main observation is that we need to consider only roughly $(c_{\sigma,\rho})^\tw$ subproblems at each step of the dynamic programming algorithm: if $(\sigma,\rho)$ is $\mname$-structured, then parity/linear algebra type of arguments show that many of the subproblems cannot be solved, and hence, need not be considered in the dynamic programming.

\begin{theorem}\label{thm:alg-main-intro}
Let $\sigma$ and $\rho$ denote two finite or cofinite sets. Given a graph $G$ with a tree decomposition of width $\tw$ and an integer $k$, we can count the number of $(\sigma,\rho)$-sets of size exactly $k$ in time $(c_{\sigma,\rho})^\tw\cdot n^{\Oh(1)}$.
\end{theorem}

In particular, as a notable example, for \textsc{Exact Independent Dominating Set}
(that is, $\sigma=\{0\}$, $\rho=\{1\}$), we have $\sigMax+\rhoMax+2=3$,
while $c_{\sigma,\rho}=2$ as $(\sigma,\rho)$ is 3-structured.
Therefore, \cref{thm:vanrooij-intro} gives a $3^\tw\cdot n^{\Oh(1)}$ time algorithm,
which we improve to $2^\tw\cdot n^{\Oh(1)}$ by \cref{thm:alg-main-intro}.
This shows that despite the decades-long interest
in algorithms for bounded-treewidth graphs,
there were new algorithmic ideas to discover
even for the \emph{most basic}
of the non-trivial $(\sigma,\rho)$-set problems.

The improvement from $\sigMax+\rhoMax+2$ to $c_{\sigma,\rho}=\max\{\sigMax,\rhoMax\}+1$ in the base of the exponent can be significant.
However, it is fair to say that the improvements of \cref{thm:alg-main-intro} over \cref{thm:vanrooij-intro} apply in somewhat exceptional cases (and there is even an exception to the exception where the improvement is only to $c_{\sigma,\rho}=\max\{\sigMax,\rhoMax\}+2$).
From the list of problems listed at the beginning of \cref{sec:intro}, algorithmic improvements are obtained only for the \textsc{Perfect Code} problem.
One may suspect that further improvements are possible, possibly leading to improvements that can be described in a more uniform way.
However, an accompanying paper \cite{FockeMMNSSW23ii} shows that the delicate nature of this improvement
is not a shortcoming of our algorithmic techniques,
but inherent to the problem: for the counting version,
$c_{\sigma,\rho}$ \emph{precisely} characterizes the best possible base of the exponent.
For the following lower bound, pairs $(\sigma,\rho)$
where the problem is trivially solvable need to be excluded.
A pair $(\sigma,\rho)$ is \emph{non-trivial}
if $\rho\neq \{0\}$, and $(\sigma,\rho)\neq (\{0,1,\ldots\},\{0,1,\ldots\})$.

\begin{theorem}[\cite{FockeMMNSSW23ii}]
    \label{thm:lower-main-intro}
	Let $(\sigma,\rho)$ denote a non-trivial pair of finite or cofinite sets.
	If there is an $\epsilon>0$ and an algorithm that
	counts in time $(c_{\sigma,\rho}-\epsilon)^\pw\cdot n^{\Oh(1)}$
	the number of $(\sigma,\rho)$-sets in a given graph
	with a given path decomposition of width $\pw$,
	then the Counting Strong Exponential Time Hypothesis (\#SETH) fails.
\end{theorem}

The algorithm of \cref{thm:alg-main-intro} achieves its running time by considering
roughly $(c_{\sigma,\rho})^\tw$ subproblems at each node of the tree decomposition. The
lower bound of \cref{thm:lower-main-intro} can be interpreted as showing that at least
that many subproblems really need to be considered by any algorithm that solves the counting problem, even when restricted to graphs of bounded pathwidth.
Does this remain true for the decision version as well?
For finite $\sigma$ and $\rho$ this is indeed the case as shown in \cite{FockeMMNSSW23ii}.

\begin{theorem}[\cite{FockeMMNSSW23ii}]
    \label{thm:lower-main-intro-decision}
	Let $(\sigma,\rho)$ be a pair of finite sets such that $0 \notin \rho$.
	If there is an $\epsilon>0$ and an algorithm that decides in time
	$(c_{\sigma,\rho}-\epsilon)^\pw\cdot n^{\Oh(1)}$
	whether there is a $(\sigma,\rho)$-set in a given graph
	with a given path decomposition of width $\pw$,
	then the Strong Exponential Time Hypothesis (SETH) fails.
\end{theorem}

Note that pairs $(\sigma,\rho)$ with $0 \in \rho$ have to be excluded since the empty set is a $(\sigma,\rho)$-set of any graph in this case.

Intriguingly, we show that the lower bounds from \cref{thm:lower-main-intro-decision} do not apply if $\sigma$ or $\rho$ is cofinite.
Indeed, it turns out that the technique
of {\em representative
sets}~\cite{FominLPS17,FominLPS16,Monien85,DBLP:conf/wg/PlehnV90,Marx09,DBLP:journals/iandc/BodlaenderCKN15}
can be used to significantly reduce the number of subproblems that need to be considered
at each node. The main idea is that we do not need to solve every subproblem, but rather,
we need only a small representative set of partial solutions with the property that if
there is a solution to the whole instance, then there is one that extends a  partial
solution in our representative set. This idea becomes relevant for example when $\sigma$
or $\rho$ is cofinite with only a few missing integers: then we do not need a collection
with every possible type of partial solution, but rather, we need only a collection of
partial solutions that can avoid a small list of forbidden degrees at every vertex. We
formalize this idea by presenting an algorithm where the base of the exponent does not depend on the largest missing integer in the cofinite set, but depends only on the number of missing integers.
We write $\emptyset \neq \tau \subseteq \NN$ for a finite or cofinite set.
If $\tau$ is finite, then we define $\cost(\tau) \coloneqq \max(\tau)$.
Otherwise, $\tau$ is cofinite and we define $\cost(\tau) \coloneqq |\NN \setminus \tau|$.
Further, let $\omega$ denote the matrix multiplication exponent \cite{AlmanW21}.

\begin{theorem}
 \label{thm:dp-representative-set-intro}
 Suppose $\sigma, \rho \subseteq \NN$ are finite or cofinite.
 Also, set $\allCost \coloneqq \max\{\cost(\sigma),\cost(\rho)\}$.
 Given a graph $G$ with a tree decomposition of width $\tw$, we can decide whether $G$ has a $(\sigma,\rho)$-set in time
 \[2^\tw \cdot (\allCost + 1)^{\tw(\omega + 1)} \cdot (\allCost + \tw)^{\Oh(1)} \cdot n.\]
\end{theorem}

While the algorithm of \cref{thm:dp-representative-set-intro} can be significantly more efficient than
the algorithm of \cref{thm:alg-main-intro}, it is unlikely to be tight in general, and it remains highly unclear what the best possible running time should be. For a tight result, one would need to overcome at least two major challenges: proving tight upper bounds on the size of representative sets, and understanding whether they can be handled without using matrix-multiplication based methods.

\subsection{Organization of this Article}

We give a general overview of ideas and proof techniques in \cref{sec:technical-overview}.
Then, after the introduction of some necessary preliminary definitions in \cref{sec:prelims}, we prove our algorithmic results.
In \cref{sec:alg-structured}, we show \cref{thm:alg-main-intro}, that is, how to count $(\sigma, \rho)$-sets in time $(c_{\sigma,\rho})^\tw\cdot n^{\Oh(1)}$. In \cref{sec:representative-set}, we give our representative set based algorithm for deciding whether a $(\sigma, \rho)$-set exists in order to prove \cref{thm:dp-representative-set-intro}.

%% file: s2-technical-overview.tex
\section{Technical Overview}
\label{sec:technical-overview}

We give an overview of the techniques used to prove our algorithmic results.

\subsection{Structured Sets}\label{sec:intro_structuredsets}

First, let us turn to \cref{thm:alg-main-intro}.
With \cref{thm:vanrooij-intro} in mind, it suffices to consider the case where $(\sigma,\rho)$ is $\mname$-structured for some $\mname \geq 2$.

As pointed out earlier, our algorithms are based on the ``dynamic programming on tree
decompositions'' paradigm. Hence, let us first briefly recall the definition of tree
decompositions (see \cref{sec:prelims} for more details).
A \emph{tree decomposition} of a graph $G$ consists of a rooted tree $T$ and a bag $X_t$ for
every node $t$ of $T$ with the following properties:
\begin{enumerate}
    \item every vertex $v$ of $G$ appears in at least one bag,
    \item for every vertex $v$ of $G$, the bags containing $v$ correspond to a connected
        subtree of $T$, and
    \item  if two vertices of $G$ are adjacent, then there is at least one bag containing both of them.
\end{enumerate}
The width of a tree decomposition is the size of the largest bag minus one, and the
treewidth of a graph \(G\) is the smallest possible width of a tree decomposition of \(G\).
For a node $t$ of $T$, we write $V_t$ for the union of all bags $X_{t'}$ where $t'$ is a descendant of $t$ (including $t$ itself).

Let us also recall the most common structure of (dynamic programming) algorithms on tree
decompositions (of width \(\tw\)). Typically, we define suitable subproblems for each node
$t$ of the decomposition, and then solve them in a bottom-up way.
In particular, we construct partial solutions that we aim to extend into full solutions
while moving up the tree decomposition.
In order to quickly identify which partial solutions can be extended to full solutions, we
classify them into a (limited) number of types: if two partial solutions have the same type and one has an
extension into a full solution, then the same extension would work for the other solution
as well.
Now, the subproblems at node $t$ correspond to computing which types of partial solutions are possible.
Finally, we need to argue that if we have solved every subproblem for every child $t'$ of
$t$, then the subproblems at $t$ can be solved efficiently as well.

For a more detailed description of how we implement said general approach, first suppose
for simplicity that both $\sigma$ and $\rho$ are finite (in fact, this is always the case
when $(\sigma,\rho)$ is $\mname$-structured for some $\mname \geq 2$; we are assuming throughout $\sigma,\rho$ to be finite or cofinite sets).
Now, in our case, a partial solution at a node $t$ is a set $S \subseteq V_t$ such that
$|N(u)\cap S|\in \sigma$ for every $u \in S \setminus X_t$, and $|N(v)\cap S|\in \rho$ for
every $v \in V_t \setminus (S \cup X_t)$, that is, all vertices of $V_t$ outside of $X_t$ satisfy the
$\sigma$-constraints and $\rho$-constraints, but the vertices in \(X_t\) may not.
In particular, it may happen that a vertex \(v \in X_t \setminus S\) does not yet have
a correct number of selected neighbors, that is, $|N(v) \cap S| \notin \rho$, since said
vertex may receive additional selected neighbors that lie outside of $V_t$.

Now,  two partial solutions $S_1, S_2 \subseteq V_t$ have the same type if
\begin{enumerate}
    \item $X_t \cap S_1 = X_t \cap S_2$, and
    \item $|N(v) \cap S_1| = |N(v) \cap S_2|$ for all $v \in X_t$.
\end{enumerate}
Indeed, in this situation, it is easy to verify that, for any $S' \subseteq V(G) \setminus
V_t$, we have that $S' \cup S_1$ is a $(\sigma,\rho)$-set if and only if $S' \cup S_2$ is a $(\sigma,\rho)$-set.
In other words, the type of a partial solution $S$ is determined by specifying, for each $v \in X_t$,
whether $v \in S$ and how many selected neighbors $v$ has.
Hence, we can describe such a type by a string $y \in \allStates^{X_t}$, where $\allStates =
\{\sigma_0,\dots,\sigma_{\sigMax},\rho_0,\dots,\rho_{\rhoMax}\}$, that associates with
every $v \in X_t$ a \emph{state} $y\position{v} \in \allStates$,
\begin{itemize}
    \item where state $\sigma_i$ means that $v \in S$ and \(v\) has $i$ selected neighbors, and
    \item where state $\rho_j$ means that $v \notin S$ and \(v\) has $j$ selected neighbors.
\end{itemize}
Observe that we can immediately dismiss partial solutions that assign too many selected
neighbors to a vertex (that is, for instance, a state \(\sigma_i\) for \(i > \sigMax\)),
since such partial solutions can never be extended to a full solution (assuming $\sigma$
is finite; for cofinite $\sigma$ all states $\sigma_i$, for $i \geq \sigMax$, are equivalent).
This gives a trivial upper bound of $(\sigMax + \rhoMax + 2)^{\tw+1}$ for the number of
types; and yields essentially the known algorithms.

However, it is highly unclear if said trivial upper bound is tight: is it really possible that for
every string $y \in \allStates^{X_t}$ there is some partial solution $S_y$ that
corresponds to $y$?
In general, this turns out to be indeed the case.
Consider the classical \textsc{Dominating Set} problem (that is, \(\sigma = \{0, 1, 2, \dots\}, \rho = \{1, 2, \dots\}\)) and the following example.
Suppose that for some bag \(X_t\), each of its vertices \(v_i \in X_t\) has a single neighbor \(v_i' \in V_t \setminus X_t\),
and suppose that all vertices \(v_{\star}'\) share a common neighbor \(w \in V_t\setminus X_t\).
Consult \cref{fig:207-intro} for a visualization of this example.

Now, for any string \(y \in \allStates^{X_t} = \{\sigma_0,\rho_0,\rho_1\}^{X_t}\), there is
indeed a partial solution (select \(w\), now any selection of \(v_{\star}'\) or
\(v_{\star}\) is a valid partial solution), that is, each string \(y \in \allStates^{X_t}\)
is indeed \emph{compatible} with \(t\). In particular, there are \((\sigMax + \rhoMax + 2)^{|X_t|}
= (0 + 1 + 2)^{|X_t|} = 3^{|X_t|}\) strings compatible with \(t\); for \(|X_t| = \tw + 1\)
the trivial upper bound on the number of types is thus indeed tight.

\begin{figure}[t]
    \centering
    \includegraphics[scale=1.35]{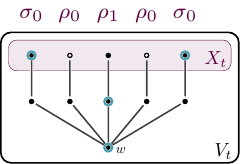}\qquad
    \includegraphics[scale=1.35]{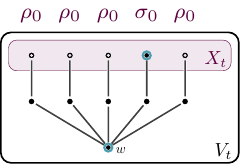}
    \caption{For \textsc{Dominating Set}
        (where \(\sigma = \{0, 1, 2, \dots\}, \rho = \{1, 2, \dots\}\)),
        it is easy to construct an example where any string
         \(y \in \allStates^{X_t} = \{\sigma_0,\rho_0,\rho_1\}^{X_t}\) is compatible
         with \(t\): we depict selected vertices as encircled blue and unselected vertices
         without a selected neighbor as a hollow black circle. Observe that after
         selecting \(w\), any selection of the remaining vertices constitutes a valid
         partial solution.
         Hence, there are \(3^{|X_t|} = (\sigMax + \rhoMax + 2)^{|X_t|}\) compatible
         strings in this case.
    }\label{fig:207-intro}
\end{figure}

In stark contrast to this rather unsatisfactory situation,
we show that for \(\mname\)-structured \((\sigma, \rho)\) where \(\mname \ge 2\), not all \((\sigMax + \rhoMax + 2)^{\tw+1}\) different strings (or types)
can be compatible with \(t\) --- we can then exploit this to obtain \cref{thm:alg-main-intro}.

To upper-bound the number of compatible strings in said case, let us first decompose strings $y \in
\allStates^{X_t}$ into a \emph{$\sigma$-vector} $\sigvec{y} \in \{0,1\}^{X_t}$, defined
via $\sigvec{y}\position{v} = 1$ if $y\position{v} = \sigma_c$ for some $c$ and $\sigvec{y}\position{v} = 0$ otherwise, and a
\emph{weight vector} $\degvec{y} \in \{0,\dots,\max\{\sigMax,\rhoMax\})\}^{X_t}$, defined via
$\degvec{y}\position{v} = c$ if $y\position{v} \in \{\sigma_c,\rho_c\}$.
Now, our main structural insight reads as follows.

\begin{lemma}
    \label{la:hamming-distance-parity-mod-p-intro}
    Suppose $(\sigma,\rho)$ is $\mname$-structured (for \(\mname \ge 2\)).
    Let $y, z \in \allStates^{X_t}$ denote strings that are compatible with $t$ with witnesses
    $S_y, S_z \subseteq V_t$ such that \(|S_{y} \setminus X_t| \equiv_{\mname} |S_{z}
    \setminus X_t|\).
    Then, \(\sigvec{y}\cdot\degvec{z} \equiv_{\mname} \sigvec{z}\cdot\degvec{y}\).
\end{lemma}
\begin{figure}[t]
    \centering
    \begin{subfigure}[b]{\textwidth}
    \centering
    \includegraphics[scale=1.35]{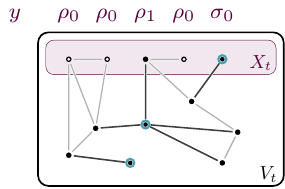}\qquad
    \includegraphics[scale=1.35]{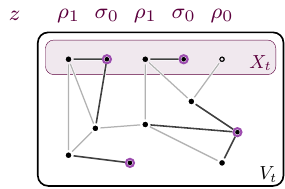}
        \caption{A bag \(X_t\) and the union \(V_t\) of the bags that are not above
            \(X_t\).
            For \(\sigma = \{0\}, \rho = \{ 1 \}\) (which are \(2\)-structured) the strings
        \(y = \rho_0\rho_0\rho_1\rho_0\sigma_0\) and
        \(z = \rho_1\sigma_0\rho_1\sigma_0\rho_0\) are compatible with \(t\); the vertices
        of the corresponding partial solutions \(S_y\) and \(S_z\) are encircled in blue
        and purple (respectively);
        we depict unselected vertices without selected neighbors as empty circles.
        We have \(|S_y\setminus X_t| = |S_z\setminus
        X_t| = 2\) and \(
            \sigvec{y}\cdot\degvec{z} = (0,0,0,0,1) \cdot (1,0,1,0,0) = 0
            \equiv_2 0 = (0,1,0,1,0) \cdot (0,0,1,0,0) = \sigvec{z}\cdot\degvec{y}.
            \)
    }\label{fig:hamming-distance-parity-mod-p-intro}
    \end{subfigure}
    \begin{subfigure}[b]{\textwidth}
        \centering
        \includegraphics[scale=1.35]{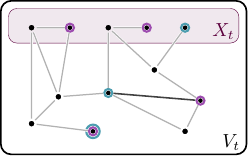}\qquad\qquad\qquad
        \includegraphics{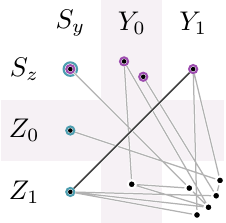}
        \caption{The partial solutions \(S_y\) and \(S_z\) depicted together.
            Observe that edges between a vertex \(v_y \in S_y\) and \(v_z \in S_z\) are
            possible only if \(v_y\) and \(v_z\) have exactly one selected neighbor in the
            corresponding other partial solution.
        }\label{fig:hamming-distance-parity-mod-p-intro2}
    \end{subfigure}
    \caption{An example for partial solutions and edges between them for \textsc{Perfect Code}.}
\end{figure}

For an intuition for this lemma, let us move to the \textsc{Exact Independent
Dominating Set} (or \textsc{Perfect Code}) problem as a specific example which corresponds to $\sigma = \{0\}$ and $\rho = \{1\}$.
Observe that $(\{0\},\{1\})$ is indeed $2$-structured.
Consider two partial solutions $S_y,S_z \subseteq V_t$.
(Also consult \cref{fig:hamming-distance-parity-mod-p-intro} for a visualization of an
example.)
Note that both $S_y$ and $S_z$ are independent sets in $G$ since $\sigma = \{0\}$.
Now, let us count the edges between $S_y$ and $S_z$.
To that end, let us define $Y_0$ as the set of vertices from $V_t \setminus S_y$ that have
no selected neighbor (i.e., no neighbor in $S_y$), and $Y_1$ as the set of vertices from $V_t \setminus S_y$ that have one selected neighbor (observe that $Y_0 \subseteq X_t$ since $\rho=\{1\}$).
Observe that $(S_y,Y_0,Y_1)$ forms a partition of $V_t$.
We define $Z_0$ and $Z_1$ analogously for the partial solution $S_z$.

Now, every vertex in $S_z \cap Y_0$ has no neighbor in $S_y$, while every vertex in $S_z
\cap Y_1$ has exactly one neighbor in $S_y$.
Recalling that vertices from $S_z \cap S_y$ have no neighbor in $S_y$ (since $S_y$ is an
independent set), the number of edges from $S_z$ to $S_y$ equals $|S_z \cap Y_1|$.
Repeating the same argument with reversed roles, we get that the number of edges from
$S_y$ to $S_z$ equals $|S_y \cap Z_1|$.
So $|S_z \cap Y_1| = |S_y \cap Z_1|$.
Assuming $X_t = V_t$, this condition is equivalent to \(\sigvec{y}\cdot\degvec{z} = \sigvec{z}\cdot\degvec{y}\).
To obtain the conclusion of the lemma also for $X_t \neq V_t$, we additionally use the assumption that \(|S_{y}
\setminus X_t| \equiv_{2} |S_{z} \setminus X_t|\) and $Y_0,Z_0 \subseteq X_t$.
Consult \cref{fig:hamming-distance-parity-mod-p-intro2} for a visualization of said proof
sketch for the example from \cref{fig:hamming-distance-parity-mod-p-intro}.

Now, how can we use \cref{la:hamming-distance-parity-mod-p-intro} to derive bounds on the number of compatible strings $y \in \allStates^{X_t}$?
First, we partition the compatible strings into sets $L_i$, for $i \in \fragment{0}{\mname-1}$, where $L_i$ contains all compatible strings $y$ for which there is a partial solution $S_y$ that satisfies $|S_{y} \setminus X_t| \equiv_\mname i$.
Hence, \cref{la:hamming-distance-parity-mod-p-intro} yields that \(\sigvec{y}\cdot\degvec{z} \equiv_{\mname} \sigvec{z}\cdot\degvec{y}\) for all $y,z \in L_i$.
We then show that $|L_i| \leq c_{\sigma,\rho}^{|X_t|}$ for all $i \in \fragment{0}{\mname-1}$ by using arguments that have a linear-algebra flavor.
For example, for the case $\sigma = \{0\}$ and $\rho = \{1\}$, we obtain that $|L_i| \leq 2^{|X_t|}$.
On a high level, our intuition here is that the condition \(\sigvec{y}\cdot\degvec{z}
\equiv_{\mname} \sigvec{z}\cdot\degvec{y}\) says that the set of $\sigma$-vectors is in
some sense ``orthogonal'' to the set of weight-vectors.
More formally, let us say that a set $A \subseteq X_t$ is \emph{$\sigma$-defining} for $L_i$ if $A$ is an inclusion-minimal set of positions that determines $\sigma$-vectors in $L_i$, i.e., fixing all positions $v \in A$ of a $\sigma$-vector $\vec{s}$ of some string in $L_i$ completely determines $\vec{s}$.
Then $B \coloneqq X_t \setminus A$ determines the weight-vectors in $L_i$ modulo $\mname$ in the following sense.

\begin{lemma}
 Suppose $A \subseteq X_t$ is a \emph{$\sigma$-defining} set for $L_i$ and set $B \coloneqq X_t \setminus A$.
 Then, for any two strings $y,z \in L_i$ with \(\sigvec{y} = \sigvec{z}\) and $\degvec{y}\vposition{v} \equiv_{\mname} \degvec{z}\vposition{v}$ for all $v \in B$, it holds that
 \[\degvec{y}\vposition{v} \equiv_{\mname} \degvec{z}\vposition{v}\]
 for all $v \in X_t$.
\end{lemma}

For example, for $\sigma = \{0\}$ and $\rho = \{1\}$, this implies that $|L_i| \leq 2^{|A|} \cdot 2^{|B|} = 2^{|X_t|}$.
Indeed, there are $2^{|A|}$ many potential $\sigma$-vectors $\vec{s}$, and, for every fixed $\vec{s}$, we have $2^{|B|}$ options for choosing the weight vector modulo $2$.
Since $\max\{\sigMax,\rhoMax\} \leq 1$,
determining the weight vector modulo $2$
actually completely determines the weight vector,
and so the upper bound follows.

For $\mname \geq 3$, the largest number of types can generally be achieved when $A = \emptyset$, since, intuitively speaking, in comparison to the trivial upper bound, we roughly save a factor of $\mname^{|A|} \cdot 2^{|B|}$.
In this case, it is easy to see that $|L_i| \leq (\max\{\sigMax,\rhoMax\} + 1)^{|X_t|}$
since we have that many choices for weight vectors.
It may be tempting to believe that the same holds for $\mname = 2$ (since here, it does not seem to matter how we choose $A$).
However, there is one notable exception when $\sigMax = \rhoMax$ is even.
In this case, consider the language
\[L^{*} \coloneqq\{y \in \allStates^{X_t} \mid \degvec{y}\position{v} \equiv_{\mname} 0 \text{ for all } v \in X_t\}\]
that arises for the choice $A = X_t$ and $B = \emptyset$.
It has size $|L^{*}|=(\floor{\rhoMax/\mname}+\floor{\sigMax/\mname}+2)^{|X_t|} = (\max\{\sigMax,\rhoMax\}+2)^{|X_t|}$ which explains why this case stands out in Definition \ref{def:intro:baseOfRunningTime}.

Overall, we obtain that the number of compatible strings is bounded by $\Oh(c_{\sigma,\rho}^{|X_t|})$.
With our improved bound at hand, we can now obtain improved dynamic programming algorithms.

At this point, we face a second challenge.
When performing dynamic programming along a tree decomposition, the most expensive step is to perform a join operation.
In this situation, the current node $t$ has exactly two children $t_1$ and $t_2$ and $X_t = X_{t_1} = X_{t_2}$.
In \cite{Rooij20}, van Rooij provides various convolution-based methods to efficiently compute the set of compatible strings for $t$ given the sets of compatible strings for $t_1$ and $t_2$.
We wish to apply the same methods, but unfortunately this is not directly possible since the convolution-based methods are not designed to handle restrictions of the input space.
To circumvent this issue, our solution is to design a specialized compression method that
again exploits the structure described above.
With the partition $(A,B)$ of $X_t$ at hand, this seems rather straightforward by simply omitting the redundant information.
However, the crux is that we need to design the compression in such a way that it
agrees with the join operation: two compatible strings $y_1$ and $y_2$ for $t_1$ and $t_2$ can be joined into a compatible string $y$ for $t$ if and only if the compressed strings $y_1'$ and $y_2'$ can be joined into the compressed string $y'$.
This condition makes the compression surprisingly tricky, and here we need to add several ``checksums'' to the compressed strings to ensure the required equivalence.

Overall, this allows us to prove \cref{thm:alg-main-intro}.
We complete this overview by stating the following variant of \cref{thm:alg-main-intro} which provides a linear-time algorithm for the decision version (i.e., it improves the dependence on the number of vertices in the running time from polynomial to linear for the decision version of the problem).

\begin{theorem}\label{thm:alg-main-decision-overview}
 Let $\sigma$ and $\rho$ denote two finite or cofinite sets.
 Given a graph $G$ with a tree decomposition of width $\tw$, we can decide whether $G$ has a $(\sigma,\rho)$-set in time $(c_{\sigma,\rho})^\tw \cdot (c_{\sigma,\rho} + \tw)^{\Oh(1)}\cdot |V(G)|$.
\end{theorem}

\subsection{Representative Sets}
The algorithm described in \cref{sec:intro_structuredsets} determined, for every node $t$ and string $y \in \allStates^{X_t}$, whether there is a partial solution $S \subseteq V_t$ that has type $y$. Our lower bounds show that this strategy is essentially optimal when we want to count the number of solutions. But, for the decision version, the idea of representative sets can give significant improvements in some cases.

As an illustrative example, let us consider the problem with $\sigma=\{0\}$, $\rho=\ZZ_{\ge0}\setminus \{10\}$, and suppose that $X_t=\{v\}$. Then, the partial solutions $S\subseteq V_t$ have $\rhoMax+\sigMax+2=13$ different types: either $v\in S$ has 0 neighbors in $S$, or $v\not\in S$ and $v$ has $0,1,\dots,10$, or $\ge 11$ neighbors in $S$. However, we do not need to know exactly which of these types correspond to partial solutions. A smaller amount of information is already sufficient to decide if there is a partial solution that is compatible with some extension $S'\subseteq V(G)\setminus V_t$. For example, if we have partial solutions for, say, the types $\sigma_0$, $\rho_7$, and $\rho_8$, then {\em every} extension $S'$ that extends \emph{some} partial solution $S\subseteq V_t$ extends one of these three partial solutions. Indeed, suppose that $v\not\in S$ and extension $S'$ gives $i$ further neighbors to $v$, then $S$ can be replaced by a partial solution corresponding to the $\rho_7$ state unless $i=3$, in which case $S$ can be replaced by the solution corresponding to $\rho_8$.

In general, we want to compute a \emph{representative set}
of all the partial solutions of $V_t$ such that
if there is one partial solution that is extended by some set $S' \subseteq V(G) \setminus V_t$,
then there is at least one partial solution in the representative set that is extendable by $S'$.
When $|X_t|>1$, then it is far from trivial to obtain upper bounds
on the size of the required representative sets and to compute them efficiently.
Earlier work \cite{MarxSS25} showed a connection between
these type of representative sets and representative sets in linear matroids,
and hence, known algebraic techniques can be used
\cite{FominLPS17,FominLPS16,DBLP:journals/jacm/KratschW20}.
The upper bounds depend on the number of missing integers from the cofinite sets,
but \emph{do not} depend on the largest missing element.
Thus, this technique is particularly efficient
when a few large integers are missing from $\rho$ or $\sigma$.
However, the price we have to pay is that
the algebraic techniques require matrix operations
and the matrix multiplication exponent $\omega$ appears in the exponent of the running time.
This makes it unlikely to obtain matching lower bounds
similar to \cref{thm:lower-main-intro}.

The technical details, including the proof of \cref{thm:dp-representative-set-intro}, are given in \cref{sec:representative-set}.

%% file: s3-prelims.tex
\section{Preliminaries}
\label{sec:prelims}
\subsection{Basics}

\subsubsection*{Numbers, Sets, Strings, and Vectors}
We use $a \equiv_{\mname} b$ as shorthand for $a \equiv b \pmod \mname$.

We write $\NN = \{0,1,2,3,\dots\}$ to denote the set of non-negative integers and $\ZZ_{>0} = \{1,2,3,\dots\}$ to denote the positive integers.
For integers \(i, j\), we write \(\fragment{i}{j}\) for the set \(\{i,\dots,j\}\), and \(\fragmentco{i}{j}\) for the set \(\{i,\dots,j-1\}\).
The sets \(\fragmentoc{i}{j}\) and \(\fragmentoo{i}{j}\) are defined similarly.
A set $\tau \subseteq \NN$ is \emph{cofinite} if $\NN \setminus \tau$ is finite.
Also, we say that $\tau$ is \emph{simple cofinite} if $\tau = \{n,n+1,n+2,\dots\}$ for some $n \in \NN$.

We write \(s =  s\position{1} s\position{2}\cdots  s\position{n}\) for a \emph{string} of length \(| s| = n\) over an alphabet \(\Sigma\).
We write \(\Sigma^{n}\) for the set (or \emph{language}) of all strings of length \(n\), and $\Sigma^* \coloneqq \bigcup_{n \geq 0} \Sigma^n$ for the set of all strings over $\Sigma$.
We use $\varepsilon$ to denote the empty string.
For a string \( s \in \Sigma^{n}\) and positions \(i \le j \in \fragment{1}{n}\),
we write \( s\fragment{i}{j} \coloneqq  s\position{i}\cdots  s\position{j}\);
accordingly, we define \( s\fragmentoc{i}{j}\), \( s\fragmentco{i}{j}\), and \( s\fragmentoo{i}{j}\).
Finally, for two strings \(s,t\), we write \(st\) for their concatenation --- in particular, we have \(s = s\fragment{1}{i}  s \fragmentoc{i}{n}\).
Sometimes we are interested in the number of occurrences of an element $a \in \Sigma$ in a string $s$.
To this end, we use the notation $\occ{s}{a}$ for the number of occurrences of $a$ in $s$,
that is, $\occ{s}{a} \coloneqq \abs{\{i \in \fragment{1}{\abs{s}} \mid s\position{i}=a\}}$.
Also, for $A \subseteq\Sigma$, $\occ{s}{A}$ denotes the number of occurrences of elements from $A$ in $s$,
that is, $\occ{s}{A} \coloneqq \sum_{a\in A} \occ{s}{a}$.

For a finite set $X$ (for instance, a set of vertices of a graph), we write \(\Sigma^X \coloneqq \Sigma^{|X|}\) to emphasize that we index the strings in (subsets of) \(\Sigma^X\) with elements from \(X\):
for an \(x \in X\), we write \(s\position{x}\) for the element at position $x$.

Similarly, for an \(n\)-dimensional vector space \(\mathbb{V}\), we view its elements as
strings of length \(n\) and correspondingly write
\(v = v\position{1} v\position{2} \dots v\position{n} \in \mathbb{V}\).
In addition to the notions defined for strings, for a set of positions \(P \subseteq \fragment{1}{n}\),
we write \(v\position{P} \coloneqq \bigcirc_{a \in P}\, x\position{a} \) for the \(|P|\)-dimensional vector that contains only the components of \(v\) whose indices are in \(P\).

To improve readability, we sometimes use a ``ranging star'' \(\star\) to range over unnamed objects.
For example, if we wish to define a function $f\colon \NN \times \NN \rightarrow \NN$, then we write $f(\star,4) = 5$ to specify that $f(i,4) = 5$ for all $i \in \NN$.

\subsubsection*{Graphs}

We use standard notation for graphs.
A \emph{graph} is a pair $G = (V(G),E(G))$ with finite vertex set $V(G)$ and edge set $E(G) \subseteq \binom{V(G)}{2}$.
Unless stated otherwise, all graphs considered in this paper are simple (that is, there
are no loops or multiple edges) and undirected.
We use $uv$ as a shorthand for edges $\{u,v\} \in E(G)$.
We write \(N_G(v)\) for the \emph{(open) neighborhood} of a vertex $v \in V(G)$,
that is, $N_G(v) \coloneqq \{w \in V(G) \mid vw \in E(G)\}$.
The \emph{degree} of $v$ is the size of its (open) neighborhood, that is,
$\deg_G(v) \coloneqq |N_G(v)|$.
The \emph{closed neighborhood} is $N_G\position{v} \coloneqq N_G(v) \cup \{v\}$.
We usually omit the index $G$ if it is clear from the context.
For $X \subseteq V(G)$, we write $G\position{X}$ to denote the \emph{induced subgraph} on the vertex set $X$, and $G - X \coloneqq G\position{V(G) \setminus X}$ denotes the induced subgraph on the complement of $X$.

\subsubsection*{Treewidth}

Next, we define tree decompositions and recall some of their basic properties.
For a more thorough introduction to tree decompositions and their many applications, we refer the reader to~\cite[Chapter 7]{CyganFKLMPPS15}.

Fix a graph $G$.
A \emph{tree decomposition} of $G$ is a pair $(T,\beta)$ that consists
of a rooted tree $T$ and a mapping $\beta\colon V(T) \to 2^{V(G)}$ such that
\begin{enumerate}[label = (T.\arabic*)]
 \item $\bigcup_{t \in V(T)} \beta(t) = V(G)$,
 \item for every edge $vw \in E(G)$, there is some node $t \in V(T)$ such that $\{u,v\} \subseteq \beta(t)$, and
 \item for every $v \in V(G)$, the set $\{t \in V(T) \mid v \in \beta(t)\}$ induces a connected subtree of $T$.
\end{enumerate}
The \emph{width} of a tree decomposition $(T,\beta)$ is defined as $\max_{t \in V(T)}|\beta(t)|-1$.
The \emph{treewidth} of a graph $G$, denoted by $\tw(G)$, is the minimum width of a tree decomposition of $G$.

When designing algorithms on graphs of bounded treewidth, it is instructive to work with
nice tree decompositions.
Let $(T,\beta)$ denote a tree decomposition and write $X_t \coloneqq \beta(t)$ for $t \in V(T)$.
We say \((T,\beta)\) is \emph{nice} if $X_r = \emptyset$ where $r$ denotes the root of $T$, $X_\ell = \emptyset$ for all leaves $\ell \in V(T)$, and every internal node $t \in V(T)$ has one of the following types:
\begin{description}
 \item[Introduce:] $t$ has exactly one child $t'$ and $X_t = X_{t'} \cup \{v\}$ for some $v \notin X_{t'}$; the vertex $v$ \emph{is introduced at $t$},
 \item[Forget:] $t$ has exactly one child $t'$ and $X_t = X_{t'} \setminus \{v\}$ for some $v \in X_{t'}$; the vertex $v$ \emph{is forgotten at $t$}, or
 \item[Join:] $t$ has exactly two children $t_1,t_2$ and $X_t = X_{t_1} = X_{t_2}$.
\end{description}

It is well-known that every tree decomposition $(T,\beta)$ of $G$ of width $\tw$ can be
turned into a nice tree decomposition of the same width $\tw$ of size $\Oh(\tw \cdot V(T))$
in time $\Oh(\tw^{2} \cdot \max\{|V(G),V(T)|\})$ (see, for instance, \cite[Lemma 7.4]{CyganFKLMPPS15}).

\subsection{Generalized Dominating Sets}

In the following, let $\sigma,\rho \subseteq \NN$ denote two sets that are finite or cofinite.

\subsubsection*{Basics}
Fix a graph $G$.
A set of vertices $S \subseteq V(G)$ is a \emph{$(\sigma,\rho)$-set} if $|N(u) \cap S| \in \sigma$ for every $u \in S$, and $|N(v) \cap S| \in \rho$ for every $v \in V(G) \setminus S$. We also refer to these two requirements as the $\sigma$-constraint and the $\rho$-constraint, respectively.
The (decision version of the) \srDomSet problem takes as input a graph $G$, and asks whether $G$ has a $(\sigma,\rho)$-set $S \subseteq V(G)$.
We use \srCountDomSet to refer to the counting version, that is, the input to the problem is a graph $G$, and the task is to determine the number of $(\sigma,\rho)$-sets $S \subseteq V(G)$.

We say $(\sigma,\rho)$ is \emph{trivial} if $\rho = \{0\}$ or $(\sigma,\rho) = (\NN,\NN)$.

\begin{fact}
 \label{fact:trivial-pairs-counting}
 Suppose $(\sigma,\rho)$ is trivial.
 Then, \srCountDomSet can be solved in polynomial time.
\end{fact}

\begin{proof}
 For $(\sigma,\rho) = (\NN,\NN)$, the number of $(\sigma,\rho)$-sets is $2^{|V(G)|}$.
 For $\rho = \{0\}$, the number of $(\sigma,\rho)$-sets is $2^{c}$, where $c$ denotes the number of those connected components of $G$ where every vertex degree is contained in $\sigma$.
\end{proof}

In order to analyze the complexity of \srDomSet (and \srCountDomSet) for non-trivial pairs $(\sigma,\rho)$, we associate the following parameters with $(\sigma,\rho)$.
We define
\begin{equation}\label{eq:rho-sig-max}
\sigMax \coloneqq
\begin{cases}
	\max(\sigma) & \text{if $\sigma$ is finite,} \\
	\max(\ZZ \setminus \sigma) + 1 & \text{if $\sigma$ is cofinite,}
\end{cases}
\text{and }
\rhoMax \coloneqq
\begin{cases}
	\max(\rho) & \text{if $\rho$ is finite,} \\
	\max(\ZZ \setminus \rho) + 1 & \text{if $\rho$ is cofinite.}
\end{cases}
\end{equation}
Observe that, if $\sigma = \NN$, then $\sigMax = 0$ (and similarly for $\rho$).
Moreover, we set $\allMax \coloneqq \max\{\sigMax,\rhoMax\}$.

Our improved algorithms heavily exploit certain structures of a pair $(\sigma,\rho)$;
formally we are interested in whether a pair is what we call ``\(\mname\)-structured''.

\begin{definition}[$\mname$-structured sets]
    Fix an integer $\mname \geq 1$.
    A set $\tau \subseteq \ZZ_{\ge 0}$ is \emph{$\mname$-structured} if there is some integer $c^* \in \ZZ_{\ge 0}$ such that
    \[c \equiv_{\mname} c^* \]
    for all $c \in \tau$.
\end{definition}
We say that $(\sigma,\rho)$ is \emph{$\mname$-structured} if both $\sigma$ and $\rho$ are $\mname$-structured.
Observe that $(\sigma,\rho)$ is always $1$-structured.

\subsubsection*{Partial Solutions and States}

For our algorithmic results, a key ingredient is the description of \emph{partial solutions}.

A \emph{graph with portals} is a pair $(G,U)$, where $G$ is a graph and $U \subseteq V(G)$.
If $U = \{u_1,\dots,u_k\}$, then we also write $(G,u_1,\dots,u_k)$ instead of $(G,U)$.

Intuitively speaking, the idea of this notion is that $G$ may be part of some larger graph that interacts with $G$ only via vertices from $U$.
In particular, in the context of the \srDomSet problem, vertices in $U$ do not necessarily need to satisfy the definition of a $(\sigma,\rho)$-set since they may receive further selected neighbors from outside of $G$.

\begin{definition}[partial solution]\label{def-partial-sol}\label{def:partialsol}
    Fix a graph with portals $(G,U)$.
    A set \(S \subseteq V(G)\) is a \emph{partial solution} (with respect to \(U\)) if
    \begin{enumerate}[label = (PS\arabic*)]
        \item \label{item:realizable-1s}\label{item:realizable-1}
            for each $v \in S \setminus U$, we have $|N(v) \cap S| \in \sigma$, and
        \item \label{item:realizable-3s}\label{item:realizable-3}
            for each $v \in V(G) \setminus (S\cup U)$, we have $|N(v) \cap S| \in \rho$.
    \end{enumerate}
\end{definition}

To describe whether vertices from $U$ are selected into partial solutions and how many selected neighbors they already have inside $G$, we associate a state with every vertex from $U$.

Formally, we write $\sigStatesExt \coloneqq \{\sigma_i \mid i \in \NN\}$ for the set of
potential $\sigma$-states, and we write $\rhoStatesExt \coloneqq \{\rho_i \mid i \in \NN\}$
for the set of potential $\rho$-states.
We also write $\allStatesExt \coloneqq \sigStatesExt \cup \rhoStatesExt$ for the set of \emph{all} potential states.

\begin{figure}[t]
    \centering
    \includegraphics[scale=1.3]{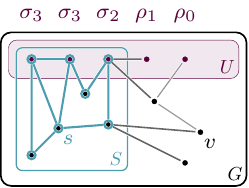}
    \caption{A graph \(G\) and subsets of vertices \(U\) and \(S\).
        For \(\sigma = \{2, 4\}, \rho = \{ 1 \}\), the set \(S\) is a partial solution
        (with respect to \(U\)), as every blue vertex \(s \in S \setminus U\) satisfies
        \(|N(s) \cap S| \in \{2,4\} = \sigma\) and every black vertex \(v \in V(G)
        \setminus (S \cup U)\) satisfies \(|N(v) \cap S| \in \{1\} = \rho\).
        The depicted set \(S\) corresponds to the compatible string
        \(\sigma_3\sigma_3\sigma_2\rho_1\rho_0\) (written above \(G\)).
        Note that \(S\) would not be a partial solution for \(\sigma = \{4\}\), as every
        blue vertex but one has only \(2\) neighbors in~\(S\).
    }\label{fig:partial}
\end{figure}

\begin{definition}[compatible strings]\label{def-compatible-general}\label{def:compatibleString}
    Fix a graph with portals $(G,U)$.
    A string \( x \in \allStatesExt^{U}\) is \emph{compatible with $(G,U)$}
    if there is a partial solution $S_{x} \subseteq V(G)$ such that
    \begin{enumerate}[label = (X\arabic*)]
        \item\label{item:realizable-2s}
            for each $v \in U \cap S_{x}$, we have
            \( x\position{v} = \sigma_{s}\), where $s = |N(v) \cap S_{x}|$, and
        \item\label{item:realizable-4s}
            for each $v \in U \setminus S_{x}$, we have
            \( x\position{v} = \rho_{r}\), where $r = |N(v) \cap S_{x}|$.
    \end{enumerate}
    We also refer to the vertices in $S_{x}$ as being \emph{selected} and say that $S_{x}$ is a \emph{(partial) solution}, \emph{selection}, or \emph{witness} that \emph{witnesses} $x$.
\end{definition}

Consult \cref{fig:partial} for a visualization of an example of a partial solution and its corresponding compatible string.

Observe that, despite $\allStatesExt$ being an infinite alphabet, for every graph with portals $(G,U)$, only finitely many strings $x$ can be realized.
Indeed, if $|V(G)| = n$, then every compatible string can have only characters from $\allStates_n = \sigStates_n \cup \rhoStates_n$, where $\sigStates_n \coloneqq \{\sigma_i \mid i \in \fragment{0}{n}\}$ and $\rhoStates_n \coloneqq \{\rho_i \mid i \in \fragment{0}{n}\}$.

\begin{definition}[realized language]\label{def:provider}
    For a graph with portals \((G, U)\), we define its \emph{realized language} as
    \[L(G,U) \coloneqq \{x \in \allStatesExt^U \mid x \text{ is compatible with } (G,U)\}.\]
\end{definition}

Again, observe that $L(G,U) \subseteq \allStates_n^U$, where $n = |V(G)|$.

In fact, for most of our applications, it makes sense to restrict the alphabet even further.
Recall the definition of $\sigMax$ and $\rhoMax$ from Equation~\eqref{eq:rho-sig-max}.
Suppose that $\sigma$ is finite.
Then, we are usually not interested in partial solutions $S$ where some vertex from $U$ is selected and already has more than $\sigMax$ selected neighbors (as it is impossible to extend this partial solution into a full solution).
Also, if $\sigma$ is infinite, it is usually irrelevant to us whether a selected vertex has exactly $\sigMax$ selected neighbors, or more than $\sigMax$ selected neighbors, since both options lead to the same outcome for all possible extensions of a partial solution.
For this reason, we typically%
\footnote{Sometimes, it turns out to be more convenient to work with the more general
variants; we clearly mark said (rare) occurrences of \(\allStatesExt\) and \(\allStates_n\).}
 restrict ourselves to the alphabets
\[\sigStates \coloneqq \{\sigma_0,\dots,\sigma_{\sigMax}\} \quad\text{and}\quad \rhoStates \coloneqq \{\rho_0,\dots,\rho_{\rhoMax}\}.\]
As before, we define $\allStates \coloneqq \sigStates \cup \rhoStates$.

%% file: s4-algorithms.tex
\section{Faster Algorithms for Structured Pairs}
\label{sec:alg-structured}

The goal of this section is to prove \cref{thm:alg-main-intro}.
With \cref{thm:vanrooij-intro} in mind, we can restrict ourselves to the case where $(\sigma,\rho)$ is $\mname$-structured for some $\mname \geq 2$.
In particular, both $\sigma$ and $\rho$ are finite in this case.

So, for the remainder of this section, suppose that
$\sigma, \rho \subseteq \ZZ_{\ge 0}$
are finite non-empty sets such that $\rho \neq \{0\}$.
Recall that $\sigMax \coloneqq \max\sigma$, $\rhoMax \coloneqq \max \rho$,
and $\allMax \coloneqq \max\{\sigMax,\rhoMax\}$.
Also suppose that $(\sigma,\rho)$ is $\mname$-structured for some $\mname \geq 2$, that is, there are integers $B, B' \in
\fragmentco{0}{\mname}$ such that
\(s \equiv_{\mname} B\) for every \(s \in \sigma\) and
\(r \equiv_{\mname} B'\) for every \(r \in \rho\).
Without loss of generality, we may assume that $\allMax + 1 \geq \mname$
(if $\mname > \allMax + 1$, then $|\rho| = |\sigma| = 1$,
which implies that $(\sigma, \rho)$ is $\mname$-structured
for every $\mname \geq 2$).

For this case, we present faster dynamic-programming on tree-decomposition-based
algorithms for \srDomSet. In particular, we prove the following result.

\begin{restatable}{theorem}{mainalgmstr}\label{thm:main-alg-m-str}
    Let $(\sigma,\rho)$ denote finite $\mname$-structured sets for some $\mname \geq 2$.
    Then, there is an algorithm $\CA$ that, given a graph $G$ and a nice tree decomposition
    of $G$ of width $\tw$, decides whether $G$ has a $(\sigma,\rho)$-set.

    If $\mname \geq 3$ or $\allMax$ is odd or $\min\{\sigMax,\rhoMax\} < \allMax$, then
    algorithm $\CA$ runs in time
    \[
      (\allMax + 1)^\tw \cdot 2^{(\allMax)^{\Oh(1)}} \cdot \tw^{\Oh(1)} \cdot |V(G)|.
    \]

    If $\mname = 2$, $\allMax$ is even, and $\sigMax = \rhoMax = \allMax$,
    then algorithm $\CA$ runs in
    time
    \[
      (\allMax + 2)^\tw \cdot 2^{(\allMax)^{\Oh(1)}} \cdot \tw^{\Oh(1)} \cdot |V(G)|.
    \]
\end{restatable}

Observe that \cref{thm:main-alg-m-str} is concerned with the decision version of the problem, and not the counting version.
The reason is that we find it more convenient to first explain the algorithm for the
decision version, and explain afterward how to modify the algorithm for the counting version.
Also observe that, for the decision version, we obtain a linear bound on the running time in terms of the number of vertices, whereas for the counting version, we only have a polynomial bound.
Finally, note that $\allMax$ only depends on $(\sigma,\rho)$,
and thus, the term $2^{(\allMax)^{\Oh(1)}}$ is a fixed constant
for any specific \srDomSet problem.

We prove \cref{thm:main-alg-m-str} in two steps. First, we obtain structural insights and
an upper bound on the number of states that need to be maintained during the run of the dynamic programming algorithm.
Second, we then show how to efficiently merge such states using a fast convolution-based
algorithm.

\subsection{Structural Insights into the \(\mname\)-Structured Case}

In this section, we work with the alphabet \(\allStates \coloneqq \{ \sigma_0, \dots,\sigma_{\sigMax}, \rho_0, \dots, \rho_{\rhoMax} \}\).
Also, recall that \(\sigStates \coloneqq \{ \sigma_0, \dots,\sigma_{\sigMax}\}\) and  \(\rhoStates \coloneqq \{\rho_0, \dots, \rho_{\rhoMax} \}\).
For a graph with portals $(G,U)$, we aim to obtain a (tight) bound on the size of the realized language $L(G,U)$ in terms of \(\sigMax\) and \(\rhoMax\).
Thereby, we also bound the number of states that are required in our
dynamic-programming-based approach for computing a solution for a given graph \(G\).
To that end, we first define certain vectors associated with a string
\(x \in \allStates^n\), essentially decomposing a string into its \(\sigma/\rho\)
component and its ``weight''-component.

To be able to reuse the definition in later sections, we state it for the full alphabet $\allStatesExt$.

\begin{definition}\label{def:vectors}
    For a string \(x \in \allStatesExt^n\), we define
    \begin{itemize}
        \item the \emph{\(\sigma\)-vector of \(x\)} as \(\sigvec{x} \in \{0, 1\}^n\)
            with\[
                \sigvec{x}\vposition{i} \coloneqq \begin{cases}
                                                   1 &\text{if } x\position{i} \in \sigStates,\\
                                                   0 &\text{if } x\position{i} \in \rhoStates.
                                                  \end{cases}\]
        \item the \emph{weight-vector of \(x\)} as \(\degvec{x} \in
                \ZZ_{\geq 0}^n\)
            with\[
                \degvec{x}\vposition{i} \coloneqq c, \quad\text{where } x\position{i} \in \{\sigma_c,
            \rho_c \}.\]
        \item the \emph{\(\mname\)-weight-vector of \(x\)} as \(\mdegvec{x} \in \ZZ_{\mname}^n\)
            with\[
                \mdegvec{x}\vposition{i} \coloneqq \degvec{x}\vposition{i} \bmod{\mname}.\]
    \end{itemize}
    For a language \(L \subseteq \allStatesExt^n\),
    we write \(\sigvec{L} \coloneqq \{ \sigvec{x} \mid x \in L\}\)
    for the set of all \(\sigma\)-vectors of~\(L\),
    we write \(\degvec{L} \coloneqq \{ \degvec{x} \mid x \in L\}\)
    for the set of all weight-vectors of \(L\), and
    we write \(\mdegvec{L} \coloneqq \{ \mdegvec{x} \mid x \in L\}\)
    for the set of all \(\mname\)-weight-vectors of \(L\).

    Finally, for a vector \(\vec{s} \in \{0, 1\}^n\), we define the \emph{capacity of
    \(\vec{s}\)} as \(\capvec{\vec{s}} \in \{0,\dots,\allMax\}^n\) with
    \[\capvec{\vec{s}}\vposition{i} \coloneqq \begin{cases}
                                         \sigMax &\text{if } \vec{s}\vposition{i} = 1,\\
                                         \rhoMax &\text{if } \vec{s}\vposition{i} = 0.
                                        \end{cases}\]
\end{definition}

To bound the size of realized languages, we proceed as follows.
First, we compare two different
partial solutions with respect to a fixed set \(U\) to obtain certain frequency properties of
the characters
of the corresponding string in \(\allStates^{U}\)
(expressed as \(\sigma\)-vectors and \(\mname\)-weight vectors). In a second step, we then
show that there is only a moderate number of strings with said structure.

Fix a graph $G$, a subset of its vertices $U \subseteq V(G)$, and the realized language $L \coloneqq L(G,U) \subseteq \allStates^{U}$.
For a string~$x \in L$ and an integer \(\mname \ge 2\), we define the ordered partition
\[P_{\mname}( x) \coloneqq (X_{\sigma,0},\dots, X_{\sigma,\mname-1},
X_{\rho, 0},\dots,X_{\rho, \mname-1}) \text{ of $U$}\] with\footnote{We chose this slightly
    convoluted-looking definition to simplify our exposition in
    \cref{la:hamming-distance-parity-mod-p-general}.}
\begin{align*}
    X_{\sigma, 0} &\coloneqq \{v \in U \mid \sigvec{x}\vposition{v} = 1 \text{ and }
    \mdegvec{x}\vposition{v} = 0\},\\
        &\dots\\
    X_{\sigma, \mname -1} &\coloneqq \{v \in U \mid \sigvec{x}\vposition{v} = 1 \text{ and }
    \mdegvec{x}\vposition{v} = \mname - 1\},\\
    X_{\rho, 0} &\coloneqq \{v \in U \mid \sigvec{x}\vposition{v} = 0 \text{ and }
    \mdegvec{x}\vposition{v} = 0\},\\
        &\dots\\
    X_{\rho, \mname -1} &\coloneqq \{v \in U \mid \sigvec{x}\vposition{v} = 0 \text{ and }
    \mdegvec{x}\vposition{v} = \mname - 1\}.
\end{align*}

\begin{figure}[t]
    \centering
    \includegraphics[scale=1.3]{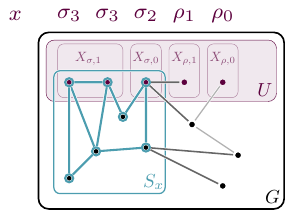}\qquad
    \includegraphics[scale=1.3]{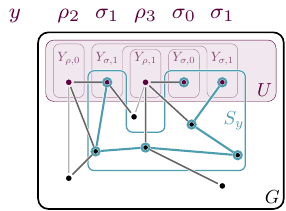}
    \caption{A graph \(G\) with portals \(U\). For \(\rho = \{ 1, 3 \}, \sigma = \{2,
        4\}\) (which are \(2\)-structured) the strings
        \(x = \sigma_3\sigma_3\sigma_2\rho_1\rho_0\) and
        \(y = \rho_2\sigma_1\rho_3\sigma_0\sigma_1\) are compatible with \((G,U)\); the
        corresponding partial solutions \(S_x\) and \(S_y\), as well as the partitions of
        \(U\) are depicted above. We have \(|S_x\setminus U| = |S_y\setminus
        U| = 4\) and \(
            \sigvec{x}\cdot\mdegvec{y} = (1,1,1,0,0) \cdot (0,1,1,0,1) = 2
            \equiv_2 2 = (0,1,0,1,1) \cdot (1,1,0,1,0) = \sigvec{y}\cdot\mdegvec{x}.
        \)
    }\label{fig:lem3-3}
\end{figure}

\begin{lemma}
    \label{la:hamming-distance-parity-mod-p-general}
    Let $(G,U)$ be a graph with portals and let \(L \coloneqq L(G,U) \subseteq \allStates^U\) denote its realized language.
    Also let $x,y \in L$ denote strings with witnesses
    $S_x, S_y \subseteq V(G)$ such that \(|S_{ x} \setminus U| \equiv_{\mname} |S_{ y}
    \setminus U|\).
    Then, \(\sigvec{x}\cdot\mdegvec{y} \equiv_{\mname} \sigvec{y}\cdot\mdegvec{x}\).
\end{lemma}
\begin{proof}
    Consult \cref{fig:lem3-3} for a visualization of an example.

    In a first step, we count the number of edges between \(S_{ x}\) and \(S_{
    y}\) in two different ways.
    The corresponding ordered partitions are \[P_{\mname}( x)=(X_{\sigma,0},\dots,X_{\sigma,\mname-1}, X_{\rho, 0},\dots,X_{\rho,
        \mname-1}) \;\text{and}\;
    P_{\mname}( y)=(Y_{\sigma,0},\dots,Y_{\sigma,\mname-1},Y_{\rho, 0},\dots,Y_{\rho, \mname-1}).\]
    Recall the integers $B, B' \in
    \fragmentco{0}{\mname}$ with
    \(s \equiv_{\mname} B\) for every \(s \in \sigma\) and
    \(r \equiv_{\mname} B'\) for every \(r \in \rho\).
    Observe that by \ref{item:realizable-3s} from \cref{def:partialsol},
    every vertex in \(V(G) \setminus (U \cup S_{ y})\) has \((B' + \mname\, \ell)\) neighbors in
    \(S_{ y}\) (for some non-negative integer \(\ell\)). In particular, this holds for all
    vertices in \(S_{ x} \setminus (U \cup S_{ y})\).
    Similarly, observe that by \ref{item:realizable-1s},
    every vertex in \((S_{x} \cap S_{y}) \setminus U\) has \((B + \mname\, \ell')\) neighbors in
    \(S_{ y}\) (for some non-negative integer \(\ell'\)).
    Finally, a vertex \(v\) in \(S_{ x} \cap U = \bigcup_{j \in \fragmentco{0}{\mname}}
    X_{\sigma,j}\) has exactly \(i + \mname\, \ell''\) neighbors in  \(S_{ y}\) (for some
    non-negative integer \(\ell''\)) if and only if \(v\) is in one of \(Y_{\rho, i}\) or \(Y_{\sigma, i}\).

    Writing $\vec{E}(X,Y) \coloneqq \{(v,w) \in E(G) \mid v \in X,
    w \in Y\}$, we obtain
    \begin{align*}
        |\vec{E}(S_{x},S_{y})| &\equiv_{\mname} B' \cdot |S_{x} \setminus (U \cup S_{y})|\\
                               &\quad+ B \cdot |(S_{x} \cap S_{y}) \setminus U| \\
                               &\quad+\sum_{i \in \fragmentco{1}{\mname}} \sum_{j \in \fragmentco{0}{\mname}} i\cdot
        \big(|X_{\sigma,j} \cap Y_{\rho, i}| + |X_{\sigma,j} \cap Y_{\sigma, i}|\big)\\
                                   &\equiv_{\mname} B' \cdot (|S_{x} \setminus U| - |S_{x} \cap S_{y}| + |S_{x} \cap S_{y} \cap U|)\\
                                   &\quad+ B \cdot |(S_{x} \cap S_{y}) \setminus U| \\
                                   &\quad+\sum_{i \in \fragmentco{1}{\mname}} \sum_{j \in \fragmentco{0}{\mname}} i\cdot
        \big(|X_{\sigma,j} \cap Y_{\rho, i}| + |X_{\sigma,j} \cap Y_{\sigma, i}|\big).
    \intertext{In a symmetric fashion, we count the edges from \(S_{y}\) to \(S_{x}\):}
        |\vec{E}(S_{ y},S_{ x})| &\equiv_{\mname} B' \cdot (|S_{y} \setminus U| - |S_{y} \cap S_{x}| + |S_{y} \cap S_{x} \cap U|)\\
                                   &\quad+ B \cdot |(S_{y} \cap S_{x}) \setminus U| \\
                                   &\quad+\sum_{i \in \fragmentco{1}{\mname}} \sum_{j \in \fragmentco{0}{\mname}} i\cdot
        \big(|Y_{\sigma,j} \cap X_{\rho,i}| + |Y_{\sigma,j} \cap X_{\sigma,i}|\big).
    \end{align*}
    Now, we combine the previous equations and use the assumption that
    \(|S_{x} \setminus U| \equiv_{\mname} |S_{y} \setminus U|\) to
    obtain
    \[
                                   \sum_{i \in \fragmentco{1}{\mname}} \sum_{j \in
                                   \fragmentco{0}{\mname}} i\cdot
        \big(|X_{\sigma,j} \cap Y_{\rho,i}| + |X_{\sigma,j} \cap Y_{\sigma,i}|\big)
        \equiv_{\mname}  \sum_{i \in \fragmentco{1}{\mname}} \sum_{j \in
                                   \fragmentco{0}{\mname}} i\cdot
        \big(|Y_{\sigma,j} \cap X_{\rho,i}| + |Y_{\sigma,j} \cap X_{\sigma,i}|\big).
    \]
    Next, we unfold the definitions for \(X_{\star,\star}, Y_{\star,\star}\)
    and observe that, for every $i \in \fragmentco{1}{\mname}$, it holds that
    \begin{align*}
        &
        \sum_{j \in \fragmentco{0}{\mname}} i\cdot
        \big(|X_{\sigma,j} \cap Y_{\rho, i}| + |X_{\sigma,j} \cap Y_{\sigma, i}|\big)\\
        &\quad\equiv_{\mname}
         \sum_{j \in \fragmentco{0}{\mname}} i\cdot
        |\{ k \in \fragment{1}{n} \mid  \sigvec{x}\vposition{k} = 1,~
        \mdegvec{x}\vposition{k} = j, \text{ and } \mdegvec{y}\vposition{k}
            = i \}|\\
        &\quad\equiv_{\mname}
          i\cdot
         |\{ k \in \fragment{1}{n} \mid  \sigvec{x}\vposition{k} = 1
            \text{ and } \mdegvec{y}\vposition{k} = i \}|.
    \end{align*}
    It follows that
    \begin{align*}
        &
        \sum_{i \in \fragmentco{1}{\mname}} \sum_{j \in \fragmentco{0}{\mname}} i\cdot
        \big(|X_{\sigma,j} \cap Y_{\rho, i}| + |X_{\sigma,j} \cap Y_{\sigma, i}|\big)\\
        &\quad\equiv_{\mname}
        \sum_{i \in \fragmentco{1}{\mname}}  i\cdot
         |\{ k \in \fragment{1}{n} \mid  \sigvec{x}\vposition{k} = 1
            \text{ and } \mdegvec{y}\vposition{k} = i \}|\\
        &\quad\equiv_{\mname} \sigvec{x} \cdot \mdegvec{y}.
    \end{align*}
    Similarly, we obtain that
    \[\sum_{i \in \fragmentco{1}{\mname}} \sum_{j \in \fragmentco{0}{\mname}} i\cdot
        \big(|Y_{\sigma,j} \cap X_{\rho,i}| + |Y_{\sigma,j} \cap X_{\sigma,i}|\big)
      \equiv_{\mname} \sigvec{y} \cdot \mdegvec{x}.\]
    All together, the claimed
    \(\sigvec{x} \cdot \mdegvec{y} \equiv_{\mname} \sigvec{y} \cdot \mdegvec{x}\) follows, completing the proof.
\end{proof}

With \cref{la:hamming-distance-parity-mod-p-general} in mind, in order to bound the size of a realized language, it suffices (up to a factor of $m$) to bound the size of a language $L \subseteq \allStates^n$ such that $L \times L \subseteq \CR^n$, where
\begin{align}
    \CR^n \coloneqq \{ ( x,  y) \in \allStates^n \times \allStates^n
    \mid \sigvec{x} \cdot \mdegvec{y} \equiv_{\mname} \sigvec{y} \cdot \mdegvec{x}\}.\label{eq:r-def}
\end{align}
Observe that the relation $\CR^n$ is reflexive and symmetric.

We proceed to exploit the relation $\CR^n$ to obtain size bounds for realized languages
(in \cref{sec:singlecompatible}).
Afterward, in \cref{sec:combinationcompatible}, we investigate how said size bounds behave
when combining realized languages.

\subsubsection{Bounding the Size of a Single Realized Language}\label{sec:singlecompatible}

The goal of this section is to show the following result.
\begin{restatable}{theorem}{upperboundmain}
    \label{thm:upper-bound}
    Let $L \subseteq \allStates^n$ denote a language with $L \times L \subseteq \CR^n$.

	If $\mname \geq 3$ or $\allMax$ is odd or $\min\{\sigMax,\rhoMax\} < \allMax$,
	then $|L| \leq (\allMax + 1)^n$.

	If $\mname = 2$, $\allMax$ is even, and $\sigMax = \rhoMax = \allMax$,
	then $|L| \leq (\allMax + 2)^n$.
\end{restatable}

Note that the bounds of \cref{thm:upper-bound} are essentially optimal; consider the
following example.

\begin{example}\label{exa:large-structured-languages}
    Consider the languages
    \begin{itemize}
        \item $L_1 \coloneqq \rhoStates^n = \{v \in \allStates^n \mid \sigvec{v} = 0\}$,
        \item $L_2 \coloneqq\{v \in \sigStates^n \mid \sum_{\ell \in \nset{n}} \degvec{v}\position{\ell} \equiv_{\mname} 0\}$, and
		\item $L_3 \coloneqq\{v \in \allStates^n \mid \degvec{v}\position{\ell} \equiv_{\mname} 0 \text{ for all } \ell \in \fragment1n\}$.
	\end{itemize}
	It is straightforward to see that $L_i \times L_i \subseteq \CR^n$ for all $i \in \{1,2,3\}$.
	We have that $|L_1| = (\rhoMax + 1)^n$, $|L_2| \geq (\sigMax + 1)^{n-1}$ and
	\[|L_3| = \left(\left\lceil\frac{\sigMax + 1}{\mname}\right\rceil
                + \left\lceil\frac{\rhoMax + 1}{\mname}\right\rceil\right)^n.\]
	Observe that $|L_3| = (\allMax + 2)^n$ if $\mname = 2$, $\allMax$ is even, and $\rhoMax = \sigMax = \allMax$.
	In all other cases, $|L_3| \leq (\allMax + 1)^n$.
	In particular, the language $L_3$ indicates why the case $\mname = 2$
  with even $\sigMax = \rhoMax = \allMax$ stands out.
\end{example}

Toward proving \cref{thm:upper-bound}, we start by showing that,
for strings with the same \(\sigma\)-vector, the difference of their \(\mname\)-weight-vectors
is ``orthogonal'' to the \(\sigma\)-vector of any other string.

\begin{lemma}
    \label{la:inner-product-zero}
    Let $L \subseteq \allStates^n$ denote a language with $L \times L \subseteq \CR^n$.
    For any three strings $v,w,z \in L$ with $\sigvec{v} = \sigvec{w}$, we have
    \[\big(\mdegvec{v} - \mdegvec{w}\big) \cdot \sigvec{z} \equiv_m 0.\]
\end{lemma}
\begin{proof}
    Fix strings \(v, w, z \in L\).
    By the definition of \(\CR^n\), we have
    \begin{equation*}
        \sigvec{v} \cdot \mdegvec{z} \equiv_{\mname} \sigvec{z} \cdot \mdegvec{v}
        \quad\text{and}\quad
        \sigvec{w} \cdot \mdegvec{z} \equiv_{\mname} \sigvec{z} \cdot \mdegvec{w}.
    \end{equation*}
    Using the assumption that $\sigvec{v} = \sigvec{w}$, we conclude that
    \[ \sigvec{z} \cdot \mdegvec{v} \equiv_{\mname} \sigvec{z} \cdot \mdegvec{w},\]
    which yields the claim after rearranging.
\end{proof}

Next, we explore the implications of \cref{la:inner-product-zero}.
Intuitively, we show that, for a language \(L\) of strings of length $n$,
each of the $n$ positions contributes either to vectors from \(\sigvec{L}\) or to vectors from
\(\mdegvec{L}\).
Formally, let us start with the notion of a \(\sigma\)-defining set.

\begin{definition}[$\sigma$-defining set]\label{def:sig-set}
    \label{rem:partition-depends-on-sigma-vectors}
    Let $L \subseteq \allStates^n$.
    A set \(S \subseteq \fragment1n\) is \emph{\(\sigma\)-defining for \(\sigvec{L}\)}
    if \(S\) is an inclusion-minimal set of positions that
    uniquely characterize the \(\sigma\)-vectors of the strings in \(L\), that is,
    for all \(u, v \in L\),
    we have
    \begin{equation}
        \label{item:partition-positions-sig}
        \sigvec{u}\vposition{S} = \sigvec{v}\vposition{S}
        \quad\Longrightarrow\quad
        \sigvec{u} = \sigvec{v}.
    \end{equation}
\end{definition}
\begin{remark}\label{rem:witness}
    As a \(\sigma\)-defining \(S\) is (inclusion-)minimal,
    observe that, for each
    position \(i \in S\), there are pairs of witness vectors \(w_{1,i}, w_{0,i} \in
    \sigvec{L}\) that  differ (on \(S\))  only at position \(i\),
    with \(w_{1,i}\vposition{i} = 1\), that is,
    \begin{itemize}
        \item $w_{1,i}\vposition{S \setminus i}
            = w_{0,i}\vposition{S \setminus i}$,
        \item $w_{1,i}\vposition{i} = 1$, and
        \item $w_{0,i}\vposition{i} = 0$.
    \end{itemize}
    We write \(\CW_{S} \coloneqq \{ w_{1,i}, w_{0,i} \mid i \in S \}\) for a set of witness
    vectors for \(\sigvec{L}\). Note that, as \(S\) itself, the witness
    vectors \(\CW_{S}\) do not directly depend on strings in \(L\), but only on
    the \(\sigma\)-vectors of \(L\).
\end{remark}

\begin{lemma}
    \label{la:partition-positions}
    Let $L \subseteq \allStates^n$ denote a language with $L \times L \subseteq \CR^n$
    and let \(S\) denote a \(\sigma\)-defining set for \(L\).

    Then, for any two strings \(u, v\in L\) with \(\sigvec{u} = \sigvec{v}\), the
    remaining positions \(\bar{S} \coloneqq \fragment1n \setminus S\) uniquely characterize the
    \(\mname\)-weight vectors of \(u\) and \(v\), that is, we have
    \begin{equation}
        \label{item:partition-positions-ind}
        \mdegvec{u}\vposition{\bar{S}} = \mdegvec{v}\vposition{\bar{S}}
        \quad\Longrightarrow\quad
        \mdegvec{u} = \mdegvec{v}.
    \end{equation}
\end{lemma}
\begin{proof}
    Let $S \subseteq \fragment1n$ denote a \(\sigma\)-defining set for \(\sigvec{L}\)
    with witness vectors \(\CW_{S}\) (see \cref{rem:witness}),
    and consider the set \(\bar{S} \coloneqq \fragment1n \setminus S\).
    We proceed to show that \eqref{item:partition-positions-ind} is satisfied.
    To that end, let $u, v \in L$ denote strings with $\sigvec{u} = \sigvec{v}$ and
    \(\mdegvec{u}\vposition{\bar{S}} = \mdegvec{v}\vposition{\bar{S}}\).
    We need to argue that $\mdegvec{u} = \mdegvec{v}$, and, in particular, that
    $\mdegvec{u}\vposition{S} = \mdegvec{v}\vposition{S}$.
    Hence, we proceed to show that, for every \(i \in S\), we have
    \(\mdegvec{u}\vposition{i} = \mdegvec{v}\vposition{i}\).

    Now, fix a position \(i \in S\) and corresponding witness vectors \(w_{1,i}, w_{0,i}
    \in \CW_{S}\).
    We proceed by showing two immediate equalities.
    \begin{claim}\label{cl:partition-1}
        The strings \(u, v, w_{1,i}\), and \(w_{0,i}\) satisfy
        \[\big(\mdegvec{u} - \mdegvec{v}\big) \cdot \big(w_{1,i} -
            {w_{0,i}}\big)
        \equiv_{\mname} 0.\]
    \end{claim}
    \begin{claimproof}
        By \cref{la:inner-product-zero}, we have that
        \begin{align*}
            \big(\mdegvec{u} - \mdegvec{v}\big) \cdot \big({w_{1,i}} -
            {w_{0,i}}\big)
            &\equiv_{\mname}
            \big(\mdegvec{u} - \mdegvec{v}\big) \cdot {w_{1,i}}
            - \big(\mdegvec{u} - \mdegvec{v}\big) \cdot {w_{0,i}} \\
            &\equiv_{\mname} 0 - 0 \\
            &\equiv_{\mname} 0,
        \end{align*}
        which completes the proof.
    \end{claimproof}

    \begin{claim}\label{cl:partition-2}
        The strings \(u, v, w_{1,i}\), and \(w_{0,i}\) satisfy
        \[\big(\mdegvec{u} - \mdegvec{v}\big) \cdot \big({w_{1,i}} -
            {w_{0,i}}\big)
        \equiv_{\mname}
        \big(\mdegvec{u}\vposition{i} - \mdegvec{v}\vposition{i}\big) \cdot
        \big({w_{1,i}}\vposition{i} - {w_{0,i}}\vposition{i}\big).
    \]
    \end{claim}
    \begin{claimproof}
        Observe that, by assumption, for every component \(j \in \bar{S}\), we have
        \(\mdegvec{u}\vposition{j} = \mdegvec{v}\vposition{j}\).
        Observe further that, by the definitions of \(i\) and \(S\), for every component
        \(j \in S\setminus i\), we have
        \({w_{1,i}}\vposition{j} = {w_{0,i}}\vposition{j}\), which
        yields the claim.
    \end{claimproof}

    Now, combining \cref{cl:partition-1,cl:partition-2} yields
    \begin{align*}
        0
        &\equiv_{\mname}
        \big(\mdegvec{u} - \mdegvec{v}\big) \cdot \big({w_{1,i}} - {w_{0,i}}\big)\\
        &\equiv_{\mname}
        \big(\mdegvec{u}\vposition{i} - \mdegvec{v}\vposition{i}\big)
        \cdot \big({w_{1,i}}\vposition{i} - {w_{0,i}}\vposition{i}\big)\\
        &\equiv_{\mname}
        \big(\mdegvec{u}\vposition{i} - \mdegvec{v}\vposition{i}\big)
        \cdot \big(1 - 0\big)\\
        &\equiv_{\mname}
        \mdegvec{u}\vposition{i} - \mdegvec{v}\vposition{i}.
    \end{align*}

    In other words, $\mdegvec{u}\vposition{i} = \mdegvec{v}\vposition{i}$
    for all $i \in S$, which yields \eqref{item:partition-positions-ind}, and hence, the claim.
\end{proof}

As a direct consequence of \cref{la:partition-positions}, we obtain a first upper bound on
the size of languages \(L \subseteq \allStates^n\) with $L \times L \subseteq \CR^n$.

\begin{corollary}
    \label{claim:upper-bound}
    Let $L \subseteq \allStates^n$ denote a language with $L \times L \subseteq \CR^n$,
    and let \(S\) denote a \(\sigma\)-defining set for \(\sigvec{L}\).
    Then, we have
    \[|L| \leq
        \smash{(\allMax + 1)^{n - |S|} \cdot \sum_{k = 0}^{|S|} \binom{|S|}{k}
    \left\lceil\frac{\sigMax + 1}{\mname}\right\rceil^{k} \left\lceil\frac{\rhoMax +
    1}{\mname}\right\rceil^{|S|-k}}. \raisebox{-2ex}{\null} \]
\end{corollary}
\begin{proof}
    For a $k \in \{0,\dots,|S|\}$, write $L_k$ to denote the set of
    all strings $x \in L$ such that $\sigvec{x}\vposition{S}$ has a Hamming-weight of
    exactly \(k\), that is, $\sigvec{x}\vposition{S}$ contains exactly $k$
    entries equal to $1$.

    As \(L\) decomposes into the different sets \(L_k\), we obtain $|L| = \sum_{k = 0}^{|S|} |L_k|$.
    Hence, it suffices to show that, for each \(k \in \{0,\dots,|S|\}\), we have
    \begin{equation}\label{eq:claim:upper-bound-1}
    |L_k| \leq
    (\allMax + 1)^{n - |S|} \cdot \binom{|S|}{k} \cdot
    \left\lceil\frac{\sigMax + 1}{\mname}\right\rceil^{k} \cdot \left\lceil
        \frac{\rhoMax +  1}{\mname}\right\rceil^{|S|-k}.
    \end{equation}

    We proceed to argue that Inequality~(\ref{eq:claim:upper-bound-1}) does indeed hold.
    To that end, fix a \(k \in \{0,\dots,|S|\}\), and
    observe that a string $x \in L$ (and hence, \(x \in L_k\))
    is uniquely determined by its $\sigma$-vector
    $\sigvec{x}$ and its weight vector $\degvec{x}$ (where elements are not taken modulo
    $\mname$).
    Note that as \(L \times L \subseteq \CR^n\),
    not all pairs of \(\sigma\)-vectors and weight-vectors correspond to a string in \(L_k\).
    Hence, we write
    \[
        |L_k| \le
        \sum_{\vec{s} \in \sigvec{L_k}}
        |\{ \degvec{x} \mid x \in L_k  \text{ and } \sigvec{x} = \vec{s}  \}|.
    \]

    Now, as \(S\) is \(\sigma\)-defining for \(\sigvec{L}\) (and hence, for \(\sigvec{L_k}
    \subseteq \sigvec{L}\)),
    for each \(\sigma\)-vector for \(L_{k}\)
    when restricted to the positions \(S\),
    there is exactly
    one \(\sigma\)-vector for \(L_{k}\) (on all positions).
    In particular,
    the number of different \(\sigma\)-vectors of \(L_k\) is equal to the number of
    different \(\sigma\)-vectors on the positions~\(S\). Further, by construction of
    \(L_k\), all \(\sigma\)-vectors on \(S\) have Hamming-weight exactly \(k\).
    Hence, we obtain\[
        |\sigvec{L_k}| =
        |\{ \sigvec{x} \mid x \in L_k \}| = |\{ \sigvec{x}\vposition{S} \mid x \in L_k \}|
        \le \binom{|S|}{k}.
    \]

    Now, fix a \(\sigma\)-vector \(\vec{s} \in \sigvec{L_k}\), and write
    \(L_{k,\vec{s}} \coloneqq \{x \in L_k \mid \sigvec{x} = \vec{s} \}\) for all strings in \(L_k\)
    with \(\sigma\)-vector \(\vec{s}\). By \cref{la:partition-positions},
    for each \(\mname\)-weight vector for \(L_{k,\vec{s}}\)
    when restricted to the positions \(\bar{S} \coloneqq \fragment{1}{n} \setminus S\),
    there is exactly
    one \(\mname\)-weight vector for \(L_{k,\vec{s}}\) (on all positions):
    \begin{equation}\label{eq:claim:upper-bound-2}
        |\{ \mdegvec{x} \mid x \in L_{k,\vec{s}} \text{ and }
        \mdegvec{x}\vposition{\bar{S}} = u \}| = 1.
    \end{equation}
    Hence, it remains to count weight-vectors instead of \(\mname\)-weight vectors.
    To that end, for each possible weight-vector on \(\bar{S}\), we count the possible
    extensions into a weight-vector on all positions.
    Writing \(u_{\mname}\) for the \(\mname\)-weight vector corresponding to a weight
    vector \(u\), we obtain
    \begin{align*}
        |\degvec{L_{k,\vec{s}}}| &
        \le
        \sum_{u \in \{ \degvec{x}\vposition{\bar{S}} \mid x \in L_{k,\vec{s}} \}}
        |\{\degvec{x}\vposition{S} \mid x \in L_{k,\vec{s}} \text{ and }
        \degvec{x}\vposition{\bar{S}} = u\}|\\
                                             &\le
                                             \sum_{u \in \{ \degvec{x}\vposition{\bar{S}} \mid x \in L_{k,\vec{s}} \}}
        |\{\degvec{x}\vposition{S} \mid x \in L_{k,\vec{s}} \text{ and }
        \mdegvec{x}\vposition{\bar{S}} = u_{\mname}\}|.
    \end{align*}

    Finally, we bound \(|\{\degvec{x}\vposition{S} \mid x \in L_{k,\vec{s}} \text{ and }
    \mdegvec{x}\vposition{\bar{S}} = u_{\mname}\}|\). To that end, observe that by
    \cref{eq:claim:upper-bound-2}, the corresponding \(\mname\)-weight vector is unique.
    Hence, we need to bound only the number of different weight vectors (on \(S\))
    that result in the same \(\mname\)-weight vector (on \(S\)). By construction, on the
    positions \(S\), the string \(x\) contains exactly \(k\) characters
    \(\sigma_{\star}\)---for each such position,
    there are at most \(\lceil (\sigMax + 1)
    / \mname \rceil\) different characters having the same \(\mname\)-weight vector;
    for each of the remaining \(|S| - k\) positions,
    there are at most \(\lceil (\rhoMax + 1)
    / \mname \rceil\) different characters having the same \(\mname\)-weight vector.
    This yields
    \begin{align*}
        &|\{\degvec{x}\vposition{S} \mid x \in L_{k,\vec{s}} \text{ and }
        \mdegvec{x}\vposition{\bar{S}} = u_{\mname}\}|
        \le
        \left\lceil\frac{\sigMax + 1}{\mname}\right\rceil^{k}
        \left\lceil\frac{\rhoMax + 1}{\mname}\right\rceil^{|S|-k}.
    \end{align*}

    Combining the previous steps with the final observation that
    \[|\{ \degvec{x}\vposition{\bar{S}} \mid x \in L_{k,\vec{s}} \}| \le (\allMax + 1)^{|\bar{S}|}
    = (\allMax + 1)^{n - |S|},\]
    we obtain the claimed Inequality~(\ref{eq:claim:upper-bound-1}):
    \begin{align*}
        |L_k| &\le
            \sum_{\vec{s} \in \sigvec{L_k}}
            |\degvec{L_{k,\vec{s}}}|\\
              &\le \binom{|S|}{k} \cdot
              \sum_{u \in \{ \degvec{x}\vposition{\bar{S}} \mid x \in L_{k,\vec{s}} \}}
        |\{\degvec{x}\vposition{S} \mid x \in L_{k,\vec{s}} \text{ and }
        \mdegvec{x}\vposition{\bar{S}} = u_{\mname}\}|\\
              &\le (\allMax + 1)^{n - |S|} \cdot \binom{|S|}{k} \cdot
    \left\lceil\frac{\sigMax + 1}{\mname}\right\rceil^{k} \cdot \left\lceil
        \frac{\rhoMax +  1}{\mname}\right\rceil^{|S|-k}.
    \end{align*}
    Overall, this yields the desired bound.
\end{proof}

In a final step before proving \cref{thm:upper-bound}, we tidy up the unwieldy upper bound
from \cref{claim:upper-bound}.

\begin{lemma}\label{lem:upper-bound-nicified}
    For any non-negative integers \(n\) and \(a \in \fragment0n\), we have
    \begin{align*}
        &{(\allMax + 1)^{n - a} \cdot \sum_{k = 0}^{a} \binom{a}{k}
    \left\lceil\frac{\sigMax + 1}{\mname}\right\rceil^{k} \left\lceil\frac{\rhoMax +
    1}{\mname}\right\rceil^{a-k}}\\
        &\qquad\le\begin{cases}
                   (\allMax + 2)^{n} &\text{if \(\mname = 2\) and \(\allMax = \sigMax = \rhoMax\) is even,}\\
                   (\allMax + 1)^{n} &\text{otherwise.}
                  \end{cases}
    \end{align*}
\end{lemma}
\begin{proof}
    In a first step, applying $\allMax = \max\{\sigMax,\rhoMax\}$
    and the Binomial Theorem yields
    \begin{align*}
        & (\allMax + 1)^{n - a} \cdot \sum_{k = 0}^{a} \binom{a}{k}
        \left\lceil\frac{\sigMax + 1}{\mname}\right\rceil^{k}
        \left\lceil\frac{\rhoMax + 1}{\mname}\right\rceil^{a-k}\\
            &\quad\leq (\allMax + 1)^{n - a} \cdot \sum_{k = 0}^{a} \binom{a}{k}
            \left\lceil\frac{\allMax + 1}{\mname}\right\rceil^{k}
            \left\lceil\frac{\allMax + 1}{\mname}\right\rceil^{a-k}\\
            &\quad=
            (\allMax + 1)^{n - a} \cdot \left(2 \cdot
            \left\lceil\frac{\allMax + 1}{\mname}\right\rceil\right)^{a}.
    \end{align*}

    In a next step, we investigate the term \(\lceil({\allMax + 1})/{\mname}\rceil\).

    First, if $\mname \geq 3$ or if \(\mname = 2\) and \(\allMax\) is odd, we have
    \[2 \cdot \left\lceil\frac{\allMax + 1}{\mname}\right\rceil \leq \allMax + 1.\]
    Hence, in these cases, we directly obtain
    \[(\allMax + 1)^{n-a} \cdot
    \left(2 \cdot \left\lceil\frac{\allMax + 1}{\mname}\right\rceil\right)^{a}
    \leq (\allMax + 1)^{n}.\]

    Next, if $\mname = 2$ and $\allMax = \sigMax = \rhoMax$ is even, we have
    \[2 \cdot \left\lceil\frac{\allMax + 1}{\mname}\right\rceil = \allMax + 2.\]
    Hence, in this case, we directly obtain
    \[(\allMax + 1)^{n - a} \cdot
    \left(2 \cdot \left\lceil\frac{\allMax + 1}{\mname}\right\rceil\right)^{a}
    \le (\allMax + 2)^{n}.\]

    Finally, if $\mname = 2$, $\allMax$ is even,
    and $\min\{\rhoMax,\sigMax\} < \allMax$, we
    need a more careful analysis. Restarting from the initial term, we apply the Binomial
    Theorem and obtain
    \begin{align*}
            &(\allMax + 1)^{n - a} \cdot \sum_{k = 0}^{a} \binom{a}{k}
            \left\lceil\frac{\sigMax + 1}{\mname}\right\rceil^{k}
            \left\lceil\frac{\rhoMax + 1}{\mname}\right\rceil^{a-k}\\
            &\quad=    (\allMax + 1)^{n - a} \cdot
            \left(\left\lceil\frac{\sigMax + 1}{\mname}\right\rceil
            + \left\lceil\frac{\rhoMax + 1}{\mname}\right\rceil\right)^{a}\\
            &\quad\leq (\allMax + 1)^{n-a}
            \cdot \left(\left\lceil\frac{\allMax}{\mname}\right\rceil
            + \left\lceil\frac{\allMax + 1}{\mname}\right\rceil\right)^{a}\\
            &\quad=    (\allMax + 1)^{n-a} \cdot \left(\frac{\allMax}{2} + \frac{\allMax +
            2}{2}\right)^{a}\\
            &\quad\leq (\allMax + 1)^{n}.
    \end{align*}

    This completes the proof.
\end{proof}

Finally, combining \cref{claim:upper-bound,lem:upper-bound-nicified} directly yields
\cref{thm:upper-bound}, which we restate here for convenience.

\upperboundmain*
\begin{proof}
    Let $L \subseteq \allStates^n$ denote a language with $L \times L \subseteq \CR^n$,
    and let \(S\) denote a \(\sigma\)-defining set for~\(L\).
    By \cref{claim:upper-bound}, we obtain
    \[|L| \leq
        \smash{(\allMax + 1)^{n - |S|} \cdot \sum_{k = 0}^{|S|} \binom{|S|}{k}
    \left\lceil\frac{\sigMax + 1}{\mname}\right\rceil^{k} \left\lceil\frac{\rhoMax +
    1}{\mname}\right\rceil^{|S|-k}}. \raisebox{-2ex}{\null} \]
    Applying \cref{lem:upper-bound-nicified} yields the claim.
\end{proof}

\subsubsection{Bounding the Size of Combinations of Realized Languages}\label{sec:combinationcompatible}

Having understood a single realized language, we turn to \emph{combinations} of
realized languages next.

\begin{definition}\label{def:comb}
    For two strings $x,y \in \allStates^n$, we define their \emph{combination} as the
    string $x \oplus y \in (\allStates \cup \{\perp\})^n$ obtained via
    \[(x \oplus y)\position{\ell} \coloneqq \begin{cases}
                                             \sigma_k &\text{if } x\position{\ell} = \sigma_i\text{ and } y\position{\ell} = \sigma_j \text{ and } i+j = k \leq \sigMax,\\
                                             \rho_k   &\text{if } x\position{\ell} = \rho_i\text{ and }y\position{\ell} = \rho_j\text{ and } i+j = k \leq \rhoMax,\\
                                             \bot     &\text{otherwise.}
                                            \end{cases}\]
    We say that $x$ and $y$ \emph{can be joined} if, for each position \(\ell \in
    \fragment1n\), we have $(x \oplus y)\position{\ell} \neq \bot$.

    For two languages $L_1, L_2 \subseteq \allStates^n$, we define their \emph{combination}
    as the set of all combinations of strings that can be joined:
    \[L_1 \oplus L_2 \coloneqq \{x \oplus y \mid x \in L_1 \text{ and } y \in L_2 \text{ such that } x,y \text{ can be joined}\}.\]
\end{definition}

Observe that, for strings \(x \in L_1\) and \(y \in L_2\), their combination
\(x \oplus y\) is in \(L_1 \oplus L_2\) if and only if \(x\) and \(y\)
share a common \(\sigma\)-vector and the sum of their
weight-vectors does not ``overflow'', that is, we have\footnote{For ease of notation, we use
``\(\le\)'' component-wise on vectors.}
\begin{align}\label{eq:oplus-decomp}
    L_1 \oplus L_2 =
    \{
        x \oplus y \mid x \in L_1,~y \in L_2,~\sigvec{x} = \sigvec{y}, \text{ and } \degvec{x} + \degvec{y} \le
        \capvec{\sigvec{x}}
    \}.
\end{align}
Finally, observe that, for strings \(x \in L_1\) and \(y \in L_2\) that can be joined,
we have
\begin{equation}
      \label{eq:oplus-join}
\degvec{x} + \degvec{y} = \degvec{x \oplus y}.
\end{equation}

We use the remainder of this section to show that
\cref{claim:upper-bound,thm:upper-bound} easily lift to the combinations of realized
languages. Thereby, we show that the number of partial solutions does not significantly
increase by combining realized languages.
We start with a small collection of useful properties of combinations of
realized languages. First, we discuss how \(\sigma\)-vectors behave under combinations of
languages.

\begin{lemma}\label{lem:comb-1}
    Let $L_1, L_2 \subseteq \allStates^n$ denote languages with $L_1 \times L_1 \subseteq \CR^n$
    and \(L_2 \times L_2 \subseteq \CR^n\).

    Then, \((L_1 \oplus L_2)\) has the \(\sigma\)-vectors that appear for both \(L_1\) and
    \(L_2\), that is,
    \[\sigvec{L_1 \oplus L_2} \subseteq \sigvec{L_1} \cap \sigvec{L_2}.\]
\end{lemma}
\begin{proof}
    The proof follows immediately from \cref{def:comb}.
    Indeed, no string of the language \(L_1 \oplus L_2\) has a
    \(\bot\) character, as \((L_1 \oplus L_2)\) contains only combinations of strings with
    the same \(\sigma\)-vectors.
    Hence, the strings in \(L_1 \oplus L_2\) may differ from strings in
    \(L_1\) or \(L_2\) only in their weight vectors, and \(\sigvec{L_1 \oplus L_2}\) has only
    those \(\sigma\)-vectors that appear in both \(\sigvec{L_1}\) and \(\sigvec{L_2}\).
    Furthermore, note that due to ``overflows'' in the weight-vectors, a \(\sigma\)-vector in
    \(\sigvec{L_1} \cap \sigvec{L_2}\) might not be in \(\sigvec{L_1 \oplus L_2}\).
\end{proof}

Next, we show that having the relation \(\CR^n\) transfers as well.

\begin{lemma}\label{lem:comb-rel}
    Let $L_1, L_2 \subseteq \allStates^n$ denote languages with $L_1 \times L_1 \subseteq \CR^n$
    and \(L_2 \times L_2 \subseteq \CR^n\).

    Then, we have $(L_1 \oplus L_2) \times (L_1 \oplus L_2) \subseteq \CR^n$.
\end{lemma}
\begin{proof}
    Fix strings $x,y \in L_1 \oplus L_2$.
    By \cref{def:comb}, this means that there are strings $x_1,y_1 \in L_1$ and $x_2,y_2 \in L_2$ such that $x = x_1 \oplus x_2$ and $y = y_1 \oplus y_2$.
    Expanding the definition of \(\oplus\) yields
    \begin{align*}
        \sigvec{x} \cdot \mdegvec{y} &= \sigvec{x} \cdot (\mdegvec{y_1} + \mdegvec{y_2})
        = \sigvec{x} \cdot \mdegvec{y_1} + \sigvec{x} \cdot \mdegvec{y_2}.
        \intertext{As \(\oplus\) does not change \(\sigma\)-vectors for strings that can
        be joined, we obtain}
        \sigvec{x} \cdot \mdegvec{y}
                                     &= \sigvec{x_1} \cdot \mdegvec{y_1} + \sigvec{x_2} \cdot \mdegvec{y_2}
        \intertext{Next, we use \((x_1, y_1) \in \CR^n\) and \((x_2, y_2) \in \CR^n\) to obtain}
        \sigvec{x} \cdot \mdegvec{y}
                                     &\equiv_{\mname} \sigvec{y_1} \cdot \mdegvec{x_1} + \sigvec{y_2}
                                     \cdot \mdegvec{x_2}\\
        \intertext{Again, as \(\oplus\) does not change \(\sigma\)-vectors for strings that can
        be joined, we obtain}
        \sigvec{x} \cdot \mdegvec{y}
                                     &\equiv_{\mname} \sigvec{y} \cdot \mdegvec{x_1} + \sigvec{y}
                                     \cdot \mdegvec{x_2}\\
                                     &= \sigvec{y} \cdot (\mdegvec{x_1} + \mdegvec{x_2})\\
                                     &= \sigvec{y} \cdot \mdegvec{x},
    \end{align*} which completes the proof that $(x,y) \in \CR^n$.
\end{proof}

Now, we directly obtain that \cref{claim:upper-bound} lifts:

\begin{corollary}\label{lem:comb-upper}
    Let $L_1, L_2 \subseteq \allStates^n$ denote languages with $L_1 \times L_1 \subseteq \CR^n$
    and \(L_2 \times L_2 \subseteq \CR^n\).
    Further, let \(S\) denote a \(\sigma\)-defining set for \(\sigvec{L_1 \oplus L_2}\).
    Then, we have
    \[|L_1 \oplus L_2| \leq
        {(\allMax + 1)^{n - |S|} \cdot \sum_{k = 0}^{|S|} \binom{|S|}{k}
    \left\lceil\frac{\sigMax + 1}{\mname}\right\rceil^{k} \left\lceil\frac{\rhoMax +
    1}{\mname}\right\rceil^{|S|-k}}. \raisebox{-2ex}{\null} \]
\end{corollary}
\begin{proof}
    By \cref{lem:comb-rel}, we can use \cref{claim:upper-bound} on the language \(L_1
    \oplus L_2\).
\end{proof}

\subsection{Exploiting Structure: Fast Join Operations}

Recall that the bound in \cref{thm:upper-bound} yields an upper bound on the number of
partial solutions for a graph \(G\) and a subset \(U\) of its vertices. Recall further
that, in the end, we intend to use an algorithm based on the dynamic programming on a tree
decomposition paradigm. Hence, we need to be able to efficiently compute possible partial
solutions for a graph given the already computed partial solutions for some of its subgraphs. We
tackle this task next. In particular, we show how to generalize known convolution
techniques to compute the combination of realized languages:

\begin{restatable}{theorem}{joinalgo}
    \label{thm:join-algorithm-proto}
    \label{thm:join-algorithm}
    Let $L_1, L_2 \subseteq \allStates^n$ denote languages with $L_1 \times L_1 \subseteq \CR^n$
    and \(L_2 \times L_2 \subseteq \CR^{n}\), and let \(S\) denote
    a \(\sigma\)-defining set for \(\sigvec{L_1} \cap \sigvec{L_2}\).
    Then, we can compute the language $L_1 \oplus L_2$ in time
    \begin{align*}
        &n^{\Oh(1)} \cdot 2^{(\allMax)^{\Oh(1)}} \cdot
        {(\allMax + 1)^{n - |S|} \cdot \sum_{k = 0}^{|S|} \binom{|S|}{k}
    \left\lceil\frac{\sigMax + 1}{\mname}\right\rceil^{k} \left\lceil\frac{\rhoMax +
1}{\mname}\right\rceil^{|S|-k}}\\
        &\qquad\le\begin{cases}
                   n^{\Oh(1)} \cdot 2^{(\allMax)^{\Oh(1)}} \cdot(\allMax + 2)^{n} & \text{if \(\mname = 2\) and \(\allMax = \sigMax = \rhoMax\) is even,}\\
                   n^{\Oh(1)} \cdot 2^{(\allMax)^{\Oh(1)}} \cdot(\allMax + 1)^{n} & \text{otherwise.}
                  \end{cases}
    \end{align*}
\end{restatable}

Toward proving \cref{thm:join-algorithm}, first recall that strings \(x_1 \in L_1\) and
\(x_2 \in L_2\) decompose into a \(\sigma\)-vector and a weight-vector each. Further,
recall from \cref{eq:oplus-decomp}
that \(x_1 \oplus x_2\) is in \(L_1 \oplus L_2\) if and only if \(x_1\) and \(x_2\)
share a common \(\sigma\)-vector and the sum of their
weight-vectors does not ``overflow''.
This observation yields the following proof strategy. For each different \mbox{\(\sigma\)-vector} \(\vec{s} \in \sigvec{L_1}\cap \sigvec{L_2}\), we compute all possible sums of the weight-vectors for strings with
\mbox{\(\sigma\)-vector}~\(\vec{s}\). Afterward, we filter out resulting vectors where an overflow
occurred. To implement this strategy, we intend to make use of the tools developed by van Rooij~\cite{Rooij20}; in particular, the following result.

\begin{theorem}[{\cite[Lemma 3]{Rooij20}}]
    \label{thm:alg-convolution}
    For integers $d_1,\dots,d_n$ and
    $D \coloneqq \prod_{i=1}^nd_i$,
    let~$p$ denote a prime such that in the field \(\FF_p\),
    the $d_i$-th root of unity exists for each $i \in \fragment1n$.
    Further, for two functions
    $f,g\colon \ZZ_{d_1} \times \dots \times \ZZ_{d_n} \rightarrow \FF_p$,
    let $h\colon\ZZ_{d_1} \times \dots \times \ZZ_{d_n} \rightarrow \FF_p$ denote
    the convolution
    \[h(a) \coloneqq \sum_{a_1 + a_2 = a} f(a_1) \cdot g(a_2).\]
    Then, we can compute the function $h$ in $\Oh(D \log D)$ many arithmetic operations
    (assuming a $d_i$-th root of unity $\omega_i$ is given for all $i \in \nset{n}$).
\end{theorem}

Before we continue, let us briefly comment on how to find an appropriate prime $p$, as well as the roots of unity $\omega_i$.

\begin{remark}
    \label{rem:find-prime-alg}
    Suppose $M$ is a sufficiently large integer such that all images of the functions $f,g,h$ are in the range $\fragment0M$.
    In particular, suppose that $M \geq D$.
    Suppose $d_1',\dots,d_\ell'$ is the list of integers obtained from $d_1,\dots,d_n$ by
    removing duplicates (in all of our applications we ensure that $\ell \leq 4$).
    Let $D' \coloneqq \prod_{i=1}^\ell d_i'$.
    We consider candidate numbers $m_j \coloneqq 1 + D'j$ for all $j \geq 1$.
    By the Prime Number Theorem for Arithmetic Progressions \cite[Theorem 1.3]{BennettMOR18}, there is a prime $p$ such that
    \begin{enumerate}
        \item $p = m_j$ for some $j \geq 1$,
        \item $p > M$, and
        \item $p = \Oh\Big(\max\big\{\varphi(D')M \log M,\exp(D')\big\}\Big)$,
    \end{enumerate}
    where $\varphi$ denotes Euler's totient function.
    Such a number can be found in time
    \[\Oh\Big(p \big(\log p\big)^c\Big)\]
    for some absolute constant $c$ exploiting that prime testing can be done in polynomial time.

    Now, fix $i \in \nset{n}$ and fix $k_i \coloneqq {D'}j/{d_i}$.
    For every $x \in \FF_p$, we have that $x^{p-1} = 1$, and hence, $x^{k_i}$ is a $d_i$-th root of unity if and only if $(x^{k_i})^i \neq 1$ for all $i < d_i$.
    Hence, given an element $x \in \FF_p$, it can be checked in time
    \[\Oh\Big( d_i \cdot (\log p)^c\Big)\]
    whether $x^{k_i}$ is a $d_i$-th root of unity.
    Due to our choice of $p$, this test succeeds for at least one $x \in \FF_p$.
    Thus, a $d_i$-th root of unity $\omega_i$ for every $i \in \nset{n}$ can be found in time
    \[\Oh\Big(n \cdot p \cdot \max_{i \in \nset{n}}d_i \cdot (\log p)^c\Big).\]
\end{remark}

Now, let us return to the problem at hand.
The most naive approach, applying \cref{thm:alg-convolution} to the weight-vectors directly, is not fast enough for
our purposes: A single convolution already takes time  \(\widetilde{\Oh}((\allMax +
1)^n)\), and so, using such a convolution for each of the up to \(2^{|S|}\) different
\(\sigma\)-vectors is far too slow.
Instead of convolving weight-vectors directly, we hence turn to
\cref{la:partition-positions}: for a fixed \(\sigma\)-vector \(\vec{s}\), there are far less
(depending on the size of \(S\)) than \((\allMax + 1)^n\) different weight vectors. We can
exploit this by \emph{compressing} the weight-vectors to a smaller representation, and then
convolving the resulting compressed vectors.
Formally, we first make our intuition of ``exploiting'' \cref{la:partition-positions} more
formal by defining a useful auxiliary vector.

\begin{definition}
    Let $L \subseteq \allStates^n$ denote a (non-empty)
    language with $L \times L \subseteq \CR^n$,
    let \(S\) denote a \(\sigma\)-defining set for \(\sigvec{L}\),
    let \(\CW_{S}
    \subseteq \sigvec{L}\)  denote a corresponding set of witness vectors, and
    set \(\bar{S} \coloneqq \fragment1n \setminus S\).

    For two vectors \(u, o \in \fragment0{\allMax}^n\) and a position \(\ell\in S\),
    we define the \emph{remainder}
    \(\remvec{\CW_{S}}{u}{o}\) at \(\ell\) as
    \[
        \remvec{\CW_{S}}{u}{o}\vposition{\ell}
        \coloneqq \sum_{i \in \bar{S}} \big(u\vposition{i} -
        o\vposition{i}\big)\cdot\big(w_{1,\ell}\vposition{i} -
    w_{0,\ell}\vposition{i}\big).
    \]
\end{definition}

Slightly abusing notation, if $u$ is a longer $(n+d)$-dimensional vector for some $d \geq 1$, we also write $\remvec{\CW_{S}}{u}{o}$ to denote the vector $\remvec{\CW_{S}}{u\fragment{1}{n}}{o}$.

\begin{remark}\label{rem:rem}
    Observe that, if we restrict \(u\) and \(o\) to be the weight-vectors of strings in
    \(L\) with a common \(\sigma\)-vector \(\vec{s} \in \sigvec{L}\), that is,
    \(u, o \in \{ \degvec{x} \mid x \in L \text{ and } \sigvec{x} = \vec{s}\}\), then,
    for any \(\ell \in S\), \cref{la:inner-product-zero} yields
     \begin{align*}
        {u}\vposition{\ell} - o\vposition{\ell}
        + \remvec{\CW_S}{u}{o}\vposition{\ell}
        &={u}\vposition{\ell} - o\vposition{\ell}
         + \sum_{i \in\fragment1n\setminus S}
         \big({u}\vposition{i} - o\vposition{i}\big)
         \cdot \big({w_{1,\ell}}\vposition{i} - {w_{0,\ell}}\vposition{i}\big)\\
        &=
        \big({u} - o\big) \cdot \big({w_{1,\ell}} - {w_{0,\ell}}\big)
        =
        \big({u} - o\big) \cdot {w_{1,\ell}} -
        \big({u} - o\big) \cdot {w_{0,\ell}}\\
        &\equiv_{\mname} 0 - 0 \equiv_{\mname} 0.
    \end{align*}
\end{remark}
Note that \cref{rem:rem} closely mirrors \cref{cl:partition-1}.
Further, if we pick an arbitrary vector \(o \in \{ \degvec{x} \mid x \in L \text{ and
} \sigvec{x} = \vec{s}\}\) to act as an ``origin'', then we can shift all vectors
in \(\{ \degvec{x} \mid x \in L \text{ and } \sigvec{x} = \vec{s}\}\) so that their
coordinates on \(S\) become divisible by \(\mname\). We can then exploit this to
compress the coordinates on \(S\). In other words, the reduced flexibility due to
fixing the \(\sigma\)-vector \(\vec{s}\) translates to a reduced flexibility in the choice of
coordinates for the positions that define the possible weight-vectors.

Finally, as we intend to add (the components of) compressed vectors modulo some number
$d_i$ (see \cref{thm:alg-convolution}), we need
to add ``checksums'' to the compressed vectors to be able to detect ``overflows''.
The simplest way to implement such checksums would be to add a single coordinate to our vectors that contains the sum of all entries.
However, in the notation of \cref{thm:alg-convolution}, this would mean that $\max_{i} d_i = 2n(\allMax + 1)$ which is too expensive for the application of \cref{rem:find-prime-alg}.
Instead, we use a binary representation for the checksums and add numbers modulo $3$ to avoid overflows in the checksum coordinates.
Also, we use two ``checksums'', one for the coordinates contained in $S$ (where $S$ is a $\sigma$-defining set) and another for the coordinates in $\bar{S}$.

For a positive integer $n \in \NN$ and a position $i \geq 1$, we define $\bit_i(n) \in
\{0,1\}$ to denote the $i$-th bit in the binary representation of $n$, that is, $n = \sum_{i \geq 1}\bit_i(n) \cdot 2^{i-1}$.
Further, we define $d \coloneqq \lceil\log(2n(\allMax + 1))\rceil$ to be the number of bits required to represent numbers strictly less than $2n(\allMax + 1)$.

\begin{definition}\label{def:comp}
    Let $L \subseteq \allStates^n$ denote a (non-empty)
    language with $L \times L \subseteq \CR^n$,
    let \(S\) denote a \(\sigma\)-defining set for \(\sigvec{L}\),
    let \(\CW_{S}
    \subseteq \sigvec{L}\)  denote a corresponding set of witness vectors, and
    set \(\bar{S} \coloneqq \fragment1n \setminus S\).

    Further, fix a vector \(\vec{s} \in \sigvec{L}\) and an origin vector
    \(o \in \fragment0{\allMax}^n\) such that, for any
    \(u \in \{ \degvec{x} \mid x \in L \text{ and } \sigvec{x} = \vec{s}\}\), we have
     \begin{align*}
        {u}\vposition{\ell} - o\vposition{\ell}
        + \remvec{\CW_S}{u}{o}\vposition{\ell}
         \equiv_{\mname} 0
    \end{align*}
    for all \(\ell \in S\).

    For a (weight-)vector \(z \in \{ \degvec{x} \mid x \in L \text{ and } \sigvec{x} = \vec{s}\}\),
    we define the \emph{\(\sigma\)-compression with
    origin \(o\) and type \(\vec{s}\)} as the following \((n+2d)\)-dimensional vector \(\comp{o}{z}\):
    \begin{alignat*}{5}
        \comp{o}z\vposition{\ell}  &\coloneqq
        {z}\vposition{\ell} &&\mod \allMax + 1, &&\ell \in  \bar{S},\\
        \comp{o}z\vposition{\ell} &\coloneqq
        \frac{
        {z}\vposition{\ell} - o\vposition{\ell}
        + \remvec{\CW_S}{z}{o}\vposition{\ell}
        }{\mname} &&\mod \left\lceil\frac{\capvec{\vec{s}}\vposition{\ell} + 1}{\mname} \right\rceil,
                                                   &\quad&\ell \in S,\\
        \comp{o}z\vposition{n + \ell} &\coloneqq
        \bit_{\ell}\Big(\sum_{i \in \bar{S}} {z}\vposition{i}\Big) &&\mod 3, &\quad&\ell \in \fragment{1}{d},\\
        \comp{o}z\vposition{n + d + \ell } &\coloneqq
        \bit_{\ell}\Big(\sum_{i \in S}{z}\vposition{i}\Big) &&\mod 3, &\quad&\ell \in \fragment{1}{d}.
    \end{alignat*}
    Further, we write
    \[\CZ_{S,\vec{s}} \coloneqq \bigtimes_{\ell \in \bar{S}} \ZZ_{\allMax + 1} \quad\times\quad \bigtimes_{\ell \in S} \ZZ_{\left\lceil\frac{\capvec{\vec{s}}\vposition{\ell} + 1}{\mname} \right\rceil} \quad\times\quad \bigtimes_{\ell \in \fragment{1}{2d}} \ZZ_3\]
    for the \((n + 2d)\)-dimensional space of all possible \(\sigma\)-compressed vectors for \(S\) and \(\vec{s}\) (and potentially different \(o\)).
\end{definition}

We stress each entry of \(\comp{o}{z}\) is defined as an element from $\ZZ_{d_\ell}$ for the appropriate dimension $d_\ell$.

\begin{remark}\label{rem:comp-imp}
    With \cref{thm:alg-convolution} in mind and writing
    \(S_{\vec{s},c} \coloneqq \{ \ell \in S \mid \vec{s}\vposition{\ell} = c\}\), we observe that
    \begin{align*}
        |\CZ_{S,\vec{s}}| &=
        {\left\lceil\frac{\sigMax + 1}{\mname}\right\rceil}^{|S_{\vec{s},1}|} \cdot
        {\left\lceil\frac{\rhoMax + 1}{\mname}\right\rceil}^{|S_{\vec{s},0}|} \cdot
        (\allMax + 1)^{n - |S|}
        \cdot 3^{2d}\\
        &\leq
        {\left\lceil\frac{\sigMax + 1}{\mname}\right\rceil}^{|S_{\vec{s},1}|} \cdot
        {\left\lceil\frac{\rhoMax + 1}{\mname}\right\rceil}^{|S_{\vec{s},0}|} \cdot
        (\allMax + 1)^{n - |S|}
        \cdot 256 \cdot n^4(\allMax + 1)^4.
    \end{align*} In particular, using \cref{thm:alg-convolution} on the \(\sigma\)-compressed
    vectors yields a significant speed-up over the direct application to the weight
    vectors (whose domain has a size of \((\allMax + 1)^n\)).\qed
\end{remark}
\begin{remark}\label{rem:decompress}
    Observe that, for a fixed origin vector \(o\), the mapping \(\comp{o}{\star}\) is
    injective and we can easily recover the original weight-vector \(z\) from its
    \(\sigma\)-compression \(\comp{o}{z}\):\footnote{Observe that we exploit
        \(\remvec{\CW_S}{z}{o}\vposition{\ell} = \remvec{\CW_S}{\comp{o}{z}}{o}\vposition{\ell}\).}
    \begin{alignat*}{5}
        {z}\vposition{\ell}
        &\coloneqq
        \comp{o}z\vposition{\ell},
        &&\ell \in  \bar{S},\\
        {z}\vposition{\ell}
        &\coloneqq
        \big( \mname \cdot \comp{o}z\vposition{\ell}
        + o\vposition{\ell}
    - \remvec{\CW_S}{\comp{o}{z}}{o}\vposition{\ell} \big)
        \bmod{
\mname \left\lceil\frac{\capvec{\vec{s}}\vposition{\ell} + 1}{\mname} \right\rceil
        }\\
        &\;=
        \Big(\mname \cdot \comp{o}z\vposition{\ell}
        + o\vposition{\ell}\\
        &\qquad- \sum_{i \in\bar{S}}
        \big(\comp{o}{z}\vposition{i} - o\vposition{i}\big)
    \cdot \big({w_{1,\ell}}\vposition{i} - {w_{0,\ell}}\vposition{i}\big)\Big)
        \bmod{
\mname \left\lceil\frac{\capvec{\vec{s}}\vposition{\ell} + 1}{\mname} \right\rceil
        },
                                                   &\quad&\ell \in S.
    \end{alignat*}
    Further, for elements \(x \in \CZ_{S,\vec{s}}\) that cannot be obtained from a
    \(\sigma\)-compression, we have that
    \begin{itemize}
     \item $x\vposition{n + \ell} \notin \{0,1\}$ for some $\ell \in \fragment{1}{2d}$, or
     \item $\sum_{i \in \bar{S}} x\vposition{i} \ne \sum_{\ell \in \fragment1d}x\vposition{n + \ell}\cdot 2^{\ell-1}$, or
     \item $\sum_{i \in S} x\vposition{i} \ne \sum_{\ell \in \fragment1d}x\vposition{n + d + \ell}\cdot 2^{\ell-1}$, or
     \item $x\vposition{\ell} > \capvec{\vec{s}}\vposition{\ell}$ for some \(\ell\in\fragment1n\).
    \end{itemize}
    Hence, given a subset of \(\CZ_{S,\vec{s}}\), we can quickly identify which vectors are
    indeed \(\sigma\)-compressed weight vectors.\qed
\end{remark}

In a next step, we discuss how addition and \(\sigma\)-compression interact with each other.
Toward this end, we define an equivalence relation on the set $\CZ_{S,\vec{s}}$.
Intuitively speaking, two vectors from $\CZ_{S,\vec{s}}$ are equivalent if they are identical on the first $n$ coordinates, and both checksums are identical, but with possibly different representations.
For example, the last $d = 4$ coordinates may contain $(0,0,0,1)$ to present the checksum $8 = 0\cdot 2^0 + 0\cdot 2^1 + 0\cdot 2^2 + 1\cdot 2^3$.
However, since these coordinates may contain entries from $\{0,1,2\}$, an alternative representation is $(0,0,2,0)$ since $8 = 0\cdot 2^0 + 0\cdot 2^1 + 2\cdot 2^2 + 0\cdot 2^3$.

\begin{definition}
    For $x,y \in \CZ_{S,\vec{s}}$, we write $x \simeq y$ if
    \begin{itemize}
        \item $x\vposition{\ell} = y\vposition{\ell}$ for all \(\ell\in\fragment1n\),
        \item $\sum_{\ell \in \fragment1d}x\vposition{n + \ell}\cdot 2^{\ell-1} = \sum_{\ell \in \fragment1d}y\vposition{n + \ell}\cdot 2^{\ell-1}$, and
        \item $\sum_{\ell \in \fragment1d}x\vposition{n + d + \ell}\cdot 2^{\ell-1} = \sum_{\ell \in \fragment1d}y\vposition{n + d + \ell}\cdot 2^{\ell-1}$.
    \end{itemize}
\end{definition}

\begin{lemma}\label{lem:comp-add}
    Let $L_1, L_2 \subseteq \allStates^n$ denote (non-empty)
    languages with $L_1 \times L_1 \subseteq \CR^n$ and \(L_2 \times L_2 \subseteq \CR^n\),
    let \(S\) denote a \(\sigma\)-defining set for \(\sigvec{L_1} \cap \sigvec{L_2}\),
    let \(\CW_{S}
    \subseteq \sigvec{L_1} \cap \sigvec{L_2}\)  denote a corresponding set of witness vectors, and
    set \(\bar{S} \coloneqq \fragment1n \setminus S\).

    Further, fix a \(\sigma\)-vector \(\vec{s} \in \sigvec{L_1} \cap \sigvec{L_2}\), as well as
    weight-vectors \[
        o \in \{ \degvec{x} \mid x \in L_1 \text{ and } \sigvec{x} = \vec{s}\}
    \quad\text{and}\quad p \in \{ \degvec{x} \mid x \in L_2 \text{ and } \sigvec{x} = \vec{s}\}.\]

    For any two strings \(u \in L_1, v \in L_2\) with \(\sigvec{u} = \sigvec{v} =
    \vec{s}\), we have that
    \(u\) and \(v\) can be joined if and only if there is a string \(z \in L_1 \oplus
    L_2\) with
    \begin{equation}\label{eq:comp-sum}
        \comp{o}{\degvec{u}} + \comp{p}{\degvec{v}}
        \simeq \comp{o + p}{\degvec{z}}.
    \end{equation}
    If \(z\) exists, we have \(z = u \oplus v\).
\end{lemma}

\begin{proof}
    Fix  two strings \(u \in L_1, v \in L_2\) with \(\sigvec{u} = \sigvec{v} = \vec{s}\), and
    suppose that they can be joined. \Cref{eq:oplus-join} now yields
    \(\degvec{u} + \degvec{v} = \degvec{u \oplus v}\) and, in particular, for each position
    \(\ell \in \fragment1n\), that\[
        \degvec{u}\vposition{\ell} + \degvec{v}\vposition{\ell} \le
        \capvec{\vec{s}}\vposition{\ell} \le \allMax.
    \]
    Now, for \cref{eq:comp-sum},
    only positions \(\ell \in S\) warrant a short justification; for all other positions
    the result is immediate from \cref{def:comp} by observing that ``overflows'' in the checksums cannot occur since those coordinates are taken modulo $3$.
    Hence, for a position \(\ell \in S\), first observe that we have
    \begin{align*}
        \remvec{\CW_S}{\degvec{u}}{o}\vposition{\ell} +
        \remvec{\CW_S}{\degvec{v}}{p}\vposition{\ell}
        &= \remvec{\CW_S}{\degvec{u} + \degvec{v}}{o + p}\vposition{\ell}\\
        &= \remvec{\CW_S}{\degvec{u \oplus v}}{o + p}\vposition{\ell}.
    \end{align*}
    Now, we obtain
    \begin{align*}
        0& \equiv_{\mname}
        \big({\degvec{u}}\vposition{\ell} - o\vposition{\ell}
        + \remvec{\CW_S}{\degvec{u}}{o}\vposition{\ell}\big)
        +
        \big({\degvec{v}}\vposition{\ell} - p\vposition{\ell}
        + \remvec{\CW_S}{\degvec{v}}{p}\vposition{\ell}\big)\\
        &=
        {\degvec{u \oplus v}}\vposition{\ell} - (o + p)\vposition{\ell}
        + \remvec{\CW_S}{\degvec{u \oplus v}}{o + p}\vposition{\ell},
    \end{align*} which yields the claim.

    For the other direction,
    fix  two strings \(u \in L_1, v \in L_2\) with \(\sigvec{u} = \sigvec{v} = \vec{s}\), and
    suppose that there is a string \(z \in L_1 \oplus L_2\)
    with \(\comp{o}{\degvec{u}} + \comp{p}{\degvec{v}} \simeq \comp{o + p}{\degvec{z}}\).
    We proceed to show that then, indeed, \(u\) and \(v\) can be joined and \(z = u \oplus v\).

    First, consider the positions in the set \(\bar{S}\), and in particular, fix an \(\ell \in
    \bar{S}\). Now, we have
    \[
        \degvec{z}\vposition{\ell}
        = \comp{o + p}{\degvec{z}}\vposition{\ell}
        \equiv_{\allMax + 1}
        \comp{o}{\degvec{u}}\vposition{\ell} + \comp{p}{\degvec{v}}\vposition{\ell}
        = \degvec{u}\vposition{\ell} + \degvec{v}\vposition{\ell}.
    \]
    In combination with
    \(0 \le \degvec{z}\vposition{\ell} \le \capvec{\vec{s}}\vposition{\ell} \le \allMax + 1\) and \(0 \le \degvec{u}\vposition{\ell}
        + \degvec{v}\vposition{\ell}\),
    we obtain
    \begin{equation}
        \label{eq:ineq-B-positions}
        \degvec{z}\vposition{\ell}
        \le
        \degvec{u}\vposition{\ell} + \degvec{v}\vposition{\ell}.
    \end{equation}
    Now, we exploit the ``checksums''.
    For $k \in \fragment1d$, we have
    \[\comp{o}{\degvec{u}}\vposition{n + k} + \comp{p}{\degvec{v}}\vposition{n + k} = (\comp{o}{\degvec{u}} + \comp{p}{\degvec{v}})\vposition{n + k}\]
    since those coordinates are taken modulo $3$ and both
    $\comp{o}{\degvec{u}}\vposition{n + k} \in \{0,1\}$ and
    $\comp{p}{\degvec{v}}\vposition{n + k} \in \{0,1\}$.
    Further, as \(u \in L_1\), \(v \in L_2\), and \(z \in L_1 \oplus L_2\), we have that, for
    all positions \(i \in \fragment1n\),
    \begin{align*}
        &0\le \degvec{u}\vposition{i} \le \capvec{\vec{s}}\vposition{i} \le \allMax,\\
        &0 \le\degvec{v}\vposition{i} \le \capvec{\vec{s}}\vposition{i} \le \allMax,\\
        \quad\text{and}\quad
        &0 \le \degvec{z}\vposition{i} \le \capvec{\vec{s}}\vposition{i} \le \allMax,
    \end{align*}
    and hence,
    \[
        0 \leq \sum_{i \in \bar{S}}\degvec{u}\vposition{i}
        + \sum_{i \in \bar{S}} \degvec{v}\vposition{i}  < 2n(\allMax+1)
        \;\,\text{and}\;\,
        0 \leq \sum_{i \in \bar{S}}\degvec{z}\vposition{i} < 2n(\allMax+1).
    \]
    Together, it follows that
    \begin{align*}
        \sum_{i \in \bar{S}}\degvec{z}\vposition{i}
        &= \sum_{k \in \fragment1d}\comp{o + p}{\degvec{z}}\vposition{n + k}\cdot 2^{k-1}
        \\
        &= \sum_{k \in \fragment1d}(\comp{o}{\degvec{u}} + \comp{p}{\degvec{v}})\vposition{n + k}\cdot 2^{k-1}
        \\
        &= \sum_{k \in \fragment1d}\comp{o}{\degvec{u}}\vposition{n + k} \cdot 2^{k-1} + \comp{p}{\degvec{v}}\vposition{n + k}\cdot 2^{k-1}
        \\
        &= \sum_{i \in \bar{S}}\degvec{u}\vposition{i}
        + \sum_{i \in \bar{S}}\degvec{v}\vposition{i}.
    \end{align*}
    Hence, in combination with \cref{eq:ineq-B-positions}, we obtain the desired
    \begin{equation*}
        \degvec{u}\vposition{\ell} + \degvec{v}\vposition{\ell}
        = \degvec{z}\vposition{\ell} \le \capvec{\vec{s}}\vposition{\ell}.
    \end{equation*}
    Therefore, indeed
    \(
    \degvec{u \oplus v}\vposition{\ell} =
    \degvec{u}\vposition{\ell} + \degvec{v}\vposition{\ell} =
    \degvec{z}\vposition{\ell}.\)

    Next, consider the positions in the set \(S\), and in particular, fix an \(\ell \in S\).
    The overall proof strategy is the same as before. We only need to adapt to the
    slightly more complicated definition of \(\comp{\star}{\star}\vposition{\ell}\).
    Writing \(\mname' \coloneqq \lceil(\capvec{\vec{s}}\vposition{\ell} +
    1)/\mname\rceil\) and applying \cref{rem:decompress}, we obtain
    \begin{align*}
        \degvec{z}\vposition{\ell}
        &\equiv_{\mname\cdot\mname'}
        \mname \cdot \comp{o + p}{\degvec{z}}\vposition{\ell}
        + (o + p)\vposition{\ell}
    - \remvec{\CW_S}{\degvec{u} + \degvec{v}}{o + p}\vposition{\ell} \\
        &\equiv_{\mname\cdot\mname'}
        \mname \cdot (\comp{o}{\degvec{u}}\vposition{\ell} + \comp{p}{\degvec{v}}\vposition{\ell})
        + (o\vposition{\ell} + p\vposition{\ell})
    - \remvec{\CW_S}{\comp{o + p}{\degvec{z}}}{o + p}\vposition{\ell}\\
        &\equiv_{\mname\cdot\mname'}
            \big(\mname \cdot \comp{o}{\degvec{u}}\vposition{\ell}
            + o\vposition{\ell}
        - \remvec{\CW_S}{\comp{o}{\degvec{u}}}{o}\vposition{\ell}\big)\\
        &\qquad+ \big(\mname \cdot \comp{p}{\degvec{v}}\vposition{\ell}
            + p\vposition{\ell}
        - \remvec{\CW_S}{\comp{p}{\degvec{v}}}{p}\vposition{\ell}\big)\\
        &\equiv_{\mname\cdot\mname'}
        \degvec{u}\vposition{\ell} + \degvec{v}\vposition{\ell}.
    \end{align*}
    In combination with
    \(0 \le \degvec{z}\vposition{\ell} \le \capvec{\vec{s}}\vposition{\ell} \le \mname\cdot
    \mname'\) and \(0 \le \degvec{u}\vposition{\ell}
        + \degvec{v}\vposition{\ell}\),
    we obtain
    \begin{equation}
        \label{eq:ineq-A-positions}
        \degvec{z}\vposition{\ell} \le
        \degvec{u}\vposition{\ell} + \degvec{v}\vposition{\ell}.
    \end{equation}
    Again, we exploit the ``checksums''.
    For $k \in \fragment1d$, we have
    \[\comp{o}{\degvec{u}}\vposition{n + d + k} + \comp{p}{\degvec{v}}\vposition{n + d + k} = (\comp{o}{\degvec{u}} + \comp{p}{\degvec{v}})\vposition{n + d + k}\]
    since those coordinates are taken modulo $3$ and
    both $\comp{o}{\degvec{u}}\vposition{n + d + k} \in \{0, 1\}$ and
    $\comp{p}{\degvec{v}}\vposition{n + d + k} \in \{0,1\}$.
    Also, as \(u \in L_1\), \(v \in L_2\), and \(z \in L_1 \oplus L_2\), we have that, for
    all positions \(i \in \fragment1n\),
    \begin{align*}
        &0\le \degvec{u}\vposition{i} \le \capvec{\vec{s}}\vposition{i} \le \allMax,\\
        &0 \le\degvec{v}\vposition{i} \le \capvec{\vec{s}}\vposition{i} \le \allMax,\\
        \quad\text{and}\quad
        &0 \le \degvec{z}\vposition{i} \le \capvec{\vec{s}}\vposition{i} \le \allMax,
    \end{align*}
    and hence,
    \[
        0 \leq \sum_{i \in {S}}\degvec{u}\vposition{i}
        + \sum_{i \in {S}} \degvec{v}\vposition{i}  < 2n(\allMax+1)
        \;\,\text{and}\;\,
        0 \leq \sum_{i \in {S}}\degvec{z}\vposition{i} < 2n(\allMax+1).
    \]
    Together, it follows that
    \begin{align*}
        \sum_{i \in S}\degvec{z}\vposition{i}
        &= \sum_{k \in \fragment1d}\comp{o + p}{\degvec{z}}\vposition{n + d + k}\cdot 2^{k-1}
        \\
        &= \sum_{k \in \fragment1d}(\comp{o}{\degvec{u}} + \comp{p}{\degvec{v}})\vposition{n + d + k}\cdot 2^{k-1}
        \\
        &= \sum_{k \in \fragment1d}\comp{o}{\degvec{u}}\vposition{n + d + k} \cdot 2^{k-1} + \comp{p}{\degvec{v}}\vposition{n + d + \ell}\cdot 2^{k-1}
        \\
        &= \sum_{i \in S}\degvec{u}\vposition{i}
        + \sum_{i \in S}\degvec{v}\vposition{i}.
    \end{align*}
    Hence, in combination with \cref{eq:ineq-A-positions}, we obtain the desired
    \begin{equation*}
        \degvec{u}\vposition{\ell} + \degvec{v}\vposition{\ell}
        = \degvec{z}\vposition{\ell} \le \capvec{\vec{s}}\vposition{\ell}.
    \end{equation*}
    Therefore, indeed
    \(
    \degvec{u \oplus v}\vposition{\ell} =
    \degvec{u}\vposition{\ell} + \degvec{v}\vposition{\ell} =
    \degvec{z}\vposition{\ell}.\)

    Overall, we obtain that $u$ and $v$ can be joined
    and that $z = u \oplus v$, which completes the proof.
\end{proof}

Finally, we are ready to give algorithms.
First, we discuss how to compute a
\(\sigma\)-defining set \(S\) (as well as witness strings that certify that \(S\) is indeed a
minimal set).

\begin{lemma}\label{lem:sig-def-set}
    Given a language \(L \subseteq \allStates^n\),
    we can compute a \(\sigma\)-defining set \(S\) for \(\sigvec{L}\),
    as well as a set of witness vectors \(\CW_{S}\) for \(S\),
    in time \(\Oh( |L| \cdot n^4)\).
\end{lemma}
\begin{proof}
    Given a language $L \subseteq \allStates^n$, we first compute the set
    $\sigvec{L}$ of all $\sigma$-vectors.
    Then, starting with \(S \coloneqq \fragment1n\), we repeatedly
    iterate over the positions \(1\) to
    \(n\); for each position \(i\), we check if \(i\) can be removed from \(S\), that is, whether
    removing the \(i\)-th position from each vector in \(\sigvec{L}\) does not decrease the size
    of \(\sigvec{L}\). If removing the \(i\)-th position would decrease the size of
    \(\sigvec{L}\), we
    also store two witness vectors that differ only at position \(i\), but not at the
    remaining positions in \(S\).
    We stop when no further positions can be removed; we return the resulting
    set~\(S\), as well as the corresponding pairs of witness vectors.

    We can check if a position \(i \in S\) can be removed by checking if the (multi-)set
    \[\sigvec{L}_{S,i} \coloneqq \big\{ v\position{S\setminus \{i\}} \mid v \in \sigvec{L}\big\}\] (where position
    \(i\) is removed) contains a duplicate element. This we can do by a linear scan over
    \(\sigvec{L}_{S,i}\)
    to construct \(\sigvec{L}_{S,i}\), sorting \(\sigvec{L}_i\), and then another linear scan
    over \(\sigvec{L}_{S,i}\).
    Observe that if we indeed detect a duplicate, we have also found the required witness.

    For the correctness, observe that during the algorithm we maintain that the positions
    in \(S\) uniquely identify the vectors in \(\sigvec{L}\); as we maintain the size of
    \(\sigvec{L}\),
    the resulting set does indeed also uniquely identify the vectors in the initial set
    \(\sigvec{L}\). Further, our algorithm trivially ensures that \(S\) is minimal, and hence, the
    returned set \(S\) is indeed \(\sigma\)-defining for \(\sigvec{L}\).

    For the running time, we can compute \(\sigvec{L}\)
    in time \(\Oh(|L| \log |\sigvec{L}| \cdot n)
    = \Oh( |L| \cdot n^2 )\)
    by iterating over \(L\) and computing each \(\sigma\)-vector separately; filtering out
    duplicates by using an appropriate data structure.
    Next, observe that we iterate over all positions in \(S\) at most \(n\) times; in each
    iteration, we check for at most \(n\) positions whether they can be removed from \(S\).
    The check if we can remove a position from \(S\) runs in the time it takes to sort
    \(\sigvec{L}\), which we can bound by \(\Oh( |\sigvec{L}| \log |\sigvec{L}| \cdot n ) = \Oh( |L| \cdot n^2 )\).
    Hence, in total the algorithm runs in the claimed running time of \(\Oh(|L|\cdot n^4)\), which
    completes the proof.
\end{proof}

Lastly, we prove the promised main result, which we restate here for convenience.

\joinalgo*

\begin{proof}
    Given the languages \(L_1\) and \(L_2\), we first compute \(\sigvec{L_1} \cap
    \sigvec{L_2}\) and drop any strings from \(L_1\) and \(L_2\) whose \(\sigma\)-vectors
    are not in \(\sigvec{L_1} \cap \sigvec{L_2}\).
    Next, we use \cref{lem:sig-def-set} to compute a \(\sigma\)-defining set \(S\) for
    \(\sigvec{L_1} \cap \sigvec{L_2}\) as well as a set of witness vectors \(\CW_S\).
    Now, for each \(\sigma\)-vector \(\vec{s} \in \sigvec{L_1} \cap \sigvec{L_2}\), we
    compute the sets \[
        \degvec{L_{1,\vec{s}}} \coloneqq \{ \degvec{u} \mid u \in L_1 \text{ and } \sigvec{u} = \vec{s}\}
        \quad\text{and}\quad
        \degvec{L_{2,\vec{s}}} \coloneqq \{ \degvec{v} \mid v \in L_2 \text{ and } \sigvec{v} = \vec{s}\}.
    \] Next, we pick arbitrary vectors \(o \in L_{1,\vec{s}}\) and \(p\in L_{2,\vec{s}}\), and compute
    the functions \(f_1, f_2 : \CZ_{S, \vec{s}} \to \ZZ\)
    \begin{equation}
     \label{eq:def-function-for-convolution}
        f_1(x) \coloneqq \begin{cases}
                          1 &\text{if } x = \comp{o}{u} \text{ for some \(u \in L_{1, \vec{s}}\)},\\
                          0 &\text{otherwise;}
                         \end{cases}
        \quad\text{and}\quad
        f_2(y) \coloneqq \begin{cases}
                          1 &\text{if }y = \comp{p}{v} \text{ for some \(v \in L_{2, \vec{s}}\)},\\
                          0 &\text{otherwise.}
                         \end{cases}
    \end{equation}
    Using \cref{thm:alg-convolution}, we compute the function \(h : \CZ_{S, \vec{s}} \to \ZZ,\)
    \[h(a) \coloneqq \sum_{x + y = a} f_1(x) \cdot f_2(y).\]
    Using \cref{rem:decompress}, we compute the set \(\degvec{L_{1,2,\vec{s}}}\) of all
    weight-vectors whose compression has a positive value for \(h\):
    \[
        \degvec{L_{1,2,\vec{s}}} \coloneqq \{ z \mid \exists a \in \CZ_{S,\vec{s}}\colon h(a) > 0 \text{ and } \comp{o + p}{z} \simeq a\}.
    \]
    Note that, for each $a \in \CZ_{S,\vec{s}}$, there is unique candidate vector $\comp{o + p}{z} \simeq a$ obtained from $a$ by ``normalizing'' the checksums to the standard binary representation.
    Iterating over \(\degvec{L_{1,2,\vec{s}}}\), we uniquely reconstruct a string
    from the weight vector and \(\vec{s}\) in the straightforward way to obtain\[
        L_{1,2,\vec{s}} \coloneqq \{ z \in \allStates^n \mid \sigvec{z} = \vec{s} \text{ and }
        \degvec{z} \in \degvec{L_{1,2,\vec{s}}} \}.
    \]
    Finally, we return the union \(L_{1,2}\) of all sets \(L_{1,2,\vec{s}}\) computed,\[
        L_{1,2} \coloneqq \bigcup_{\vec{s}\in \sigvec{L_1} \cap \sigvec{L_2}} L_{1,2,\vec{s}}.
    \]

    For the correctness, first observe that by \cref{lem:comb-1}, any string \(z \in
    L_1 \oplus L_2\) has the same \(\sigma\)-vector as a string in \(L_1\) and a string in
    \(L_2\). Hence, we can indeed compute the strings in \(L_1 \oplus L_2\) for each
    \(\sigma\)-vector separately.
    Next, by \cref{lem:comp-add,rem:decompress}, the set \(\degvec{L_{1,2,\vec{s}}}\) is indeed
    the set of all weight-vectors of strings in \(L_1 \oplus L_2\) with \(\sigma\)-vector
    \(\vec{s}\):\[
        \degvec{L_{1,2,\vec{s}}} = \degvec{\{ z \in L_1 \oplus L_2 \mid \sigvec{z} = \vec{s} \}}.
    \]
    Hence, in total, the algorithm does indeed compute \(L_1 \oplus L_2 = L_{1,2}\).

    For the running time,
    we can compute the sets \(\sigvec{L_1}\) and \(\sigvec{L_2}\)
    in total time \[\Oh\big((|L_1| + |L_2|) \log(|\sigvec{L_1}| + |\sigvec{L_2}|) \cdot n\big)
    = \Oh( \max\{|L_1|, |L_2|\} \cdot n^2 )\]
    by iterating over \(L_1\) (or \(L_2\))
    and computing each \(\sigma\)-vector separately; filtering out
    duplicates by using an appropriate data structure (note that $|\sigvec{L_1}|,|\sigvec{L_2}| \leq 2^n$).
    Afterward, we can compute \(\sigvec{L_1} \cap \sigvec{L_2}\) in the same running time
    by using standard algorithms for merging sets.

    Using the algorithm from \cref{lem:sig-def-set}, we can compute
    the \(\sigma\)-defining set \(S\) (as well as the corresponding witness vectors) in
    time \(\Oh(\max\{|L_1|, |L_2|\} \cdot n^4)\).
    As \(S\) is \(\sigma\)-defining for \(\sigvec{L_1}\cap\sigvec{L_2}\),
    we have \(|\sigvec{L_1}\cap\sigvec{L_2}| \le 2^{|S|}\).
    Now, for a fixed \(\sigma\)-vector \(\vec{s} \in \sigvec{L_1}\cap\sigvec{L_2}\),
    write \(k(\vec{s}) \coloneqq |\{i \in S \mid \vec{s}\vposition{i} = 1\}|\)
    for the number of entries $1$ of the vector $\vec{s}$ on positions from $S$.
    Recalling \cref{rem:comp-imp}, we see that the application of
    \cref{thm:alg-convolution} takes time
    \[
        (n + \allMax)^{\Oh(1)} \cdot (\allMax + 1)^{n - |S|}
        \left\lceil\frac{\sigMax + 1}{\mname}\right\rceil^{k(\vec{s})}
        \left\lceil\frac{\rhoMax + 1}{\mname}\right\rceil^{|S|-k(\vec{s})}.
    \]
    Using the notation of \cref{rem:find-prime-alg}, we set $M \coloneqq |\CZ_{S, \vec{s}}|$ and observe that $D' = (\allMax)^{\Oh(1)}$
    to compute the desired prime $p$ and the required roots of unity in the desired time.

    Finally, recovering \(L_{1,2,\vec{s}}\) can be done with a linear pass over \(\CZ_{S,\vec{s}}\) in
    the same running time; combining the recovered (disjoint) sets can then be done in
    linear time of the returned result \(L_1 \oplus L_2\).

    In total, the algorithm thus runs in time
    \begin{align*}
        &(n + 2^{\allMax})^{\Oh(1)} \cdot \max\{|L_1|, |L_2|, |L_1 \oplus L_2|\}\\
        &\quad +
        (n + 2^{\allMax})^{\Oh(1)} \cdot
        {(\allMax + 1)^{n - |S|} \cdot \sum_{k = 0}^{|S|} \binom{|S|}{k}
        \left\lceil\frac{\sigMax + 1}{\mname}\right\rceil^{k}
        \left\lceil\frac{\rhoMax + 1}{\mname}\right\rceil^{|S|-k}}.
    \end{align*}
    Using \cref{claim:upper-bound,lem:comb-upper}, the running time simplifies to \[
        (n + 2^{\allMax})^{\Oh(1)} \cdot
        {(\allMax + 1)^{n - |S|} \cdot \sum_{k = 0}^{|S|} \binom{|S|}{k}
        \left\lceil\frac{\sigMax + 1}{\mname}\right\rceil^{k}
        \left\lceil\frac{\rhoMax + 1}{\mname}\right\rceil^{|S|-k}}.
    \]
    Finally, \cref{lem:upper-bound-nicified} yields the tidier upper bound, which completes the
    proof.
\end{proof}

\subsection{Faster Algorithms for Generalized Dominating Set Problems}

Finally, we use \cref{thm:join-algorithm} to obtain faster algorithms for \srDomSet; that
is, we prove \cref{thm:main-alg-m-str}, which we restate here for convenience.

\mainalgmstr*
\begin{proof}
    For ease of notation, let us define $\tau \coloneqq \allMax + 1$ if $\mname \geq 3$ or
    $\allMax$ is odd or $\min\{\sigMax,\rhoMax\} < \allMax$,
    and $\tau \coloneqq \allMax +
    2$ if $\mname = 2$, $\allMax$ is even, and $\sigMax = \rhoMax = \allMax$.

    Let $(T,\beta)$ denote the nice tree decomposition of $G$.
    For $t \in V(T)$, we set $X_t \coloneqq \beta(t)$ and
    write $V_t$ for the set of vertices contained in bags below $t$ (including $t$ itself).

    For each node $t \in V(T)$ and each $i \in \fragmentco{0}{\mname}$, we compute the
    language $L_{t,i} \subseteq \allStates^{X_t}$ of all strings $x \in \allStates^{X_t}$
    that are compatible with $(G\position{V_t},X_t)$\footnote{To follow standard notation
    for dynamic programming algorithms on tree decompositions, we use $X_t$ here to denote the set of portal vertices.} via a witnessing solution set
    $S_x$ such that $|S_x \setminus X_t| \equiv_{\mname} i$.
    We have that $L_{t,i} \times L_{t,i} \subseteq \CR^{|X_t|}$ by
    \cref{la:hamming-distance-parity-mod-p-general}, and hence, \cref{thm:upper-bound} yields
    \begin{equation}
        \label{eq:language-size-bound-dp}
        |L_{t,i}| \leq \tau^{|X_t|} \leq \tau^{\tw + 1}.
    \end{equation}

    We compute the sets $L_{t,i}$ for nodes $t \in V(T)$ in a bottom-up fashion
    starting at the leaves of $T$.

    For a leaf $t$ of $T$, we have $X_t = V_t = \emptyset$ and
    $L_{t,i} = \{\varepsilon\}$ for every $i \in \fragmentco{0}{\mname}$
    (where $\varepsilon$ denotes the empty string).

    For an internal node $t$, suppose we already computed all sets
    $L_{t',i}$ for all children $t'$ of $t$.
    We proceed depending on the type of \(t\).

    \begin{description}
        \item[Forget:]
            First, suppose $t$ is a forget-node, and let $t'$ denote the unique child of $t$.
            Also, assume that $X_{t'} = X_t \cup \{v\}$, that is,
            $v \in V(G)$ is the vertex forgotten at $t$.
            We say that a string $x \in \allStates^{X_{t'}}$ \emph{is $\rho$-happy at $v$}
            if $x\position{v} = \rho_c$ and $c \in \rho$.
            Also, we say that a string $x \in \allStates^{X_{t'}}$ \emph{is $\sigma$-happy at $v$}
            if $x\position{v} = \sigma_d$ and $d \in \sigma$.
            It is easy to see that
            \begin{align*}
             L_{t,i} = \;\;\; &\{x\position{X_t} \mid x \in L_{t',i} \text{ such that $x$ is $\rho$-happy at }v\}\\
                       \cup\; &\{x\position{X_t} \mid x \in L_{t',i-1} \text{ such that $x$ is $\sigma$-happy at }v\},
            \end{align*}
            where the index $i-1$ is taken modulo $\mname$.
            Hence, using \cref{eq:language-size-bound-dp},
            for each \(i \in \fragmentco0{\mname}\),
            we can compute the set $L_{t,i}$ in time
            $\tau^\tw \cdot (\allMax + \tw)^{\Oh(1)}$.
        \item[Introduce:]
            Next, consider the case that $t$ is an introduce-node, and let $t'$ denote
            the unique child of $t$.
            Suppose that $X_t = X_{t'} \cup \{v\}$, that is,
            $v \in V(G)$ is the vertex introduced at $t$.
            Note that we have $N_G(v) \cap V_t \subseteq X_{t'}$.
            We say a string $x \in \allStates^{X_t}$ \emph{$\rho$-extends} a
            string $y \in \allStates^{X_{t'}}$ if
            \begin{enumerate}
                \item $x\position{v} = \rho_c$ for some $c \in \{0,\dots,\rhoMax\}$,
                \item $c = |\{w \in N_G(v) \mid y\position{w} \in \sigStates\}|$, and
                \item $x\position{X_{t'}} = y\position{X_{t'}}$.
            \end{enumerate}
            A string $x \in \allStates^{X_t}$ \emph{$\sigma$-extends} a
            string $y \in \allStates^{X_{t'}}$ if
            \begin{enumerate}
                \item $x\position{v} = \sigma_d$ for some $d \in \{0,\dots,\sigMax\}$,
                \item $d = |\{w \in N_G(v) \mid y\position{w} \in \sigStates\}|$,
                \item $\sigvec{x}\position{w} = \sigvec{y}\position{w}$ for all $w \in X_{t'}$,
                \item $\degvec{x}\position{w} = \degvec{y}\position{w}$ for all
                    $w \in X_{t'} \setminus N_G(v)$, and
                \item $\degvec{x}\position{w} = \degvec{y}\position{w} + 1$
                    for all $w \in X_{t'} \cap N_G(v)$.
            \end{enumerate}
            Finally, a string $x \in \allStates^{X_t}$ \emph{extends} a
            string $y \in \allStates^{X_{t'}}$ if it $\rho$-extends $y$ or it $\sigma$-extends $y$.
            We have that
            \begin{align*}
                L_{t,i} &= \{x \in \allStates^{X_t} \mid x \text{ extends some $y \in L_{t',i}$}\}.
            \end{align*}
            Since, for every string $y \in \allStates^{X_{t'}}$, there are at most two strings
            $x \in \allStates^{X_t}$ such that $x$ extends $y$, all of the languages
            $L_{t,i}$ can be computed by iterating over all sets $L_{t',i}$ once.
            This takes time $\tau^\tw \cdot (\allMax + \tw)^{\Oh(1)}$ using
            \cref{eq:language-size-bound-dp}.
        \item[Join:]
            Finally, suppose that $t$ is a join-node, and let $t_1,t_2$ denote the two children
            of $t$.
            Observe that $X_t = X_{t_1} = X_{t_2}$.

            To compute the sets $L_{t,i}$, we intend to rely on the algorithm from
            \cref{thm:join-algorithm}.
            However, this is not directly possible since the weight-vectors of the resulting
            strings would not be correct.
            Indeed, suppose that $x_1,x_2 \in \allStates^{X_t}$ are strings that are
            compatible with $(G\position{V_{t_1}},X_{t_1})$ and $(G\position{V_{t_2}},X_{t_2})$,
            respectively.
            Also, suppose that $x_1$ and $x_2$ can be joined.
            Then, $x_1 \oplus x_2$ is not (necessarily) compatible with $(G\position{V_t},X_t)$
            since, for every $v \in X_t$, the vertices from $N_G(v) \cap X_t$
            that are contained in the solution set are counted twice.
            For this reason, we first modify all the strings from the languages
            $L_{t_1,i}$ such that indices do not take solution vertices from the set
            $X_t$ into account.

            For each \(i \in \fragmentco0{\mname}\) and each
            $x \in L_{t_1,i}$, we perform the following steps.
            We define the string $\widehat{x} \in \allStates^{X_t}$ as
            \[\widehat{x}\position{v} \coloneqq \begin{cases}
                                                 \rho_{\widehat c} &\text{if } x\position{v} = \rho_{c} \text{ and } c = \widehat c + |\{w \in N_G(v) \cap X_t \mid x\position{w} \in \sigStates\}|,\\
                                                 \sigma_{\widehat d} &\text{if } x\position{v} = \sigma_{d} \text{ and } d = \widehat d + |\{w \in N_G(v) \cap X_t \mid x\position{w} \in \sigStates\}|
                                                \end{cases}
            \]
            and add $\widehat{x}$ to the set $\widehat L_{t_1,i}$.
            Observe that, for each \(i \in \fragmentco{0}{\mname}\), we can compute
            the set $\widehat L_{t_1,i}$ in time $\tau^\tw \cdot (\allMax + \tw)^{\Oh(1)}$
            using \cref{eq:language-size-bound-dp}.

            Now, we easily observe that
            \[L_{t,i} = \bigcup_{j \in \fragmentco{0}{\mname}} \widehat L_{t_1,j} \oplus
                L_{t_2,i-j},\]
            where the index $i-j$ is taken modulo $\mname$.
            We iterate over all choices of $i,j \in \fragmentco{0}{\mname}$
            and compute $\widehat L_{t_1,j} \oplus L_{t_2,i-j}$ using \cref{thm:join-algorithm}.
            To that end, we need to ensure that the requirements of
            \cref{thm:join-algorithm} are satisfied.
            We have already argued that $L_{t_2,i-j} \times L_{t_2,i-j} \subseteq \CR^{|X_t|}$.
            Further, $\widehat L_{t_1,j}$ is realized by $(G\position{V_t} - E(X_t,X_t),X_t)$
            (that is, the graph obtained from $G\position{V_t}$ by removing all edges within
            $X_t$), and hence, $\widehat L_{t_1,j} \times \widehat L_{t_1,j} \subseteq \CR^{|X_t|}$ using
            \cref{la:hamming-distance-parity-mod-p-general}.

            Overall, this allows us to compute the join in time
            $\tau^\tw \cdot \tw^{\Oh(1)} \cdot 2^{(\allMax)^{\Oh(1)}}$
            using \cref{thm:join-algorithm}.
    \end{description}
    Since processing a single node of $T$ takes time
    $\tau^\tw \cdot \tw^{\Oh(1)} \cdot 2^{(\allMax)^{\Oh(1)}}$
    and we have $|V(T)| = \Oh(\tw \cdot |V(G)|)$, we see that all sets $L_{t,i}$ can be
    computed in the desired time.

    To decide whether $G$ has a $(\sigma,\rho)$-set, we consider the root node
    $t \in V(T)$ for which $X_t = \emptyset$ and $V_t = V(G)$.
    Then, $G$ has a $(\sigma,\rho)$-set if and only if $\varepsilon \in L_{t,i}$ for some
    $i \in \fragmentco{0}{\mname}$, which completes the proof.
\end{proof}

Next, we explain how to extend the algorithm to the optimization and counting version of the problem.
For the optimization version, it is easy to see that we can keep track of the size of partial solutions in the dynamic programming tables.
This increases the size of all tables by a factor of $|V(G)|$.
Hence, we obtain the following theorem for the optimization version.

\begin{theorem}\label{thm:main-alg-m-str-opt}
    Let $(\sigma,\rho)$ denote finite, $\mname$-structured sets for some $\mname \geq 2$.
    Then, there is an algorithm $\CA$ that, given a graph $G$, an integer $k$, and a nice tree decomposition
    of $G$ of width $\tw$, decides whether $G$ has a $(\sigma,\rho)$-set of size at most (at least) $k$.

    If $\mname \geq 3$ or $\allMax$ is odd
    or $\min\{\sigMax,\rhoMax\} < \allMax$, then
    algorithm $\CA$ runs in time
    \[
      (\allMax + 1)^\tw \cdot 2^{(\allMax)^{\Oh(1)}} \cdot \tw^{\Oh(1)} \cdot |V(G)|^2.
    \]

    If $\mname = 2$, $\allMax$ is even, and $\sigMax = \rhoMax = \allMax$,
    then algorithm $\CA$ runs in
    time
    \[
      (\allMax + 2)^\tw \cdot 2^{(\allMax)^{\Oh(1)}} \cdot \tw^{\Oh(1)} \cdot |V(G)|^2.
    \]
\end{theorem}

For the counting version of the problem (that is, we wish to compute the number of solution sets), the situation is more complicated.
The main challenge for the counting version stems from the application of \cref{thm:alg-convolution} for which we need to find an appropriate prime $p$ as well as certain roots of unity.
In Remark~\ref{rem:find-prime-alg}, we explained how to find these objects in time roughly linear in $p$ (ignoring various lower-order terms that are not relevant to the discussion here).
However, for the counting version, we would need $p$ to be larger than the number of solutions, which results in a running time that is exponential in $|V(G)|$ in the worst case.

Luckily, we can circumvent this problem using the Chinese Remainder Theorem.
The basic idea is to compute the number of solutions modulo $p_i$ for a sufficiently large number of distinct small primes $p_i$.
Assuming $\prod_i p_i > 2^{|V(G)|}$, the number of solutions can be uniquely recovered using the Chinese Remainder Theorem.

\begin{theorem}[{Chinese Remainder Theorem \cite[Section 5.4]{GathenG13}}]
 Let $m_1,\dots,m_\ell$ denote a sequence of integers that are pairwise coprime, and define $M
 \coloneqq \prod_{i \in \nset{\ell}}m_i$.
 Also, let $0 \leq a_i < m_i$ for all $i \in \nset{\ell}$.
 Then, there is a unique number $0 \leq s < M$ such that
 \[s \equiv a_i \pmod {m_i}\]
 for all $i \in \nset{\ell}$.
 Moreover, there is an algorithm that, given $m_1,\dots,m_\ell$ and $a_1,\dots,a_\ell$, computes the number $s$ in time $\Oh((\log M)^2)$.
\end{theorem}

More generally, we can build on the following extension of \cref{thm:alg-convolution} which may also be interesting in its own right.

\begin{theorem}
 \label{thm:alg-convolution-extended}
 Let $d_1,\dots,d_n \geq 2$, and let $D \coloneqq \prod_{i=1}^nd_i$.
 Suppose $d_1',\dots,d_\ell'$ is the list of integers obtained from $d_1,\dots,d_n$ by removing duplicates, and let $D' \coloneqq \prod_{i=1}^\ell d_i'$.
 Also, let $f,g\colon \ZZ_{d_1} \times \dots \times \ZZ_{d_n} \rightarrow \ZZ$ denote a function, and
 let $h\colon\ZZ_{d_1} \times \dots \times \ZZ_{d_n} \rightarrow \FF_p$ denote the convolution
 \[h(a) \coloneqq \sum_{a_1 + a_2 = a} f(a_1) \cdot g(a_2).\]
 Moreover, let $M$ denote a non-negative integer such that all images of $f,g$ and $h$ are contained in $\{0,\dots,M\}$.
 Then, the function $h$ can be computed in time $D \cdot (\log D + n + M')^{\Oh(1)}$ where
 \[M' \coloneqq \max\Big\{\log M,8 \cdot 10^9,\exp(D')\Big\}.\]
\end{theorem}

\begin{proof}
 Let $m \coloneqq \lceil M' \rceil$.
 We compute the list of the first $m$ primes $p_1 < \dots < p_m$ such that $p_i \equiv 1 \pmod {D'}$ for all $i \in \nset{m}$.
 By the Prime Number Theorem for Arithmetic Progressions \cite[Theorem 1.5]{BennettMOR18}
 we get that $p_i = \Oh(\varphi(D') \cdot i \cdot \log i)$,
 for all~$i\in \nset{m}$,
 where $\varphi$ denotes Euler's totient function.
 In particular, the largest prime~$p_m$ satisfies
 $p_m = \Oh(m \cdot (\log m)^2)$ because $\varphi(D') \leq D'$ and $m \geq \exp(D')$.
 Since prime testing can be done in polynomial time, we can find the sequence $p_1,\dots,p_m$ in time $\Oh(m \cdot (\log m)^c)$ for some constant $c$.

 Next, for every $i \in \nset{m}$ and $j \in \nset{n}$, we compute a $d_j$-th root of unity in $\FF_{p_i}$ as follows.
 First, observe that such a root of unity exists since $d_j$ divides $p_i - 1$.
 Now, we simply iterate over all elements $x \in \FF_{p_i}$ and test whether a given element $x$ is a $d_j$-th root of unity in time $(d_j + \log p_i)^{\Oh(1)}$.
 So, overall, computing all roots of unity can be done in time
 \[p_m \cdot (n + m)^{\Oh(1)} = (n + m)^{\Oh(1)}.\]

 Now, for every $i \in \nset{m}$ and $a \in \ZZ_{d_1} \times \dots \times \ZZ_{d_n}$, we compute
 \[h_i(a) \coloneqq h(a) \pmod{p_i}\]
 using \cref{thm:alg-convolution} taking $\Oh(m \cdot D \cdot \log D)$ many arithmic operations.
 Since each arithmetic operation can be done time $(\log p_m)^{\Oh(1)}$, we obtain a total running time of
 \[D \cdot (m + \log D)^{\Oh(1)}.\]
 Finally, we can recover all numbers $h(a)$ by the Chinese Remainder Theorem in time $m^{\Oh(1)}$.
 Note that $\prod_{i \in \nset{m}}p_i > 2^m \geq M$, which implies that all numbers are indeed uniquely recovered.
 In total, this achieves the desired running time.
\end{proof}

Now, to obtain an algorithm for the counting version, we follow the algorithm from \cref{thm:main-alg-m-str} and replace the application of \cref{thm:alg-convolution} by \cref{thm:alg-convolution-extended}.
Also, we change the definition of the functions $f_i$ in Equation \eqref{eq:def-function-for-convolution} to give the number of partial solutions.
Note that we can set $M \coloneqq 2^{|V(G)|}$ since the number of solutions is always bounded by $2^{|V(G)|}$.
Also observe that $D' = (\allMax)^{\Oh(1)}$ which implies that $M' \leq |V(G)| + 2^{(\allMax)^{\Oh(1)}}$.
Additionally, similarly to the optimization version, we keep track of the size of solutions.
Overall, we obtain the following theorem for the counting version.

\begin{theorem}\label{thm:main-alg-m-str-count}
    Let $(\sigma,\rho)$ denote finite, $\mname$-structured sets for some $\mname \geq 2$.
    Then, there is an algorithm $\CA$ that, given a graph $G$, an integer $k$, and a nice tree decomposition
    of $G$ of width $\tw$, computes the number of $(\sigma,\rho)$-sets of size exactly $k$ in $G$.

    If $\mname \geq 3$ or $\allMax$ is odd
    or $\min\{\sigMax,\rhoMax\} < \allMax$, then
    algorithm $\CA$ runs in time
    \[
      (\allMax + 1)^\tw \cdot 2^{(\allMax)^{\Oh(1)}} \cdot (\tw + |V(G)|)^{\Oh(1)}.
    \]

    If $\mname = 2$, $\allMax$ is even, and $\sigMax = \rhoMax = \allMax$,
    then algorithm $\CA$ runs in
    time
    \[
      (\allMax + 2)^\tw \cdot 2^{(\allMax)^{\Oh(1)}} \cdot (\tw + |V(G)|)^{\Oh(1)}.
    \]
\end{theorem}

For the counting version, we omit a more detailed analysis on the dependence on the number of vertices.
However, due to the application of the Chinese Remainder Theorem, the running time increases at least by a factor of $|V(G)|^2$ in comparsion to \cref{thm:main-alg-m-str-opt}.

We obtain \cref{thm:alg-main-intro} by combining \cref{thm:vanrooij-intro,thm:main-alg-m-str-count}.

\section{Faster Algorithms via Representative Sets}
\label{sec:representative-set}

Next, we present a second algorithm for \srDomSet which is designed
for the decision version and the case that one of the sets $\sigma,\rho$ is cofinite.
More precisely, the aim of this section is to prove \cref{thm:dp-representative-set-intro}.
Let us stress again that the algorithm given in this section
works for all finite or cofinite sets $\sigma, \rho$,
but in the case where both $\rho$ and $\sigma$ are finite,
it is slower than existing algorithms (see \cref{thm:vanrooij-intro}).
The algorithm is based on representative sets.
Intuitively speaking, for a graph $G$ and a set $U \subseteq V(G)$,
the idea is to not compute the entire set $L \subseteq \allStates^U$ of strings
that are compatible with $(G,U)$,
but only a \emph{representative set} $R \subseteq L$ such that,
if there is a partial solution for some $x \in L$
that can be extended to a full solution via some $y \in \allStates^{U}$,
then there is also a partial solution $x' \in R$
that can be extended to a full solution via $y$.
If one of the sets $\sigma, \rho$ is cofinite,
then it is possible to obtain representative sets $R \subseteq L$
that are much smaller than the number of partial solutions
that are maintained by standard dynamic programming algorithms
(see, e.g., \cite{Rooij20}).

For technical reasons, it turns out to be more convenient to work with the alphabet $\allStates_n = \{\sigma_0,\dots,\sigma_n,\rho_0,\dots,\rho_n\}$, where $n$ denotes the number of vertices of the graph $G$ under investigation.

To obtain the representative sets, we build on ideas that were already used in \cite{MarxSS25}.
In the following, let $\omega < 2.37286$ denote the matrix multiplication exponent \cite{AlmanW21}.

Let us first restrict ourselves to the case where both $\rho$ and $\sigma$ are cofinite.
Let $k \geq 1$ and let $F_1,\dots,F_k \subseteq \NN$ denote finite sets of \emph{forbidden elements}.
Intuitively speaking, for a set $X \subseteq V(G)$ consisting of $k$ vertices $v_1,\dots,v_k$, we  set $F_i \coloneqq \ZZ_{\geq 0} \setminus \sigma$ if $v_i$ is selected into a partial solution, and $F_i \coloneqq \ZZ_{\geq 0} \setminus \rho$ otherwise.

\begin{definition}
 \label{def:compatibility-graph-forbidden}
 The \emph{compatibility graph for forbidden sets $\mathbf F = (F_1,\dots,F_k)$}
 is the infinite graph $\CC = \CC(\mathbf F)$ with
 \begin{itemize}
  \item $V(\CC) \coloneqq U^k \cup V^k$ where $U,V$ are disjoint sets both identified with $\NN$, and
  \item $E(\CC) \coloneqq \{((a_1,\dots,a_k),(b_1,\dots,b_k)) \mid \forall i \in \nset{k}\colon a_i + b_i \notin F_i\}$.
 \end{itemize}
 Let $\CS \subseteq \NN^k$ denote a finite set.
 We say that $\CS' \subseteq \CS$ is an \emph{$\mathbf F$-representative set of $\CS$} if, for every $b \in \NN^k$, we have that
 \begin{equation}
  \exists a \in \CS\colon (a,b) \in E(\CC(F_1,\dots,F_k)) \;\;\;\Longleftrightarrow\;\;\; \exists a' \in \CS'\colon (a',b) \in E(\CC(F_1,\dots,F_k)).
 \end{equation}
\end{definition}

Now, the basic idea is that, for a set $X = \{v_1,\dots,v_k\}$ and a fixed $\sigma$-vector $\vec{s} \in \{0,1\}^{|X|}$, it suffices to keep an $(F_1,\dots,F_k)$-representative set of the weight vectors of those partial solutions that have $\sigma$-vector $\vec{s}$.
Here, $F_i \coloneqq \ZZ_{\geq 0} \setminus \sigma$ if $\vec{s}\position{v_i} = 1$, and $F_i \coloneqq \ZZ_{\geq 0} \setminus \rho$ if $\vec{s}\position{v_i} = 0$.

\begin{lemma}[{\cite[Lemma 3.2 \& 3.5]{MarxSS25}}]
 \label{la:representative-set-forbidden}
 Let $F_1,\dots,F_k \subseteq \NN$ denote finite sets such that $|F_i| \leq t$ for all $i \in
 \nset{k}$.
 Further, let $\CS \subseteq \NN^k$ denote a finite set.
 Then, one can compute an $(F_1,\dots,F_k)$-representative set $\CS'$ of $\CS$ such that $|\CS'| \leq (t+1)^k$ in time $\Oh(|\CS| \cdot (t+1)^{k(\omega - 1)}k)$.
\end{lemma}

We can use \cref{la:representative-set-forbidden} to compute representative sets of small size if both $\rho$ and $\sigma$ are cofinite.
To also cover finite sets, we need to extend the above results as follows.
Consider $k,\ell \in \NN$ such that $k + \ell \geq 1$ and let
$F_1,\dots,F_k,P_1,\dots,P_\ell \subseteq \NN$ denote finite sets of \emph{forbidden elements} and \emph{positive elements}.
The basic intuition is similar to before.
Consider a set $X \subseteq V(G)$ consisting of $k+\ell$ vertices $v_1,\dots,v_{k+\ell}$ and a fixed $\sigma$-vector $\vec{s} \in \{0,1\}^{|X|}$ that has $k$ entries corresponding to an infinite set, and $\ell$ entries corresponding to a finite set (e.g., if $\sigma$ is infinite and $\rho$ is finite, then there are $k$ many $1$-entries since they correspond to the infinite set $\sigma$).
With this intuition in mind, we generalize \cref{def:compatibility-graph-forbidden} as follows.

\begin{definition}
 The \emph{compatibility graph for forbidden sets $\mathbf F = (F_1,\dots,F_k)$ and positive sets $\mathbf P = (P_1,\dots,P_\ell)$} is the infinite graph $\CC = \CC(\mathbf F;\mathbf P)$ with
 \begin{itemize}
  \item $V(\CC) \coloneqq U^{k + \ell} \cup V^{k + \ell}$ where $U,V$ are disjoint sets both identified with $\NN$, and
  \item $E(\CC) \coloneqq \{((a_1,\dots,a_{k+\ell}),(b_1,\dots,b_{k+\ell})) \mid \forall i
      \in \nset{k}\colon a_i + b_i \notin F_i \text{ and } \forall j \in \nset{\ell}\colon \newline a_{k+j} + b_{k+j} \in P_j\}$.
 \end{itemize}
 Let $\CS \subseteq \NN^{k+\ell}$ be a finite set.
 We say that $\CS' \subseteq \CS$ is an \emph{$(\mathbf F;\mathbf P)$-representative set of $\CS$} if, for every $b \in \NN^{k+\ell}$, we have that
 \begin{equation}
  \label{eq:representation}
  \exists a \in \CS\colon (a,b) \in E(\CC) \;\;\;\Longleftrightarrow\;\;\; \exists a' \in \CS'\colon (a',b) \in E(\CC).
 \end{equation}
\end{definition}

By taking a brute-force approach to positions with positive sets, we can also generalize \cref{la:representative-set-forbidden}.

\begin{lemma}
 \label{la:representative-set-mixed}
 Let $F_1,\dots,F_k,P_1,\dots,P_\ell \subseteq \NN$ denote finite sets such that $|F_i| \leq
 t$ for all $i \in \nset{k}$ and $\max(P_j) \leq t$ for all $j \in \nset{\ell}$.
 Further, write $\CS \subseteq \NN^{k+\ell}$ for a finite set.
 Then, one can compute an $(\mathbf F; \mathbf P)$-representative set $\CS'$ of $\CS$
 where~$\mathbf F = (F_1,\dots,F_k)$ and~$\mathbf P = (P_1,\dots,P_\ell)$
 such that $|\CS'| \leq (t+1)^{k + \ell}$ in time $\Oh(|\CS| \cdot (t+1)^{\ell + k(\omega - 1)}(k + \ell))$.
\end{lemma}

\begin{proof}
 We proceed in two steps.
 We first compute the set
 \[\CS'' \coloneqq \{(a_1,\dots,a_{k+\ell}) \in \CS \mid \forall j \in \nset{\ell}\colon a_{k+j} \leq t\}.\]
 Clearly, the set $\CS''$ can be computed in time $\Oh(|\CS| \cdot \ell)$.
 Also, every element $(a_1,\dots,a_{k+\ell}) \in \CS \setminus \CS''$ is an isolated
 vertex in $\CC = \CC(\mathbf F;\mathbf P)$ because $\max(P_j) \leq t$ for all
 $j \in \nset{\ell}$.
 Hence, it suffices to compute an $(\mathbf F; \mathbf P)$-representative set of $\CS''$.

 We say that two elements $(a_1,\dots,a_{k+\ell}),(b_1,\dots,b_{k+\ell}) \in \CS''$ are
 \emph{positive-equivalent} if $a_{k+j} = b_{k+j}$ for all $j \in \nset{\ell}$.
 Let $\CS_1,\dots,\CS_p$ denote the equivalence classes of this relation.
 Note that, by the definition of the set $\CS''$, we have $p \leq (t+1)^{\ell}$.
 We can compute the sets $\CS_1,\dots,\CS_p$ in time $\Oh(|\CS| \cdot (t+1)^{\ell}(k + \ell))$.
 For each $i \in \nset{p}$, we can compute an $(\mathbf F; \mathbf P)$-representative set $\CS_i'$ of $\CS_i$ using Lemma \ref{la:representative-set-forbidden} in time $\Oh(|\CS| \cdot (t+1)^{k(\omega - 1)}k)$.
 Then, $|\CS_i'| \leq (t+1)^k$.
 We define
 \[\CS' \coloneqq \bigcup_{i \in \nset{p}} \CS_i'.\]
 Observe that $|\CS'| \leq p \cdot (t+1)^k \leq (t+1)^{k + \ell}$.
 Moreover, computing $\CS'$ overall takes time $\Oh(|\CS| \cdot \ell + |\CS| \cdot (t+1)^{\ell}(k + \ell) + p \cdot |\CS| \cdot (t+1)^{k(\omega - 1)}k) = \Oh(|\CS| \cdot (t+1)^{\ell + k(\omega - 1)}(k + \ell))$.
\end{proof}

Before we state the final algorithm, we first define the following operations on representative sets and prove that they preserve $(\mathbf F; \mathbf P)$-representation.

\begin{lemma}
    \label{la:representative-set-mixed-operations}
    Let $F_1,\dots,F_k,P_1,\dots,P_\ell \subseteq \NN$ denote finite sets.
    Further, write $\CS \subseteq \NN^{k+\ell}$ for a finite set and let $\CS'$ be a $(\mathbf F; \mathbf P)$-representative set of~$\CS$
    where~$\mathbf F = (F_1,\dots,F_k)$ and $\mathbf P = (P_1,\dots,P_\ell)$.
    \begin{enumerate}[label = (\alph*)]
        \item 
        For every vector $d \in \ZZ^{k+\ell}$, the set $\CS' + d \coloneqq \{s + d \mid s \in \CS\}$ is an $(\mathbf F;\mathbf P)$-representative set of $\CS+d$.
        \item 
        For every $i \in \numb{k}$, the set $\CS' \ominus i \coloneqq \{s \fragment{1}{i-1}~s\fragment{i+1}{k+\ell} \mid s \in \CS', s\position{i} \notin F_i\}$
        is an $(\mathbf F';\mathbf P)$-representative set of $\CS \ominus i$ where~$\mathbf F' \coloneqq (F_1,\dots,F_{i-1},F_{i+1},\dots,F_k)$.
        \item 
        For every $j \in \numb{\ell}$, the set $\CS'\ominus (k+j) \coloneqq \{s\fragment{1}{k+j-1}~s\fragment{k+j+1}{k+\ell} \mid s \in \CS', s\position{k+j} \in P_j\}$
        is an $(\mathbf F;\mathbf P')$-representative set of $\CS\ominus (k+j)$ where~$\mathbf P' \coloneqq (P_1,\dots,P_{j-1},P_{j+1},\dots,P_\ell)$.
    \end{enumerate}
    For every finite set~$Q\subseteq \NN$, and
    \begin{enumerate}[resume, label = (\alph*)]
        \item for every $i \in \numb{k+1}$, the set $\CS' \oplus i \coloneqq \{s \fragment{1}{i-1}~0~s\fragment{i}{k+\ell} \mid s \in \CS'\}$
        is an $(\mathbf F';\mathbf P)$-representative set of $\CS \oplus i$ where~$\mathbf F' \coloneqq (F_1,\dots,F_{i-1},Q,F_i,\dots,F_k)$, and
        \item for every $j \in \numb{\ell+1}$, the set $\CS'\oplus (k+j) \coloneqq \{s\fragment{1}{k+j-1}~0~s\fragment{k+j}{k+\ell} \mid s \in \CS'\}$
        is an $(\mathbf F;\mathbf P')$-representative set of $\CS \oplus (k+j)$ where~$\mathbf P' \coloneqq (P_1,\dots,P_{j-1},Q,P_j,\dots,P_\ell)$.
    \end{enumerate}
    For every finite set $\CR \subseteq \NN^{k+\ell}$ and every $(\mathbf F;\mathbf P)$-representative set~$\CR'$ of~$\CR$,
    \begin{enumerate}[resume, label = (\alph*)]
        \item the set~$\CS' \cup \CR'$ is an $(\mathbf F;\mathbf P)$-representative set of $\CS \cup \CR$, and
        \item the set~$\CS' + \CR' \coloneqq \{ s + r \mid s \in \CS', r \in \CR'\}$ is an $(\mathbf F;\mathbf P)$-representative set of $\CS + \CR$.
    \end{enumerate}
\end{lemma}

\begin{proof}
    For two vectors~$a, b \in \NN^{k+\ell}$, we write $a \sim_{\mathbf F;\mathbf P} b$ if $(a,b) \in E(\CC(\mathbf F;\mathbf P))$,
    that is, if $a \pos{i} + b\pos{i} \notin F_i$ for all $i \in \numb{k}$ and $a \pos{k+j} + b\pos{k+j} \in P_j$ for all $j \in \numb{\ell}$.

    We first observe that we only have to prove the forward direction of \eqref{eq:representation} since $\CS' \subseteq \CS$.
    We prove the different cases independently.
    \begin{enumerate}[label = (\alph*)]
        \item Fix some vector~$d \in \ZZ^{k+\ell}$, and let~$a \in \CS + d$ and~$b \in \NN^{k+\ell}$ such that~$a \sim_{\mathbf F;\mathbf P} b$.
        By the definition of the set~$\CS+d$, there is some~$\hat a \in \CS$ such that~$a = \hat a + d$.
        So~$\hat a + d \sim_{\mathbf F;\mathbf P} b$ which is equivalent to $\hat a \sim_{\mathbf F;\mathbf P} b + d$.
        Since~$\CS'$ is an $(\mathbf F;\mathbf P)$-representative set of~$\CS$, there is some~$\hat a' \in \CS'$ such that $\hat a' \sim_{\mathbf F;\mathbf P} b + d$.
        It follows that $\hat a'+d \sim_{\mathbf F;\mathbf P} b$ and $\hat a'+d \in \CS' + d$, which concludes this case.

        \item Without loss of generality, we assume~$i = 1$.
        Let~$a \in \CS \ominus 1$ and~$b \in \NN^{k-1+\ell}$ such that~$a \sim_{\mathbf F';\mathbf P} b$ where~$\mathbf F' = (F_2,\dots,F_k)$.
        By the definition of the set~$\CS\ominus i$, there is some~$\hat a \in \CS$ such that~$\hat a\fragment{2}{k+\ell} = a$ and $\hat a \pos{1} \notin F_1$.
        In particular, $\hat a \sim_{\mathbf F;\mathbf P} 0~b$.
        Since $\CS'$ is an $(\mathbf F;\mathbf P)$-representative set of $\CS$,
        there is some~$\hat a' \in \CS'$ such that $\hat a' \sim_{\mathbf F;\mathbf P} 0~b$.
        This means~$\hat a' \pos{1} \notin F_1$ and thus, $\hat a'\fragment{2}{k+\ell} \sim_{\mathbf F';\mathbf P} b$ and~$\hat a'\fragment{2}{k+\ell} \in \CS'\ominus 1$.

        \item This case is analogous to the previous case except that we use~$\hat a \pos{k+j} \in P_j$.

        \item Fix a finite set~$Q \subseteq \NN$ and let~$i \in \numb{k}$.
        Assume without loss of generality that~$i = 1$.
        Moreover, let~$a \in \CS \oplus i$ and~$b \in \NN^{k+\ell+1}$ such that~$a \sim_{\mathbf F';\mathbf P} b$.
        By the definition of the set~$\CS\oplus 1$, we get~$a \pos{1} = 0$ and hence, $b \pos{1} \notin Q$.
        Moreover,~$a\fragment{2}{k+\ell+1} \sim_{\mathbf F;\mathbf P} b \fragment{2}{k+\ell+1}$.
        Since $\CS'$ is a $(\mathbf F;\mathbf P)$-representative set of~$\CS$, there is some vector~$a' \in \CS'$ such that $a' \sim_{\mathbf F;\mathbf P} b\fragment{2}{k+\ell+1}$.
        Then $0~a' \in \CS' \oplus 1$ and~$0~a' \sim_{\mathbf F';\mathbf P} b$ because $b \pos{1} \notin Q$.

        \item This case is analogous to the previous case except that we use~$a\pos{k+j} = 0$ and $b \pos{k+j} \in Q$.

        \item This case holds trivially since the vectors are not changed.

        \item Let~$a \in \CS + \CR$ and~$b \in \NN^{k+\ell}$ such that~$a \sim_{\mathbf F;\mathbf P} b$.
        By the definition of the set~$\CS + \CR$, there are~$a_1 \in \CS$ and~$a_2 \in \CR$ such that~$a_1 + a_2 = a$.
        This means $a_1 \sim_{\mathbf F;\mathbf P} b + a_2$ and, since~$\CS'$ is an $(\mathbf F;\mathbf P)$-representative set of~$\CS$,
        there is some~$a'_1 \in \CS'$ such that~$a'_1 \sim_{\mathbf F;\mathbf P} b + a_2$ which is equivalent to $a_2 \sim_{\mathbf F;\mathbf P} b + a'_1$.

        Since~$\CR'$ is an $(\mathbf F;\mathbf P)$-representative set of~$\CR$, we obtain a vector~$a'_2 \in \CR'$ such that~$a'_2 \sim_{\mathbf F;\mathbf P} b + a'_1$ which is equivalent to $a'_1 + a'_2 \sim_{\mathbf F;\mathbf P} b$.
        Also $a'_1 + a'_2 \in \CS' \oplus \CR'$.
    \qedhere
  \end{enumerate}
\end{proof}

Now, we have all the tools to present an algorithm
for \srDomSet on graphs of small treewidth
based on representative sets.
To state its running time, let us introduce the following cost measure for sets of natural numbers.
We write $\emptyset \neq \tau \subseteq \NN$ for a finite or cofinite set.
If $\tau$ is finite, then we define $\cost(\tau) \coloneqq \max(\tau)$.
Otherwise, $\tau$ is cofinite and we define $\cost(\tau) \coloneqq |\NN \setminus \tau|$.

\begin{theorem}[Theorem \ref{thm:dp-representative-set-intro} restated]
 \label{thm:dp-representative-set}
 Suppose $\sigma, \rho \subseteq \NN$ are finite or cofinite.
 Also, let $\allCost \coloneqq \max\{\cost(\sigma),\cost(\rho)\}$.
 Then, there is an algorithm $\CA$ that, given a graph $G$ and a nice tree decomposition of $G$ of width $\tw$, decides whether $G$ has a $(\sigma,\rho)$-set in time
 \[2^\tw \cdot (\allCost + 1)^{\tw(\omega + 1)} \cdot (\allCost + \tw)^{\Oh(1)} \cdot |V(G)|.\]
\end{theorem}

Before we dive into the proof, let us again compare the running times from this algorithm and the existing algorithm by van Rooij (Theorem \ref{thm:vanrooij-intro}).
To this end, we define a modified cost measure.
Write $\tau \subseteq \NN$ for a finite or cofinite set.
If $\tau$ is finite, then we define $\cost'(\tau) \coloneqq \max(\tau)$.
Otherwise, $\tau$ is cofinite and we define $\cost'(\tau) \coloneqq \max(\NN \setminus \tau) + 1$ if $\NN \setminus \tau \neq \emptyset$, and $\cost'(\NN) = 0$.
Observe the difference in the definition for cofinite sets.
Also, note that $\cost(\tau) \leq \cost'(\tau)$ for all finite or cofinite sets $\tau \subseteq \NN$.
Moreover, for a cofinite set $\tau \subseteq \NN$, we have that $\cost(\tau) = \cost'(\tau)$ if and only if $\tau = \{c,c+1,c+2,\dots\}$ for some number $c \in \NN$.

Using this cost measure, the running time of the algorithm from \cref{thm:vanrooij-intro} is
\[\Big(\cost'(\sigma) + \cost'(\rho) + 2\Big)^\tw |V(G)|^{\Oh(1)}.\]
On the other hand, the algorithm from Theorem \ref{thm:dp-representative-set} runs in time
\[\Big(2 \cdot \big(\max\{\cost(\sigma),\cost(\rho)\} + 1\big)^{\omega + 1}\Big)^{\tw} |V(G)|^{\Oh(1)}.\]
If $\sigma$ and $\rho$ are both finite, then the algorithm by van Rooij is clearly faster using that $\cost'(\sigma) = \cost(\sigma)$ and $\cost'(\rho) = \cost(\rho)$ (but, in this case, Theorem \ref{thm:main-alg-m-str} provides an improved algorithm for structured sets).
However, if one the sets $\sigma, \rho$ is cofinite,
then our algorithm may be substantially faster
since $\cost(\tau)$ can be arbitrarily smaller than $\cost'(\tau)$
for a cofinite set $\tau$.
As a concrete example, suppose that $\rho = \NN \setminus \{c\}$ and $\sigma = \NN \setminus \{d\}$.
Then, $\cost(\rho) = \cost(\sigma) = 1$, but $\cost'(\rho) = c+1$ and $\cost'(\sigma) = d+1$.
Hence, van Rooij's algorithm runs in time $(c+d+4)^{\tw}|V(G)|^{\Oh(1)}$, where the algorithm from Theorem \ref{thm:dp-representative-set} takes time $2^{\tw(\omega + 2)}|V(G)|^{\Oh(1)}$ which is at most $20.72^\tw |V(G)|^{\Oh(1)}$.
Observe that the second running time is independent of $c$ and $d$.

\begin{proof}[Proof of \Cref{thm:dp-representative-set}]
 Let $(T,\beta)$ denote the nice tree decomposition of $G$ and suppose $n = |V(G)|$.
 For ease of notation, let us set $\allStates \coloneqq \allStates_n = \{\rho_0,\dots,\rho_n,\sigma_0,\dots,\sigma_n\}$ for the remainder of this proof.
 For a node $t \in V(T)$, we denote by $X_t \coloneqq \beta(t)$ the bag of node $t$ and by $V_t$ the set of vertices contained in bags below $t$ (including $t$ itself).
 For each node $t \in V(T)$ and each $\vec{s} \in \{0,1\}^{X_t}$,
 we denote by $\widehat{L}_{t,\vec{s}} \subseteq \allStates^{X_t}$ the set of all strings $x \in \allStates^{X_t}$
 that are compatible with $(G\position{V_t},X_t)$ and satisfy $\sigvec{x} = \vec{s}$.
 To simplify notation, we write~$L_{t, \vec s} = \{\degvec{x} \mid x \in \widehat{L}_{t,\vec s}\}$ for the corresponding set of weight-vectors.

 Now, let us fix some $t \in V(T)$ and $\vec{s} \in \{0,1\}^{X_t}$.
 Also, suppose that $v \in X_t$.
 We say that $v$ is an \emph{F-position} if $\vec{s}\vposition{v} = 1$ and $\sigma$ is cofinite, or $\vec{s}\vposition{v} = 0$ and $\rho$ is cofinite.
 In the former case, we define $F_v \coloneqq \NN \setminus \sigma$, and in the latter case we define $F_v \coloneqq \NN \setminus \rho$.
 If $v$ is not an \emph{F-position}, then we say that $v$ is a \emph{P-position}.
 Note that $v$ is a P-position
 if $\vec{s}\vposition{v} = 1$ and $\sigma$ is finite, or $\vec{s}\vposition{v} = 0$ and $\rho$ is finite.
 In the former case, we define $P_v \coloneqq \sigma$, and in the latter case we define $P_v \coloneqq \rho$.
 By ordering elements in $X_t$ accordingly, we may assume that $X_t = \{v_1,\dots,v_k,v_{k+1},\dots,v_{k + \ell}\}$ such that $v_1,\dots,v_k$ are F-positions and $v_{k+1},\dots,v_{k + \ell}$ are P-positions.

 For ease of notation, we say that a set $R_{t,\vec{s}} \subseteq L_{t,\vec{s}}$ is a \emph{$(t,\vec{s})$-representative set of $L_{t,\vec{s}}$}
 if $R_{t,\vec s}$ is an $(\mathbf F; \mathbf P)$-representative set of $L_{t,\vec s}$ where $\mathbf F = (F_{v_1},\dots,F_{v_k})$ and $\mathbf P = (P_{v_{k+1}},\dots,P_{v_{k+\ell}})$.

 The algorithm computes, for each $t \in V(T)$ and $\vec{s} \in \{0,1\}^{X_t}$, a $(t,\vec{s})$-representative set $R_{t,\vec{s}}$ of $L_{t,\vec{s}}$ such that $|R_{t,\vec{s}}| \leq (\allCost+1)^{|X_t|}$.
 To compute these sets we proceed in a bottom-up fashion starting at the leaves of $T$.
 Suppose $t$ is a leaf of $T$.
 Then, $X_t = V_t = \emptyset$ and $L_{t,\vec{s}} = \{\varepsilon\}$.
 We set $R_{t,\vec{s}} \coloneqq \{\varepsilon\}$.

 Next, let $t$ denote an internal node and suppose the algorithm already computed all sets $R_{t',\vec{s}}$ for all children $t'$ of $t$.

 \begin{description}
  \item[Forget:]
   First, suppose $t$ is a forget-node and write $t'$ for the unique child of $t$.
   Also, assume that $X_{t'} = X_t \cup \{v\}$, i.e., $v \in V(G)$ is the vertex forgotten at $t$.
   Fix some $\vec{s} \in \{0,1\}^{X_t}$.
   For $i \in \{0,1\}$, we write $\vec{s}_i \in \{0,1\}^{X_t \cup \{v\}}$ for the extension of $\vec{s}$ for which $\vec{s}\vposition{v} = i$.
   Letting $j$ be the position of vertex~$v$ in~$\vec s_i$, we get
   \begin{align*}
    L_{t,\vec{s}} = (L_{t',\vec{s}_0}\ominus j)
               \cup (L_{t',\vec{s}_1}\ominus j).
   \end{align*}

   We set
   \begin{align*}
    \widehat{R}_{t,\vec{s}}
     = (R_{t',\vec{s}_0}\ominus j) \cup (R_{t',\vec{s}_1}\ominus j).
   \end{align*}
   By \cref{la:representative-set-mixed-operations}, the set $\widehat{R}_{t,\vec{s}}$ is a $(t,\vec{s})$-representative set of $L_{t,\vec{s}}$.
   Observe that $\widehat{R}_{t,\vec{s}}$ can be computed in time $\Oh((\allCost+1)^{|X_{t'}|}\cdot |X_t|) = \Oh((\allCost+1)^{\tw} \cdot (\allCost + \tw)^{\Oh(1)})$.

   Finally, we obtain $R_{t,\vec{s}}$ by computing a $(t,\vec{s})$-representative set of $\widehat{R}_{t,\vec{s}}$ using Lemma \ref{la:representative-set-mixed}.
   Note that $|R_{t,\vec{s}}| \leq (\allCost+1)^{|X_t|}$ as desired.
   Also,  $R_{t,\vec{s}}$ is a $(t,\vec{s})$-representative set of $L_{t,\vec{s}}$ since $\widehat{R}_{t,\vec{s}}$ is a $(t,\vec{s})$-representative set of $L_{t,\vec{s}}$.
   This step takes time
   \begin{align*}
       &\Oh(|\widehat{R}_{t,\vec{s}}| \cdot (\allCost+1)^{|X_t|(\omega - 1)} \cdot |X_t|)\\
    =\ &(\allCost+1)^{\tw} \cdot (\allCost+1)^{\tw(\omega - 1)} \cdot (\allCost + \tw)^{\Oh(1)}\\
    =\ &(\allCost+1)^{\tw\cdot \omega} \cdot (\allCost + \tw)^{\Oh(1)}.
   \end{align*}

   So overall, computing $R_{t,\vec{s}}$ for every $\vec{s} \in \{0,1\}^{X_t}$ takes time $2^\tw \cdot (\allCost+1)^{\tw\cdot \omega} \cdot (\allCost + \tw)^{\Oh(1)}$.
  \item[Introduce:]
   Next consider the case that $t$ is an introduce-node and write $t'$ for the unique child of $t$.
   Suppose that $X_t = X_{t'} \cup \{v\}$, that is, $v \in V(G)$ is the vertex introduced at $t$.
   Note that $N_G(v) \cap V_t \subseteq X_{t'}$.

   Now, fix some $\vec{s} \in \{0,1\}^{X_t}$.
   We define a string~$z_{\vec s} \in \NN^{X_t}$ via
   \[z_{\vec s} \position{v} \coloneqq |\{w \in N_G(v) \cap X_t \mid \vec{s}\position{w} = 1\}|\]
   and
   \[z_{\vec s} \position{w} \coloneqq \begin{cases}
                                        1 &\text{if } \vec{s}\position{v} = 1 \text{ and } w \in N_G(v),\\
                                        0 &\text{otherwise}
                                       \end{cases}\]
   for all $w \in X_{t'}$.
   Then, letting $j$ be the position of vertex~$v$ in~$\vec s$, we have
   \[L_{t,\vec{s}} = (L_{t', \vec s \position{X_{t'}}} \oplus j) + z_{\vec s}\]
   We compute
   \[\widehat R_{t,\vec s} \coloneqq (R_{t',\vec{s}\position{X_{t'}}} \oplus j) + z_{\vec s}\]

   Again, by \cref{la:representative-set-mixed-operations}, the set $\widehat{R}_{t,\vec{s}}$ is a $(t,\vec{s})$-representative set of $L_{t,\vec{s}}$.

   We obtain $R_{t,\vec{s}}$ by computing a $(t,\vec{s})$-representative set
   of $\widehat{R}_{t,\vec{s}}$ using \cref{la:representative-set-mixed}.
   Note that $|R_{t,\vec{s}}| \leq (\allCost+1)^{|X_t|}$ as desired.
   Also,  $R_{t,\vec{s}}$ is a $(t,\vec{s})$-representative set of $L_{t,\vec{s}}$ since $\widehat{R}_{t,\vec{s}}$ is a $(t,\vec{s})$-representative set of $L_{t,\vec{s}}$.
   This step takes time
   \begin{align*}
       &\Oh(|\widehat{R}_{t,\vec{s}}| \cdot (\allCost+1)^{|X_t|(\omega - 1)} \cdot |X_t|)\\
    =\ &(\allCost+1)^{\tw} \cdot (\allCost+1)^{\tw(\omega - 1)} \cdot (\allCost + \tw)^{\Oh(1)}\\
    =\ &(\allCost+1)^{\tw\cdot \omega} \cdot (\allCost + \tw)^{\Oh(1)}.
   \end{align*}

   So, in total, computing $R_{t,\vec{s}}$ for every $\vec{s} \in \{0,1\}^{X_t}$ takes time $2^\tw \cdot (\allCost+1)^{\tw\cdot \omega} \cdot (\allCost + \tw)^{\Oh(1)}$.

  \item[Join:]
   Finally, suppose that $t$ is a join-node and write $t_1,t_2$ for the two children of $t$.
   Note that $X_t = X_{t_1} = X_{t_2}$.

   Again, let us fix some $\vec{s} \in \{0,1\}^{X_t}$.
   We define the string~$z_{\vec s} \in \NN^{X_t}$ via
   \[z_{\vec s} \position{v} \coloneqq -|\{w \in N_G(v) \cap X_t \mid \vec{s}\position{w} = 1\}|.\]
   Now,
   \[L_{t,\vec{s}} = (L_{t_1,\vec{s}} + L_{t_2,\vec{s}}) + z_{\vec s}.\]
   Again, the algorithm computes
   \[\widehat{R}_{t,\vec{s}} \coloneqq  (R_{t_1,\vec s} + R_{t_2,\vec s}) + z_{\vec s}.\]
   This can be done in time $|R_{t_1,\vec{s}}| \cdot |R_{t_1,\vec{s}}| \cdot (\allCost + \tw)^{\Oh(1)} = (\allCost + 1)^{2\tw} \cdot (\allCost + \tw)^{\Oh(1)}$.
   By \cref{la:representative-set-mixed-operations}, the set $\widehat{R}_{t,\vec{s}}$ is a $(t,\vec{s})$-representative set of $L_{t,\vec{s}}$.
´
   As usual, we obtain $R_{t,\vec{s}}$ by computing a $(t,\vec{s})$-representative set of $\widehat{R}_{t,\vec{s}}$ using \cref{la:representative-set-mixed}.
   Note that $|R_{t,\vec{s}}| \leq (\allCost+1)^{|X_t|}$ as desired.
   Also, $R_{t,\vec{s}}$ is a $(t,\vec{s})$-representative set of $L_{t,\vec{s}}$ since $\widehat{R}_{t,\vec{s}}$ is a $(t,\vec{s})$-representative set of $L_{t,\vec{s}}$.
   This step takes time
   \begin{align*}
       &\Oh(|R_{t_1,\vec{s}}| \cdot |R_{t_2,\vec{s}}| \cdot (\allCost+1)^{|X_t|(\omega - 1)} \cdot |X_t|)\\
    =\ &(\allCost+1)^{2\tw} \cdot (\allCost+1)^{\tw(\omega - 1)} \cdot (\allCost + \tw)^{\Oh(1)}\\
    =\ &(\allCost+1)^{\tw\cdot(\omega + 1)} \cdot (\allCost + \tw)^{\Oh(1)}.
   \end{align*}

   So, in total, computing $R_{t,\vec{s}}$ for every $\vec{s} \in \{0,1\}^{X_t}$ takes time $2^\tw \cdot (\allCost+1)^{\tw\cdot(\omega + 1)} \cdot (\allCost + \tw)^{\Oh(1)}$.
 \end{description}
 Since processing a single node of $T$ takes time $2^\tw \cdot (\allCost + 1)^{\tw(\omega + 1)} \cdot (\allCost + \tw)^{\Oh(1)}$ and $|V(T)| = \Oh(\tw \cdot |V(G)|)$, it follows that all sets $L_{t,\vec{s}}$ can be computed in the desired time.

 To decide whether $G$ has a $(\sigma,\rho)$-set, the algorithm considers the root node $t \in V(T)$ for which $X_t = \emptyset$ and $V_t = V(G)$.
 Then, $G$ has a $(\sigma,\rho)$-set if and only if $\varepsilon \in R_{t,\vec{s}}$, where $\vec{s}$ denotes the empty vector.
\end{proof}

Similarly to the previous section, we can also obtain an algorithm for the optimization version by incorporating the size of solution sets.

\begin{theorem}
 \label{thm:dp-representative-set-opt}
 Suppose $\sigma, \rho \subseteq \NN$ are finite or cofinite.
 Also, let $\allCost \coloneqq \max\{\cost(\sigma),\cost(\rho)\}$.
 Then, there is an algorithm $\CA$ that, given a graph $G$, an integer $k$, and a nice tree decomposition of $G$ of width $\tw$, decides whether $G$ has a $(\sigma,\rho)$-set of size at most (at least) $k$ in time
 \[2^\tw \cdot (\allCost + 1)^{\tw(\omega + 1)} \cdot (\allCost + \tw)^{\Oh(1)} \cdot |V(G)|^2.\]
\end{theorem}

However, in contrast to the previous section, the representative set approach cannot be extended to the counting version of the problem.
Note that this is by design, since the fundamental idea of this approach is to not keep all the partial solutions which would be necessary for the counting version.
Actually, \cref{thm:lower-main-intro} implies that there is no algorithm
counting $(\sigma,\rho)$-sets in time $(f(\allCost))^{\tw} \cdot |V(G)|^{\Oh(1)}$
for any function $f$ assuming \#SETH.

%% file: s5-conclusion.tex
\section{Conclusion}
\label{sec:conclusion}

For every pair of finite or cofinite sets $(\sigma,\rho)$,
the present work together with the accompanying paper \cite{FockeMMNSSW23ii} determines (assuming the Counting Strong Exponential Time Hypothesis)
the best possible value $c_{\sigma,\rho}$ such that there is an algorithm
that counts $(\sigma,\rho)$-sets in time $(c_{\sigma,\rho})^\tw\cdot n^{\O(1)}$
(if a tree decomposition of width $\tw$ is given in the input).
In doing so, we obtain improved algorithms for both counting $(\sigma,\rho)$-sets, as well as deciding whether there is a $(\sigma,\rho)$-set for $\mname$-structured pairs $(\sigma,\rho)$ where $\mname \geq 2$.

For finite sets $\sigma$ and $\rho$, the lower bounds \cite{FockeMMNSSW23ii} extend to the decision version (assuming $0 \notin \rho$).
In contrast, for the decision problem with cofinite sets, we show that significant improvements are possible for certain pairs $(\sigma,\rho)$ using the technique of representative sets.
More precisely, we prove that, in this setting, the base of the running time depends only on $\cost(\sigma) + \cost(\rho)$, where $\cost(\tau)$ counts the number of elements that are missing from a cofinite set $\tau$. Thus, $\cost(\tau)$ might be much smaller than the largest missing integer from $\tau$ (which determined the value $c_{\sigma,\rho}$).

Of course, the most natural open problem is to determine the precise complexity of deciding whether a given graph of bounded treewidth has a $(\sigma,\rho)$-set (of a certain size).
However, as already pointed out above, for a tight result, one would need to overcome at least two major challenges: proving tight upper bounds on the size of representative sets, and understanding whether they can be handled without using matrix-multiplication based methods.

A much more approachable problem seems to be to obtain tight bounds for arbitrary finite
and \emph{simple} cofinite sets (that is, cofinite sets of the form $\tau = \{k,k+1,k+2,\dots\}$).
For such pairs of sets, the representative set approach does not lead to faster algorithms, and many interesting problems such as \textsc{Dominating Set} and \textsc{Independent Set} are still covered.
Note that, for such pairs, the decision problem (i.e., the problem of deciding whether there is a $(\sigma,\rho)$-set) becomes polynomial-time solvable in many cases (e.g., for \textsc{Dominating Set} and \textsc{Independent Set} the decision version is trivial).
So, in order to obtain meaningful bounds for all the relevant problems, one would also have to consider the maximization or minimization versions (i.e., given a graph, find a $(\sigma,\rho)$-set of maximal or minimal size).